\documentclass[a4paper,UKenglish,cleveref, autoref, thm-restate]{lipics-v2021}

\theoremstyle{definition}
\newtheorem{definitionn}[theorem]{Definition}
\theoremstyle{remark}

\usepackage{latexsym}
\usepackage{tikz}
\usepackage{mathdots}
\usepackage{virginialake}
\vlnostructuressyntax
\vlnosmallleftlabels
\odframefalse
\odbackgroundtrue
\odbackgroundfirstfalse

\def\ttt{\mathit{t}}
\def\fff{\mathit{f}}

\def\vltt{\mathbf{t}}
\def\vlff{\mathbf{f}}

\def\tU{\tilde{U}}

\def\cA{{\mathcal A}}

\def\cD{{\mathcal D}}

\def\cF{{\mathcal F}}

\def\cU{{\mathcal U}}

\def\SKS{\mathsf{SKS}}

\def\KS{\mathsf{KS}}

\def\K{\mathsf{K}}
\def\MLL{\mathsf{MLL}}

\def\set#1{\{#1\}}

\def\cons#1{\{#1\}}
\def\conhole      {\cons{\enspace}}%

\def\tuple#1{\langle#1\rangle}

\def\grammareq {\mathrel{\raise.4pt\hbox{::}{=}}}%

\newcommand{\dotto}[1][]{\mathrel{\!\xy\ar@{.>}^-{#1}(5,0)\endxy\!}}
\newcommand{\solto}[1][]{\mathrel{\!\xy\ar@{->}^-{#1}(5,0)\endxy\!}}
\newcommand{\longsolto}[1][]{\mathrel{\!\xy\ar@{->}^-{#1}(11,0)\endxy\!}}
\newcommand{\longdotto}[1][]{\mathrel{\!\xy\ar@{.>}^-{#1}(11,0)\endxy\!}}
\newcommand{\xldotto}[2][]{\mathrel{\!\xy\ar@{.>}^-{#1}(#2,0)\endxy\!}}

\def\sqn  #1{{\;\turnstile #1\;}}%
\def\ssqn#1#2{{\;#1\turnstile #2\;}}%

\def\cutr{\mathsf{cut}}

\def\weakr{\mathsf{weak}}
\def\conr{\mathsf{cont}}

\def\rr{\mathsf{r}}

\def\MLL{\mathsf{MLL}}

\def\cneg#1{\bar{#1}}

\def\wcneg#1{\overline{#1}}

\newbox\cutbox
\newdimen\cutwd
\newdimen\cutht
\newdimen\cutdp
\def\ccut{%
  \setbox\cutbox\hbox{$\lozenge$}
  \cutwd=\wd\cutbox
  \cutht=\ht\cutbox
  \cutdp=\dp\cutbox
  \setbox\cutbox\hbox to\cutwd{\hss\vrule width.3pt height\cutht depth\cutdp\hss}
  \mathbin{\lozenge\hskip-\cutwd\copy\cutbox}}
\def\scriptcut{%
  \setbox\cutbox\hbox{$\scriptstyle\lozenge$}
  \cutwd=\wd\cutbox
  \cutht=\ht\cutbox
  \cutdp=\dp\cutbox
  \setbox\cutbox\hbox to\cutwd{\hss\vrule width.3pt height\cutht depth\cutdp\hss}
  \mathord{\lozenge\hskip-\cutwd\copy\cutbox}}

\def\vccut{%
  \setbox\cutbox\hbox{$\lozenge$}
  \cutwd=\wd\cutbox
  \cutht=\ht\cutbox
  \cutdp=\dp\cutbox
  \setbox\cutbox\hbox to\cutwd{\hss\hskip.3pt\vrule width.3pt height\cutht depth\cutdp\hss}
  \mathbin{\lozenge\hskip-\cutwd\copy\cutbox}}

\def\lrgldel {\mathchoice{(}{(}{\langle}{\langle}}%
\def\lrgrdel {\mathchoice{)}{)}{\rangle}{\rangle}}%
\def\aprldel {\mathchoice
    {\mathopen {\setbox0=\hbox{$\displaystyle     \lrgldel$}\hbox to\wd0
                         {\hfil$\displaystyle     (       $\hfil}}}%
    {\mathopen {\setbox0=\hbox{$\textstyle        \lrgldel$}\hbox to\wd0
                         {\hfil$\textstyle        (        $\hfil}}}%
    {\mathopen {\setbox0=\hbox{$\scriptstyle      \lrgldel$}\hbox to\wd0
                         {\hfil$\scriptstyle      (        $\hfil}}}%
    {\mathopen {\setbox0=\hbox{$\scriptscriptstyle\lrgldel$}\hbox to\wd0
                         {\hfil$\scriptscriptstyle(        $\hfil}}}}%
\def\aprrdel {\mathchoice
    {\mathclose{\setbox0=\hbox{$\displaystyle     \lrgrdel$}\hbox to\wd0
                         {\hfil$\displaystyle     )       $\hfil}}}%
    {\mathclose{\setbox0=\hbox{$\textstyle        \lrgrdel$}\hbox to\wd0
                         {\hfil$\textstyle        )        $\hfil}}}%
    {\mathclose{\setbox0=\hbox{$\scriptstyle      \lrgrdel$}\hbox to\wd0
                         {\hfil$\scriptstyle      )        $\hfil}}}%
    {\mathclose{\setbox0=\hbox{$\scriptscriptstyle\lrgrdel$}\hbox to\wd0
                         {\hfil$\scriptscriptstyle)        $\hfil}}}}%
\def\seqldel {\mathchoice
    {\mathopen {\setbox0=\hbox{$\displaystyle     \lrgldel$}\hbox to\wd0
                         {\hfil$\displaystyle     \langle  $\hfil}}}%
    {\mathopen {\setbox0=\hbox{$\textstyle        \lrgldel$}\hbox to\wd0
                         {\hfil$\textstyle        \langle  $\hfil}}}%
    {\mathopen {\setbox0=\hbox{$\scriptstyle      \lrgldel$}\hbox to\wd0
                         {\hfil$\scriptstyle      \langle  $\hfil}}}%
    {\mathopen {\setbox0=\hbox{$\scriptscriptstyle\lrgldel$}\hbox to\wd0
                         {\hfil$\scriptscriptstyle\langle  $\hfil}}}}%
\def\seqrdel {\mathchoice
    {\mathclose{\setbox0=\hbox{$\displaystyle     \lrgrdel$}\hbox to\wd0
                         {\hfil$\displaystyle     \rangle  $\hfil}}}%
    {\mathclose{\setbox0=\hbox{$\textstyle        \lrgrdel$}\hbox to\wd0
                         {\hfil$\textstyle        \rangle  $\hfil}}}%
    {\mathclose{\setbox0=\hbox{$\scriptstyle      \lrgrdel$}\hbox to\wd0
                         {\hfil$\scriptstyle      \rangle  $\hfil}}}%
    {\mathclose{\setbox0=\hbox{$\scriptscriptstyle\lrgrdel$}\hbox to\wd0
                         {\hfil$\scriptscriptstyle\rangle  $\hfil}}}}%
\def\parldel {\mathchoice
    {\mathopen {\setbox0=\hbox{$\displaystyle     \lrgldel$}\hbox to\wd0
                         {\hfil$\displaystyle     [       $\hfil}}}%
    {\mathopen {\setbox0=\hbox{$\textstyle        \lrgldel$}\hbox to\wd0
                         {\hfil$\textstyle        [        $\hfil}}}%
    {\mathopen {\setbox0=\hbox{$\scriptstyle      \lrgldel$}\hbox to\wd0
                         {\hfil$\scriptstyle      [        $\hfil}}}%
    {\mathopen {\setbox0=\hbox{$\scriptscriptstyle\lrgldel$}\hbox to\wd0
                         {\hfil$\scriptscriptstyle[        $\hfil}}}}%
\def\parrdel {\mathchoice
    {\mathclose{\setbox0=\hbox{$\displaystyle     \lrgrdel$}\hbox to\wd0
                         {\hfil$\displaystyle     ]       $\hfil}}}%
    {\mathclose{\setbox0=\hbox{$\textstyle        \lrgrdel$}\hbox to\wd0
                         {\hfil$\textstyle        ]        $\hfil}}}%
    {\mathclose{\setbox0=\hbox{$\scriptstyle      \lrgrdel$}\hbox to\wd0
                         {\hfil$\scriptstyle      ]        $\hfil}}}%
    {\mathclose{\setbox0=\hbox{$\scriptscriptstyle\lrgrdel$}\hbox to\wd0
                         {\hfil$\scriptscriptstyle]        $\hfil}}}}%

\def\eightpoint{\small}                         
\def\pluldel {\mathchoice
   {\mathopen {\setbox0=\hbox{$\displaystyle     \lrgldel$}\hbox to\wd0
                        {\hfil$\displaystyle     [       $\hfil}%
                        \kern-\wd0\hbox to\wd0
                        {\hss$\vcenter{\hbox{\eightpoint$\scriptscriptstyle\bullet$}}$\hss}}}%
   {\mathopen {\setbox0=\hbox{$\textstyle        \lrgldel$}\hbox to\wd0
                        {\hfil$\textstyle        [       $\hfil}%
                        \kern-\wd0\hbox to\wd0
                        {\hss$\vcenter{\hbox{\eightpoint$\scriptscriptstyle\bullet$}}$\hss}}}%
   {\mathopen {\setbox0=\hbox{$\scriptstyle      \lrgldel$}\hbox to\wd0
                        {\hfil$\scriptstyle      [       $\hfil}%
                        \kern-\wd0\hbox to\wd0
                        {\hss$\raise.1ex\hbox{\eightpoint$\scriptscriptstyle\bullet$}$\hss}}}%
   {\mathopen {\setbox0=\hbox{$\scriptscriptstyle\lrgldel$}\hbox to\wd0
                        {\hfil$\scriptscriptstyle[       $\hfil}%
                        \kern-\wd0\hbox to\wd0
                        {\hss$\raise.03ex\hbox{\eightpoint$\scriptscriptstyle\bullet$}$\hss}}}}%
\def\plurdel {\mathchoice
   {\mathclose{\setbox0=\hbox{$\displaystyle     \lrgldel$}\hbox to\wd0
                        {\hfil$\displaystyle     ]       $\hfil}%
                        \kern-\wd0\hbox to\wd0
                        {\hss$\vcenter{\hbox{\eightpoint$\scriptscriptstyle\bullet$}}$\hss}}}%
   {\mathclose{\setbox0=\hbox{$\textstyle        \lrgldel$}\hbox to\wd0
                        {\hfil$\textstyle        ]       $\hfil}%
                        \kern-\wd0\hbox to\wd0
                        {\hss$\vcenter{\hbox{\eightpoint$\scriptscriptstyle\bullet$}}$\hss}}}%
   {\mathclose{\setbox0=\hbox{$\scriptstyle      \lrgldel$}\hbox to\wd0
                        {\hfil$\scriptstyle      ]       $\hfil}%
                        \kern-\wd0\hbox to\wd0
                        {\hss$\raise.1ex\hbox{\eightpoint$\scriptscriptstyle\bullet$}$\hss}}}%
   {\mathclose{\setbox0=\hbox{$\scriptscriptstyle\lrgldel$}\hbox to\wd0
                        {\hfil$\scriptscriptstyle]       $\hfil}%
                        \kern-\wd0\hbox to\wd0
                        {\hss$\raise.03ex\hbox{\eightpoint$\scriptscriptstyle\bullet$}$\hss}}}}%
\def\witldel {\mathchoice
   {\mathopen {\setbox0=\hbox{$\displaystyle     \lrgldel$}\hbox to\wd0
                        {\hfil$\displaystyle     (       $\hfil}%
                        \kern-\wd0\hbox to\wd0
                        {\hss$\vcenter{\hbox{\eightpoint$\scriptscriptstyle\bullet\mkern3.2mu$}}$\hss}}}%
   {\mathopen {\setbox0=\hbox{$\textstyle        \lrgldel$}\hbox to\wd0
                        {\hfil$\textstyle        (       $\hfil}%
                        \kern-\wd0\hbox to\wd0
                        {\hss$\vcenter{\hbox{\eightpoint$\scriptscriptstyle\bullet\mkern3.2mu$}}$\hss}}}%
   {\mathopen {\setbox0=\hbox{$\scriptstyle      \lrgldel$}\hbox to\wd0
                        {\hfil$\scriptstyle      (       $\hfil}%
                        \kern-\wd0\hbox to\wd0
                        {\hss$\raise.1ex\hbox{\eightpoint$\scriptscriptstyle\bullet\mkern3.2mu$}$\hss}}}%
   {\mathopen {\setbox0=\hbox{$\scriptscriptstyle\lrgldel$}\hbox to\wd0
                        {\hfil$\scriptscriptstyle(       $\hfil}%
                        \kern-\wd0\hbox to\wd0
                        {\hss$\raise.03ex\hbox{\eightpoint$\scriptscriptstyle\bullet\mkern3.2mu$}$\hss}}}}%
\def\witrdel {\mathchoice
   {\mathclose{\setbox0=\hbox{$\displaystyle     \lrgldel$}\hbox to\wd0
                        {\hfil$\displaystyle     )       $\hfil}%
                        \kern-\wd0\hbox to\wd0
                        {\hss$\vcenter{\hbox{\eightpoint$\scriptscriptstyle\mkern3.2mu\bullet$}}$\hss}}}%
   {\mathclose{\setbox0=\hbox{$\textstyle        \lrgldel$}\hbox to\wd0
                        {\hfil$\textstyle        )       $\hfil}%
                        \kern-\wd0\hbox to\wd0
                        {\hss$\vcenter{\hbox{\eightpoint$\scriptscriptstyle\mkern3.2mu\bullet$}}$\hss}}}%
   {\mathclose{\setbox0=\hbox{$\scriptstyle      \lrgldel$}\hbox to\wd0
                        {\hfil$\scriptstyle      )       $\hfil}%
                        \kern-\wd0\hbox to\wd0
                        {\hss$\raise.1ex\hbox{\eightpoint$\scriptscriptstyle\mkern3.2mu\bullet$}$\hss}}}%
   {\mathclose{\setbox0=\hbox{$\scriptscriptstyle\lrgldel$}\hbox to\wd0
                        {\hfil$\scriptscriptstyle)       $\hfil}%
                        \kern-\wd0\hbox to\wd0
                        {\hss$\raise.03ex\hbox{\eightpoint$\scriptscriptstyle\mkern3.2mu\bullet$}$\hss}}}}%

\newbox\ldelbox
\setbox\ldelbox=\hbox{$\lrgldel$}

\newbox\rdelbox
\setbox\rdelbox=\hbox{$\lrgrdel$}

\def\cons #1{\conldel #1\conrdel}%
\let\turnstile=\vdash
\def\lone{1}

\def\loc{\mathord{!}}
\def\lneg{^\bot}

\def\quadfs {\rlap{\rm\quad.}}%
\def\quand {\quad\mbox{and}\quad}%
\def\qquand {\qquad\mbox{and}\qquad}%
\def\quor {\quad\mbox{or}\quad}%
\def\qquor {\qquad\mbox{or}\qquad}%
\def\qquato {\qquad\to\qquad}%
\def\qualto {\quad\leadsto\quad}%
\def\qqualto {\qquad\leadsto\qquad}%

\def\quand {\quad\mbox{and}\quad}%
\def\qquand {\qquad\mbox{and}\qquad}%

\def\quor {\quad\mbox{or}\quad}%
\def\qquor {\qquad\mbox{or}\qquad}%
\def\clap#1{\hbox to 0pt{\hss#1\hss}}
\def\sqlap#1{\hbox to .5em{\hss#1\hss}}
\def\qlap#1{\hbox to 1em{\hss#1\hss}}
\def\qqlap#1{\hbox to 2em{\hss#1\hss}}
\def\qqqlap#1{\hbox to 3em{\hss#1\hss}}
\def\qqqqlap#1{\hbox to 4em{\hss#1\hss}}
\def\qqqqqlap#1{\hbox to 5em{\hss#1\hss}}
\def\qqqqqqlap#1{\hbox to 6em{\hss#1\hss}}
\def\qqqqqqqlap#1{\hbox to 7em{\hss#1\hss}}
\def\qqqqqqqqlap#1{\hbox to 8em{\hss#1\hss}}
\def\qqqqqqqqqlap#1{\hbox to 9em{\hss#1\hss}}
\newcommand{\wlap}[2][10ex]{\hbox to#1{\hss#2\hss}}

\newcommand{\wlapm}[2][10ex]{\hbox to#1{\hss$#2$\hss}}
\def\rlapm#1{\hbox to 0pt{$#1$\hss}}
\def\llapm#1{\hbox to 0pt{\hss$#1$}}

\newcommand{\vclap}[2][0pt]{\hbox to #1{\hss#2\hss}}
\newcommand{\vclapm}[2][0pt]{\hbox to #1{\hss$#2$\hss}}

\def\interdisplayskip{.5ex}
\newskip\mydisplaywidth
\newcommand{\twolinedisplay}[3][10pt]{%
  \mydisplaywidth=\displaywidth
  \advance\mydisplaywidth-#1
  \begin{array}{c}
    \clap{\hbox to\mydisplaywidth{$\displaystyle#2$\hss}}\\[\interdisplayskip]
    \clap{\hbox to\mydisplaywidth{\hss$\displaystyle#3$}}
  \end{array}
}

\def\nosmash#1{#1}

\newcommand{\vlstr}[3]{\vltr{#1}{#2}{\vlhy{}}{#3}{\vlhy{}}}
\newcommand{\vlhtr}[2]{\vlstr{#1}{#2}{\vlhy{\hskip1em}}}

\newcommand{\vlone}{1} 
\newcommand{\vlbot}{\bot}
\newcommand{\vltop}{\top}
\newcommand{\vlzer}{0}

\newcommand{\vlbox}{\mathord\boxempty} 
\newcommand{\vldia}{\mathord\Diamond}
\newcommand{\vlwn}{\mathord?}
\newcommand{\vloc}{\mathord!}
\newcommand{\vlpl}{\vlbin\varoplus}%
\newcommand{\vlwi}{\vlbin\binampersand}%
\newcommand{\vlva}{\vlbin\varoast}%
\newcommand{\vlwa}{\vlbin\varodot}%
\newcommand{\vlda}{\mathord\dagger}%
\newcommand{\vlca}{\mathord\triangledown}%

\newcommand{\aiD}{\mathsf{ai}\mathord\downarrow}
\newcommand{\aiU}{\mathsf{ai}\mathord\uparrow}
\newcommand{\topD}{\mathsf{\top}\mathord\downarrow}
\newcommand{\topU}{\mathsf{\top}\mathord\uparrow}
\newcommand{\tensD}{\mathsf{\vlte}\mathord\downarrow}
\newcommand{\tensU}{\mathsf{\vlte}\mathord\uparrow}
\newcommand{\withD}{\mathsf{\vlwi}\mathord\downarrow}
\newcommand{\withU}{\mathsf{\vlwi}\mathord\uparrow}
\newcommand{\andD}{\mathsf{\vlan}\mathord\downarrow}
\newcommand{\andU}{\mathsf{\vlan}\mathord\uparrow}
\newcommand{\eD}{\mathsf{e}\mathord\downarrow}
\newcommand{\eU}{\mathsf{e}\mathord\uparrow}
\newcommand{\gD}{\mathsf{4}\mathord\downarrow}
\newcommand{\gU}{\mathsf{4}\mathord\uparrow}
\newcommand{\dD}{\mathsf{d}\mathord\downarrow}
\newcommand{\dU}{\mathsf{d}\mathord\uparrow}
\newcommand{\dzD}{\mathsf{d_0}\mathord\downarrow}
\newcommand{\dzU}{\mathsf{d_0}\mathord\uparrow}
\newcommand{\dzp}{\mathsf{d'_0}}
\newcommand{\dzpD}{\mathsf{d'_0}\mathord\downarrow}
\newcommand{\dzpU}{\mathsf{d'_0}\mathord\uparrow}
\newcommand{\sD}{\mathsf{s}\mathord\downarrow}
\newcommand{\sU}{\mathsf{s}\mathord\uparrow}
\newcommand{\dpD}{\mathsf{d'}\mathord\downarrow}
\newcommand{\dhD}{\mathsf{\hat{d}}\mathord\downarrow}
\newcommand{\dpU}{\mathsf{d'}\mathord\uparrow}
\renewcommand{\cD}{\mathsf{c}\mathord{\downarrow}}

\renewcommand{\cU}{\mathsf{c}\mathord\uparrow}
\newcommand{\cpD}{\mathring{\mathsf{c}}\mathord{\downarrow}}
\newcommand{\cpU}{\mathring{\mathsf{c}}\mathord\uparrow}
\newcommand{\wpD}{\mathring{\mathsf{w}}\mathord{\downarrow}}
\newcommand{\wpU}{\mathring{\mathsf{w}}\mathord\uparrow}
\newcommand{\wlD}{\acute{\mathsf{w}}\mathord{\downarrow}}
\newcommand{\wlU}{\acute{\mathsf{w}}\mathord\uparrow}
\newcommand{\wrD}{\grave{\mathsf{w}}\mathord{\downarrow}}
\newcommand{\wrU}{\grave{\mathsf{w}}\mathord\uparrow}
\newcommand{\pD}{\mathsf{p}\mathord\downarrow}
\newcommand{\pU}{\mathsf{p}\mathord\uparrow}
\newcommand{\bD}{\mathsf{b}\mathord\downarrow}

\newcommand{\rD}{\mathsf{r}\mathord\downarrow}
\newcommand{\rU}{\mathsf{r}\mathord\uparrow}
\newcommand{\ppD}{\mathsf{p'}\mathord\downarrow}
\newcommand{\ppU}{\mathsf{p'}\mathord\uparrow}
\newcommand{\phU}{\mathsf{\hat{p}}\mathord\uparrow}
\newcommand{\tD}{\mathsf{t}\mathord\downarrow}
\renewcommand{\tU}{\mathsf{t}\mathord\uparrow}
\newcommand{\swl}{\mathsf{sl}}
\newcommand{\swr}{\mathsf{sr}}

\newcommand{\axr}{\mathsf{ax}}
\newcommand{\kr}{\mathsf{k}}
\newcommand{\ktr}{\mathsf{t}}
\newcommand{\kfr}{\mathsf{4}}
\newcommand{\rcD}{\mathsf{r_c}\mathord\downarrow}
\newcommand{\rncD}{\mathsf{r_{nc}}\mathord\downarrow}

\newcommand{\eqD}{\mathord{\fequ}\mathord\downarrow}
\newcommand{\eqU}{\mathord{\fequ}\mathord\uparrow}

\usepackage{xspace}
\newcommand{\resp}[1]{(resp.\ #1)\xspace}
\newcommand{\Resp}[1]{$\left(\text{resp.\ #1}\right)$}

\newcommand{\proves}[1][]{\vdash_{#1}}

\newcommand{\nodv}[4]{{#1}\vdash^{#2}_{#4}{#3}}
\newcommand{\nnodv}[4]{{#1}\vdash_{#4}{#3}}

\newcommand{\transl}{\leftrightarrow}

\title{Interpolation via Generalized Splitting} 

\author{Lutz {Stra{\ss}burger}}{LIX, Ecole Polytechnique, Institut Polytechnique de Paris, France\\Inria Saclay, France \and \url{https://www.lix.polytechnique.fr/Labo/Lutz.Strassburger/} }{}{https://orcid.org/0000-0003-4661-6540}{}

\authorrunning{L.~Stra{\ss}burger} 

\Copyright{Lutz Stra\ss burger} 



\category{} 

\acknowledgements{} 

\nolinenumbers 

\EventEditors{Corina C\^{\i}rstea and Yannick Forster}
\EventNoEds{2}
\EventLongTitle{35th EACSL Annual Conference on Computer Science Logic (CSL 2027)}
\EventShortTitle{CSL 2027}
\EventAcronym{CSL}
\EventYear{2027}
\EventDate{January 26--29, 2027}
\EventLocation{Brighton, UK}
\EventLogo{}
\SeriesVolume{403}
\ArticleNo{18}

\newcommand{\prem}[1]{\mathsf{pr}(#1)}
\newcommand{\conc}[1]{\mathsf{cn}(#1)}
\newcommand{\premf}{\mathsf{pr}}
\newcommand{\concf}{\mathsf{cn}}

\newcommand{\Fms}{\cF}
\newcommand{\Atms}{\cA}
\newcommand{\Ctxs}{\mathcal{C}}
\newcommand{\Ders}{\mathcal{D}}

\newcommand{\Jast}{J^\ast}

\newcommand{\deri}{\delta}
\newcommand{\derib}{\theta}
\newcommand{\deric}{\zeta}
\newcommand{\derid}{\xi}
\newcommand{\derisw}{\delta_{\mathsf{s}}}
\newcommand{\sys}{\mathsf{X}}
\newcommand{\sysS}{\mathsf{X}}
\newcommand{\sysX}{\mathsf{X}}
\newcommand{\sysY}{\mathsf{Y}}
\newcommand{\sysSX}{\mathsf{SX}}
\newcommand{\sysSXD}{\mathsf{SX}\mathord{\downarrow}}

\newcommand{\sysU}{\mathsf{X}\mathord{\uparrow}}
\newcommand{\sysD}{\mathsf{X}\mathord{\downarrow}}

\newcommand{\SLSp}{\mathsf{SLS'}}
\newcommand{\LS}{\mathsf{LS}}
\newcommand{\MLS}{\mathsf{MLS}}
\newcommand{\LSc}{\mathsf{LS_c}}
\newcommand{\LSnc}{\mathsf{LS_{nc}}}
\newcommand{\LSp}{\mathsf{LS'}}
\newcommand{\LScp}{\mathsf{LS_c'}}
\newcommand{\SLS}{\mathsf{SLS}}

\newcommand{\SMLS}{\mathsf{SMLS}}
\newcommand{\SLSc}{\mathsf{SLS_c}}
\newcommand{\SLSnc}{\mathsf{SLS_{nc}}}

\newcommand{\SLScU}{\mathsf{SLS_c}\mathord{\uparrow}}
\newcommand{\SLSncD}{\mathsf{SLS_{nc}}\mathord{\downarrow}}
\newcommand{\SLSncU}{\mathsf{SLS_{nc}}\mathord{\uparrow}}
\newcommand{\SLSpD}{\mathsf{SLS'}\mathord{\downarrow}}
\newcommand{\SLSpU}{\mathsf{SLS'}\mathord{\uparrow}}
\newcommand{\SLScpD}{\mathsf{SLS_c'}\mathord{\downarrow}}
\newcommand{\SLScpU}{\mathsf{SLS_c'}\mathord{\uparrow}}
\newcommand{\SLScp}{\mathsf{SLS'_c}}
\newcommand{\SKSc}{\mathsf{SKS_c}}
\newcommand{\SKSnc}{\mathsf{SKS_{nc}}}
\newcommand{\KSc}{\mathsf{KS_c}}
\newcommand{\KSnc}{\mathsf{KS_{nc}}}
\newcommand{\SKScU}{\mathsf{SKS_c}\mathord{\uparrow}}

\newcommand{\SKSU}{\mathsf{SKS}\mathord{\uparrow}}
\newcommand{\SKSD}{\mathsf{SKS}\mathord{\downarrow}}
\newcommand{\KSkc}{\mathsf{KS^K_{c}}}
\newcommand{\KSync}{\mathsf{KS^Y_{nc}}}
\newcommand{\KSknc}{\mathsf{KS^K_{nc}}}
\newcommand{\KSktnc}{\mathsf{KS^{KT}_{nc}}}
\newcommand{\KSkfnc}{\mathsf{KS^{K4}_{nc}}}
\newcommand{\KSsfnc}{\mathsf{KS^{S4}_{nc}}}
\newcommand{\SKSkc}{\mathsf{SKS^K_{c}}}
\newcommand{\SKSkcU}{\mathsf{SKS^K_{c}}\mathord{\uparrow}}
\newcommand{\SKSync}{\mathsf{SKS^Y_{nc}}}
\newcommand{\SKSknc}{\mathsf{SKS^K_{nc}}}
\newcommand{\SKSktnc}{\mathsf{SKS^{KT}_{nc}}}
\newcommand{\SKSkfnc}{\mathsf{SKS^{K4}_{nc}}}
\newcommand{\SKSsfnc}{\mathsf{SKS^{S4}_{nc}}}
\newcommand{\KSy}{\mathsf{KS^Y}}
\newcommand{\KSk}{\mathsf{KS^K}}

\newcommand{\SKSy}{\mathsf{SKS^Y}}
\newcommand{\SKSyD}{\mathsf{SKS^Y}\mathord{\downarrow}}
\newcommand{\SKSyU}{\mathsf{SKS^Y}\mathord{\uparrow}}
\newcommand{\SKSk}{\mathsf{SKS^K}}
\newcommand{\SKSkt}{\mathsf{SKS^{KT}}}
\newcommand{\SKSkf}{\mathsf{SKS^{K4}}}
\newcommand{\SKSsf}{\mathsf{SKS^{S4}}}
\newcommand{\MELL}{\mathsf{MELL}}
\newcommand{\KT}{\mathsf{KT}}
\newcommand{\Kfour}{\mathsf{K4}}
\newcommand{\Kfourthree}{\mathsf{K4.3}}
\newcommand{\Sfour}{\mathsf{S4}}
\newcommand{\GSone}{\mathsf{GS1}}
\newcommand{\GSk}{\mathsf{GS1^K}}
\newcommand{\GSkt}{\mathsf{GS1^{KT}}}
\newcommand{\GSkf}{\mathsf{GS1^{K4}}}
\newcommand{\GSsf}{\mathsf{GS1^{S4}}}
\newcommand{\GSy}{\mathsf{GS1^{Y}}}

\renewcommand{\lneg}[1]{#1^\bot}
\newcommand{\llneg}[1]{#1^{\bot\bot}}
\renewcommand{\cons}[1]{[#1]}
\newcommand{\Cons}[1]{\left[{#1}\right]}
\renewcommand{\conhole}{\cons{\cdot}}
\newcommand{\Ccons}[2][]{C_{#1}\cons{#2}}
\newcommand{\Cpcons}[2][]{C'_{#1}\cons{#2}}
\newcommand{\CCons}[2][]{C_{#1}\left[{#2}\right]}

\newcommand{\Chole}[1][\cdot]{\Ccons{#1}}
\newcommand{\CHole}[1][\cdot]{\CCons{#1}}
\newcommand{\chole}{\cdot}

\newcommand{\ccneg}[1]{\cneg{\cneg #1}}
\newcommand{\sizeof}[1]{\left|#1\right|}

\newcommand{\stto}[1]{\stackrel{#1}{\to}}
\newcommand{\qsto}[1]{\quad\stto{#1}\quad}

\newcommand{\fequ}{\equiv}
\newcommand{\eqpa}{\mathbin{\fequ_{\vlpa}}}
\newcommand{\eqte}{\mathbin{\fequ_{\vlte}}}
\newcommand{\eqor}{\mathbin{\fequ_{\vlor}}}
\newcommand{\eqan}{\mathbin{\fequ_{\vlan}}}

\newcommand{\proviso}[1]{\mbox{\small(#1)}}

\newcommand{\defn}[1]{\textbf{\itshape #1}}
\newcommand{\nodf}[1]{#1}

\begin{document}

\maketitle

\begin{abstract}
  We propose a new proof theoretical method for proving Lyndon interpolation. Our proof does not use the sequent calculus but is based on a generalization of the splitting lemma in deep inference.  We then formulate the interpolation theorem as a decomposition of a derivation into an up-fragment and a down-fragment. This can be seen as (i) a strengthening of the standard formulation of the interpolation theorem, and (ii) a generalization of the cut elimination theorem.  We demonstrate the flexibility of our approach by applying it to linear logic, classical logic, and modal logics. For this, we also introduce novel cut-free proof systems for several modal logics in deep inference.
\end{abstract}

\section{Introduction}

Given an implication $A\to B$, the interpolation theorem by Craig~\cite{craig:57:interpol} says that there is a formula $C$, called the \emph{interpolant}, such that the two implications $A\to C$ and $C\to B$ hold, and furthermore, the language of $C$ is  contained in the intersection of the languages of $A$ and $B$. A stronger form of this result is Lyndon interpolation~\cite{lyndon:59}, which also demands that the propositional symbols appearing in $C$ must have the same polarity as in $A$ and $B$, i.e., if they occur negated in $C$, then they also must occur negated in $A$ and in $B$.

Interpolation has been studied for various logics, including classical first-order logic~\cite{craig:57:interpol,lyndon:59}, intuitionistic and intermediate logics~\cite{kuznets:lellmann:21}, various modal logics~\cite{fitting:kuznets:15}, and also linear logic~\cite{roorda:94,saurin:25,dissvonlutz}. And there are now many applications of interpolations in various areas of computer science, for example in program verification~\cite{rummer:26}, database query rewriting~\cite{benedikt:26}, or knowledge representation~\cite{JKK:25:interpol}.

There are many different ways of proving interpolation, using model-theoretic or automata-theoretic or syntactic or proof theoretic methods~\cite{BCI:interpolation}. The most common proof-theoretic proof of the interpolation theorem is due to  Maehara~\cite{maehara:interpolation}. That proof surgically extracts the interpolant from a cut-free sequent proof of $\ssqn AB$. Many modern proofs for richer logics are variations of  Maehara's method (see, e.g., \cite{GJK:interpolation,kuznets:lellmann:21,fitting:kuznets:15}).

Our first contribution in this paper is a new method for proving the interpolation theorem. It will be proof theoretic, but instead of the sequent calculus, it will use techniques borrowed from deep inference~\cite{gug:str:01,brunnler:tiu:01,gug:SIS,dissvonlutz}. 

Recently, Saurin~\cite{saurin:25} strengthened Maehara's method to obtain a proof relevant cut-introduction. I.e., given a cut-free proof $\pi$ of $\ssqn AB$, the interpolant $C$ comes with two proofs $\pi_1$ of $\ssqn AC$ and $\pi_2$ of  $\ssqn CB$, such that eliminating the cut from\vadjust{\vskip-3.5ex}
$$
\scalebox{.95}{$
\vlderivation{
  \vliin{\cutr}{}{\ssqn AB}{
    \vlhtr{\pi_1}{\ssqn AC}}{
    \vlhtr{\pi_2}{\ssqn CB}}}
$}
\vadjust{\vskip-2.5ex}
$$
yields exactly $\pi$.

This leads us to the second contribution of this paper, which is a theorem that (i) can be seen as a first step towards a deep inference formulation of Saurin's result, and at the same time (ii) is a simple statement that generalizes both, cut elimination and interpolation. More precisely, our main theorem is the following:
\begin{theorem}\label{thm:interpol}
  Every derivation $\odv A\deri B\sys$ 
  can be decomposed into
  $\smash{
    \od{\odd{\odd{\odh{A}}{
            \deri_1}{I}{\sysU}}{
        \deri_2}{B}{\sysD}}
  }$  
  for some formula~$I$. 
\end{theorem}
Here, $\odv A\deri B\sys$ 
stands for a derivation $\deri$ with premise $A$ and conclusion $B$, in which all inference rules come from the proof system $\sys$, and we discuss in the course of this paper, for which proof systems $\sys$, we prove this theorem.
Note that in a deep inference system, every inference rule is single premise.
Furthermore, a typical deep inference proof system $\sys$ is divided into a \emph{down-fragment}, here denoted by $\sysD$, and an \emph{up-fragment}, denoted by $\sysU$. The rules in the down-fragment correspond to the rules in a cut-free sequent calculus, and the rules in the up-fragment correspond to the $\cutr$ in the sequent calculus. However, in a deep inference system, there is a perfect up-down duality between these two fragments. And \Cref{thm:interpol} above says that they can be separated in a derivation.

Since the rules in the down-fragment cannot create new propositional variables (or change their polarity) while going up in a derivation, and dually, the up-rules cannot  create new propositional variables (or change their polarity) while going down in a derivation, it follows immediately, that the formula $I$\/ in \Cref{thm:interpol} must be a Lyndon interpolant.
Furthermore, if the formula $A$ in \Cref{thm:interpol} is \emph{truth}\footnote{I.e., the logical constant that starts a proof. This would be $\vltt$ in classical logic and $\vlone$ in linear logic.}, then $I$ also must be \emph{truth}, which gives us a cut-free proof of~$B$. Hence, \Cref{thm:interpol} also states cut elimination.

To demonstrate the flexibility of our new approach to interpolation, we show it here for linear logic, classical logic, and several modal logics.

This brings us to the third contribution in this paper. We introduce novel cut-free deep inference proof systems for the modal logics $\K$, $\KT$, $\Kfour$ and $\Sfour$. The reason for this is that the deep inference proof theory for modal logic is rather underdeveloped, and the existing proof systems~\cite{hein:stewart:05,stewart:stouppa:05} are not suited for our purposes. In fact, it is not even clear whether they enjoy cut elimination.

Let us now briefly discuss how our proof of \Cref{thm:interpol} works.
The main ingredient is a generalization of the so-called \emph{splitting lemma}~\cite{gug:SIS}.
Using the notation of linear logic, the standard variant of this lemma starts from a proof of a formula $(A\vlte B)\vlpa K$ and splits the context $K$ into $K_A\vlpa K_B$ such that $A\vlpa K_A$ and $B\vlpa K_B$ are both provable.
In our generalization, the derivation of $(A\vlte B)\vlpa K$ can have an arbitrary formula $J$ as premise, instead of the unit $\vlone$. With this generalized version of splitting, we can then prove a so-called \emph{flipping lemma} which allows us to transform a derivation
\begin{equation}
  \label{eq:flip}
  \hfill
  \od{\odd{\odd{\odh{H}}{
        \deri_1}{J}{\sysU}}{
      \deri_2}{\lneg A\vlpa K}{\sysD}}
  \qquad\text{into}\qquad
  \od{\odd{\odd{\odh{A\vlte H}}{
        \derib_1}{I}{\sysU}}{
      \derib_2}{K}{\sysD}}
  \quad.
  \hfill
\end{equation}
for some formula $I$. In other words, the decomposition into up-fragment and down-fragment is preserved under \emph{currying}. This means that our method allows to investigate how the interpolant is affected by the process of currying. Then, \Cref{thm:interpol} is an immediate consequence of cut elimination and the flipping in~\eqref{eq:flip}.

However, in order to obtain this flipping lemma from generalized splitting, we need to be able to decompose the proof system $\sys$ into a \emph{core fragment} and a \emph{non-core fragment}. Fortunately, this separation has already been observed for linear logic~\cite{dissvonlutz} and classical logic~\cite{brunnler:phd}, but not yet for modal logic, which is the reason for the need of the novel proof systems mentioned above.

This core/non-core separation gives the proof of the interpolation theorem a rather modular structure, in the sense that minor changes in the proof system do not force us to redo the whole proof from scratch.
As mentioned above, we are going to prove our results for classical, linear, and modal logics, and the tools of deep inference allow us to design the proof systems in such a way that the results for classical and modal logic follow almost trivially from the results for linear logic.

For this reason, 
we will start this paper by recalling linear logic and how it is presented in a deep inference setting in \Cref{sec:prelim}. In that section we will also recall some standard properties of deep inference proof systems.
Then, in \Cref{sec:splitting} we discuss the generalized splitting lemma, and in \Cref{sec:flipping} we discuss the flipping lemma and show how it is used to prove interpolation.
After that, we repeat in \Cref{sec:classical} the same procedure for classical logic, and in \Cref{sec:modal} for modal logic.

\section{Preliminaries: Linear Logic and Deep Inference}
\label{sec:prelim}

The set $\Fms=\set{A,B,C,\ldots}$ of \defn{formulas} of linear logic~\cite{girard:87} is generated from a countable set $\Atms=\set{a,b,c,\ldots}$ of \defn{atoms} by the following grammar:
$$
\cF\grammareq \vlone\mid\vlbot\mid\vltop\mid\vlzer\mid\cA\mid\lneg\cA\mid
\cF\vlte\cF\mid\cF\vlpa\cF\mid\cF\vlwi\cF\mid\cF\vlpl\cF\mid
\vloc\cF\mid\vlwn\cF
$$
where the \defn{units} $\vlone,\vlbot,\vltop,\vlzer$ are called \defn{one}, \defn{bottom}, \defn{top}, \defn{zero}, respectively; the binary connectives $\vlte,\vlpa,\vlwi,\vlpl$ are called \defn{tensor}, \defn{par}, \defn{with}, \defn{plus}, respectively; and the two modalities $\vloc,\vlwn$ are called \defn{of course} and \defn{why not}, respectively. \defn{Linear negation} is defined for all formulas via De~Morgan duality:
$$
\def\myskip{\hskip.9em}
\scalebox{.95}{$
\begin{array}{@{}r@{\;}c@{\;}l@{\myskip}r@{\;}c@{\;}l@{\myskip}r@{\;}c@{\;}l@{\myskip}r@{\;}c@{\;}l@{\myskip}r@{\;}c@{\;}l@{\myskip}r@{\;}c@{\;}l@{\myskip}r@{\;}c@{\;}l}
  \lneg{(a)}&=&\lneg a
  &
  \lneg{(A\vlte B)}&=&\lneg A\vlpa\lneg B
    &
    \lneg{(A\vlwi B)}&=&\lneg A\vlpl\lneg B
    &
    \lneg{(\vloc A)}&=&\vlwn\lneg A
    &
    \lneg\vlone&=&\vlbot
    &
    \lneg\vltop&=&\vlzer
    \\
    \llneg a&=&a
  &
  \lneg{(A\vlpa B)}&=&\lneg A\vlte\lneg B
    &
    \lneg{(A\vlpl B)}&=&\lneg A\vlwi\lneg B
    &
    \lneg{(\vlwn A)}&=&\vloc\lneg A
    &
    \lneg\vlbot&=&\vlone
    &
    \lneg\vlzer&=&\vltop
\end{array}
$}
$$
It follows that $\llneg{A}=A$ for all formulas. We define \defn{linear implication} as $A\vlli B=\lneg A\vlpa B$. We also define two equivalence relations $\eqpa$ and $\eqte$, which are the smallest congruence relations obeying the following:
$$
\begin{array}{r@{\;}c@{\;}l@{\hskip3em}r@{\;}c@{\;}l@{\hskip3em}r@{\;}c@{\;}l}
  A\vlpa(B\vlpa C) & \eqpa &(A\vlpa B)\vlpa C
  &
  A\vlpa  B&\eqpa&B\vlpa A
  &
  A\vlpa\vlbot &\eqpa&A
  \\
  A\vlte(B\vlte C) & \eqte&(A\vlte B)\vlte C
  &
  A\vlte  B&\eqte&B\vlte A
  &
  A\vlte\vlone &\eqte&A
\end{array}
$$
\begin{figure}[!t]
  \def\mskipa{-1.5ex}
  \newcommand{\mcol}[2][c@{\quad}]{\multicolumn{2}{#1}{#2}}
  \vskip-\baselineskip
  $$
  \begin{array}{cccc|cccc} 
    \vlinf{\eD}{}{\vloc\vlone}{\vlone}
    &
    \vlinf{\topD}{}{\vltop}{\vlone}
    &
    \vlinf{\dzD}{}{\vltop\vlpa\vlzer}{\vltop}
    &
    \vlinf{\dzpD}{}{\vlzer\vlpa\vlzer}{\vlzer}
    &
    \vlinf{\dzpU}{}{\vltop}{\vltop\vlte\vltop}
    &
    \vlinf{\dzU}{}{\vltop}{\vlzer\vlte\vltop}
    &
    \vlinf{\topU}{}{\vlbot}{\vlzer}
    &
    \vlinf{\eU}{}{\vlbot}{\vlwn\vlbot}
    \\&&&\\[\mskipa]
    \multicolumn{4}{c|}{
      \begin{array}{@{}c@{\quad\;}c@{\quad\;}c}
        \vlinf{\withD}{}{\vlone\vlwi\vlone}{\vlone}
        &
        \vlinf{\tensD}{}{\vlone\vlte\vlone}{\vlone}
        &
        \vlinf{\eqD}{\proviso{if $A\eqpa B$}}{B}{\;A\;}
      \end{array}
    }
    &
    \multicolumn{4}{c}{
       \begin{array}{@{}c@{\quad\;}c@{\quad\;}c}
         \vlinf{\eqU}{\proviso{if $A\eqte B$}}{\;B\;}{A}
         &
         \vlinf{\tensU}{}{\vlbot}{\vlbot\vlpa\vlbot}
         &
         \vlinf{\withU}{}{\vlbot}{\vlbot\vlpl\vlbot}
       \end{array}
    }
    \\&&&\\[\mskipa]
    \multicolumn{4}{c|}{
      \begin{array}{c@{\qquad}c}
        \vlinf{\aiD}{}{\lneg a\vlpa a}{\vlone}
        &
        \vlinf{\sD}{}{(A\vlte B)\vlpa(C\vlpa D)}{(A\vlpa C)\vlte(B\vlpa D)}       
        \\\\[\mskipa]
        \vlinf{\pD}{}{\vloc A\vlpa\vlwn B}{\vloc(A\vlpa B)}
        &
        \vlinf{\dD}{}{(A\vlwi B)\vlpa(C\vlpl D)}{(A\vlpa C)\vlwi(B\vlpa D)}
        \\\\[\mskipa]
        \vlinf{\ppD}{}{\vlwn A\vlpa\vlwn B}{\vlwn(A\vlpa B)}
        &
        \vlinf{\dpD}{}{(A\vlpl B)\vlpa(C\vlpl D)}{(A\vlpa C)\vlpl(B\vlpa D)}        
      \end{array}
    }
    &
    \multicolumn{4}{c}{ 
      \begin{array}{c@{\qquad}c}
        \vlinf{\sU}{}{(A\vlte B)\vlpa(C\vlte D)}{(A\vlpa C)\vlte(B\vlte D)}
        &
        \vlinf{\aiU}{}{\vlbot}{\lneg a\vlte a}
        \\\\[\mskipa]
        \vlinf{\dU}{}{(A\vlte B)\vlpl(C\vlte D)}{(A\vlpl C)\vlte(B\vlwi D)}
        &
        \vlinf{\pU}{}{\vlwn(A\vlte B)}{\vlwn A\vlte \vloc B}
        \\\\[\mskipa]
        \vlinf{\dpU}{}{(A\vlte B)\vlpl(C\vlte D)}{(A\vlpl C)\vlte(B\vlpl D)}
        &
        \vlinf{\ppU}{}{\vloc(A\vlte B)}{\vloc A\vlte \vloc B}
      \end{array}
    }
    \\&&&\\[\mskipa]    
    \vlinf{\wD}{}{\vlwn A}{\vlbot}
    &
    \vlinf{\tD}{}{\vlwn A}{A}
    &
    \vlinf{\gD}{}{\vlwn A}{\vlwn\vlwn A}
    &
    \vlinf{\cD}{}{\vlwn A}{\vlwn A\vlpa \vlwn A}
    &
    \vlinf{\cU}{}{\vloc A\vlte\vloc A}{\vloc A}
    &
    \vlinf{\gU}{}{\vloc\vloc A}{\vloc A}
    &
    \vlinf{\tU}{}{A}{\vloc A}
    &
    \vlinf{\wU}{}{\vlone}{\vloc A}
    \\&&&\\[\mskipa]
    \vlinf{\wpD}{}{A}{\vlzer}
    &
    \vlinf{\wlD}{}{A\vlpl B}{A}
    &
    \vlinf{\wrD}{}{A\vlpl B}{B}
    &
    \vlinf{\cpD}{}{A}{A\vlpl A}
    &
    \vlinf{\cpU}{}{A\vlwi A}{A}
    &
    \vlinf{\wrU}{}{A}{A\vlwi B}
    &
    \vlinf{\wlU}{}{B}{A\vlwi B}
    &
    \vlinf{\wpU}{}{\vltop}{A}
  \end{array}
  $$\vskip-2ex
  \caption{System $\SLSp$, a deep inference system for linear logic. Left: down-fragment, Right: up-fragment. First four lines without the $\dzp$-rules: System $\SLSc$. First five lines: System~$\SLScp$. Last two lines: System~$\SLSnc$. Everything, except the fifth line and the $\dzp$-rules: System $\SLS$. Second and third line without the $\vlwi$-rules: System $\SMLS$, a system for multiplicative linear logic ($\MLL$). Then, systems $\LSp,\LSc,\LScp,\LSnc,\LS,\MLS$ are the corresponding systems without the up-fragment.}
  \label{fig:SLS}
\end{figure}
%
\begin{definitionn}
A \defn{proof system} is a set of inference rules of the shape $\upsmash{\vlinf{\rr}{}{B}{A}}$, where $A$ is called the \defn{premise} and $B$ is called the \defn{conclusion} of $\rr$. The inference rules for linear logic that we are using in this paper are all shown in \Cref{fig:SLS}. The rules with $\downarrow$ \resp{$\uparrow$} in the name are called \defn{down-rules} \resp{\defn{up-rules}}. We define the following proof systems:
\begin{equation}
  \label{eq:systems}
  \begin{array}{@{\hskip-1em}l@{\;=\;}l@{\qquad}l@{\;=\;}l}
    \MLS&\set{\aiD,\sD,\tensD,\eqD}
    &
    \SMLS&\MLS\cup\set{\aiU,\sU,\tensU,\eqU}
    \\
    \LSc&\MLS\cup\set{\pD,\dD,\dzD,\eD,\withD,\topD}
    &
    \SLSc&\LSc\cup\set{\pU,\dU,\dzU,\eU,\withU,\topU}
    \\
    \LSnc&\set{\cD,\wD,\gD,\tD,\cpD,\wpD,\wlD,\wrD}
    &
    \SLSnc&\LSnc\cup\set{\cU,\wU,\gU,\tU,\cpU,\wpU,\wlU,\wrU}
    \\
    \LS&\LSc\cup\LSnc 
    &
    \SLS&\SLSc\cup\SLSnc 
    \\
    \LScp&\LSc\cup\set{\ppD,\dpD,\dzpD}
    &
    \SLScp&\LScp\cup\set{\ppU,\dpU,\dzpU}
    \\
    \LSp&\LScp\cup\LSnc 
    &
    \SLSp&\SLScp\cup\SLSnc 
  \end{array}
\end{equation}
where the index $\mathsf{c}$ in the name of the systems stands for \defn{core} and the $\mathsf{nc}$ for \defn{non-core}. 
For every proof system $\sys$ defined above, we write $\sysD$ \resp{$\sysU$} consisting only of the down-rules \resp{up-rules} of $\sys$. This means that for every $\sys\in\set{\MLS,\LSc,\LS,\LScp,\LSp,\LSnc}$, we have $\sysX=\sysD=\sysSXD$ and $\sysU=\emptyset$.
\end{definitionn}
\begin{definitionn}\label{def:derivation}
We define the set $\Ders_\sysX=\set{\deri,\derib,\ldots}$ of \defn{derivations} in a proof system $\sysX$ together with two functions \hbox{$\premf,\concf\colon\Ders_\sysX\to\Fms$} (called \defn{premise} and \defn{conclusion}, respectively) inductively as follows:
\begin{itemize}
\item For every formula $A$, we have that $\nodf{A}$ is a derivation and $\prem{\nodf{A}}=\conc{\nodf{A}}=A$.
\item If $\deri$ and $\derib$ are derivations and $\vlva \in \set{\vlte,\vlpa,\vlwi,\vlpl}$, then 
    $\nodf{\deri\vlva \derib}$ is a derivation and $\prem{\nodf{\deri\vlva \derib}}=\prem{\deri}\vlva\prem{\derib}$ and $\conc{\nodf{\deri\vlva \derib}}=\conc{\deri}\vlva\conc{\derib}$.
  \item If $\deri$ is a derivation and $\vlda\in\set{\vloc,\vlwn}$, then so is $\vlda\deri$ and $\prem{\vlda\deri}=\vlda\prem\deri$ and $\conc{\vlda\deri}=\vlda\conc\deri$. 
  \item If $\deri$ and $\derib$ are derivations and $\downsmash{\vlinf{\rr}{}{B}{A}}$ is an instance of an inference rule in $\sysX$ and $\conc\deri=A$ and $\prem\derib=B$ then $\deric={\odn{\deri}{\rr}{\derib}{}}$ is a derivation and $\prem{\deric}=\prem{\deri}$ and $\conc{\deric}=\conc{\derib}$.
  \end{itemize}
  We write $\vlupsmash{\odv{A}{\deri}{B}{\sysS}}$ \Resp{$\downsmash{\odr{\deri}{B}{\sysS}}$} to denote a derivation $\deri$ with conclusion $\conc\deri=B$ and premise $\prem\deri=A$ \resp{$\prem\deri=\vlone$} where all inference rules used in $\deri$ are in the proof system~$\sysS$, and we write $A\proves[\sysS] B$ \resp{$\proves[\sysS] B$} to indicate that such a derivation exists.
\end{definitionn}
This definition of  derivations follows the \emph{open deduction} notation, and we refer the reader to~\cite{GGP:open} for more details.


\begin{theorem}\label{thm:soundSLS}
  Every system $\sysS$ with $\LS\subseteq\sysS\subseteq\SLS'$ is sound and complete for linear logic.
\end{theorem}

\begin{proof}
  Every inference rule in $\SLS'$ constitutes a sound implication of linear logic, and conversely, every sequent calculus rule of linear logic can be simulated by $\LS$. Details can be found in~\cite{dissvonlutz,str:02,str:MELL}.
\end{proof}
\begin{remark}\label{rem:rulesSKS}
  There are some design choices that make our system $\LS$ differ from existing versions \cite{dissvonlutz,str:02,CGS:foccos} in the literature:
  \begin{itemize}
  \item Instead of the standard (self-dual) \defn{switch} rule $\upsmash{\vlinf{\sw}{}{A\vlpa(B\vlte C)}{(A\vlpa B)\vlte C}}$ as in~\cite{dissvonlutz,str:02}, we have the two rules $\sD$ and $\sU$. However, observe that $\sw$ is derivable in $\set{\sD,\eqD}$, and $\sD$ is derivable in $\set{\eqD,\sw,\eqU}$. Our choice has been made to avoid the use of $\eqte$ in the down-fragment. In~\cite{CGS:foccos}, the authors used two versions of switch, $\swl$ and $\swr$, instead.
  \item Our variant of $\LS$ uses much less equations than the variant in \cite{dissvonlutz,str:02} and we use the rules $\eD$, $\gD$, $\tensD$, $\withD$, and $\topD$ instead. This is similar to the variant of $\LS$ in~\cite{CGS:foccos}, which uses different rules. 
  \item We replaced the rule $\upsmash{\vlinf{\bD}{}{\vlwn A}{\vlwn A\vlpa A}}$ of~\cite{dissvonlutz,str:02} by the rules $\cD$ and $\tD$. Note that the version in~\cite{CGS:foccos} selects another combination of rules. Our choice makes it easier to translate our results to modal logic in \Cref{sec:modal}.
  \end{itemize}
  It should be emphasized that these differences are only of cosmetic nature and all three systems (the one in~\cite{str:02}, the one in~\cite{CGS:foccos}, and the one here) are equivalent.
\end{remark}
\begin{remark}
  However, adding the  rules $\ppD$, $\dpD$, $\dzpD$ is more than a cosmetic change. They are not needed for completeness of the system, and they are not derivable in $\LS$. They do not allow to prove more formulas (i.e., they are admissible for $\LS$), but they play a crucial role for proving \Cref{thm:interpol} and the generalized splitting lemma in the next section. The rule $\ppD$ has been observed already in~\cite{dissvonlutz} for proving interpolation for multiplicative exponential linear logic ($\MELL$), and the rules $\dpD$ and $\dzpD$ have been used in~\cite{str:02} as part of a set of rules to reduce contraction (our rules $\cD$ and $\cpD$) and weakening (all rules with a $\mathsf{w}$ in the name) to an atomic form.
\end{remark}

The following two rules are called \defn{axiom} and \defn{cut}, respectively:
\begin{equation}\label{eq:cut}
  \hfill
  \scalebox{.9}{$
    \vlinf{\iD}{}{\lneg A \vlpa A}{\vlone}
    \qquand
    \vlinf{\iU}{}{\vlbot}{A \vlte \lneg A}
    $}
  \hfill
\end{equation}
They are the generalizations of the atomic interaction rule $\aiD$ and its dual $\aiU$,
and they play a crucial role in the following proposition which collects some well-known facts about deep inference systems (see, e.g., \cite{dissvonlutz,str:02,brunnler:phd,ATS:esslli2019}, or \Cref{app:prelim} for more details
).

\begin{restatable}{proposition}{propAI}\label{prop:DI}
  We have the following:
  \begin{enumerate}
  \item\label{prop:ai}
    The rule $\iD$ is derivable in $\LSc$, and dually, the rule $\iU$ is derivable in $\SLScU$.
  \item \label{prop:updown}
    Every rule $\rU\in\SLSpU$ is derivable in $\set{\rD,\iD,\sD,\iU,\eqD,\eqU}$.
  \item\label{prop:flip1}
    For every $\sys\in\set{\SLSc,\SLScp,\SLS,\SLSp}$, we have $A\proves[\sys]B$ iff\/ $\proves[\sys]\lneg A\vlpa B$. 
  \item\label{prop:flip2}
    For
    every $\sys\in\set{\SLSc,\SLScp,\SLS,\SLSp,\SLSnc}$, we have $A\proves[\sysD]B$ iff\/ $\lneg B\proves[\sysU]\lneg A$.
  \item\label{prop:cutup}
  For every $\sys\in\set{\LSc,\LScp,\LS,\LSp}$, we have $\proves[\sysSX]A$ iff\/ $\proves[\sys\cup\set{\iU}]A$.\\[-2ex]
  \end{enumerate}
\end{restatable}




The following theorem states the cut elimination result. It follows immediately from \Cref{thm:soundSLS} and \Cref{prop:DI}.\ref{prop:cutup}. For details on how to prove it inside deep inference, the interested reader is referred to \cite{dissvonlutz,CGS:foccos,CGRS:foccos-long}.
\begin{restatable}[Cut Elimination]{theorem}{thmCutElim}\label{thm:cutelim}
  For every $\sys\in\set{\LSc,\LScp,\LS,\LSp}$, if\/ $\proves[\sys\cup\set{\iU}]A$ then\/ $\proves[\sysX]A$.\\[-2ex]
\end{restatable}

Finally, we also need the notion of \defn{context} $\Chole$, which is a formula with a single occurrence of a \defn{hole} $\cons\chole$, taking the place of a subformula, or a derivation. We write $\Chole[A]$ for the formula obtained by replacing the hole in $\Chole$ with the formula $A$. Similarly, we write \scalebox{.9}{$\nodf{\CHole[\odv{A}{\deri}{B}{\sysS}]}$} for the derivation with premise $\Chole[A]$ and conclusion $\Chole[B]$ that is obtained from putting $\deri$ in context $\Chole$.

If the hole $\conhole$ occurs only inside multiplicative connectives $\vlte$ and $\vlpa$, we speak of a \defn{multiplicative context} (or \defn{m-context}). More formally, the set $\Ctxs$ of m-contexts is generated by the grammar:
$$
\Ctxs\grammareq\conhole\mid\Fms\vlte\Ctxs\mid\Ctxs\vlte\Fms
\mid\Fms\vlpa\Ctxs\mid\Ctxs\vlpa\Fms
$$
The reason for their special treatment is the following lemma which does not hold for arbitrary contexts. The proof is standard (see e.g., \cite{dissvonlutz,ATS:esslli2019}) but we show it also in \Cref{app:prelim}.
\begin{restatable}{lemma}{lemSwitch}\label{lem:switch}
  For all formulas $A,B$ and m-contexts $\Chole$, there are derivations
  $$
  \odv{\Chole[A\vlpa B]}{\deri}{A\vlpa\Chole[B]}{\set{\eqD,\sD}}
  \qquand
  \odv{A\vlte\Chole[B]}{\derib}{\Chole[A\vlte B]}{\set{\eqU,\sU}}
  $$
\end{restatable}

\section{Generalized Splitting}
\label{sec:splitting}

The splitting lemma in a deep inference system~\cite{gug:SIS} says that when a formula $(A\vlte B)\vlpa K$ is provable, then the context $K$ can be split into two contexts $K_A$ and $K_B$, such that $A\vlpa K_A$ and $B\vlpa K_B$ are both provable. This is similar to what happens in the sequent calculus: when a sequent $\sqn{A\vlte B,\Gamma}$ is provable, then at some point in the proof, the formula $A\vlte B$ becomes principal and splits the context, via the $\vlte$-rule in linear logic, as indicated below:
$$
\begin{tikzpicture}[x=1em,y=1ex]
  \node (a) at (0,0) {$\vliiinf{\vlte}{}{
      \sqn{A\vlte B,\Gamma_1,\Gamma_2}}{\sqn{A,\Gamma_1}}{}{\sqn {B,\Gamma_2}}$};
  \node (b) at (-2,5.5) {$\ddots\iddots$};
  \node  at (3.5,5.5) {$\ddots\iddots$};
  \node at (2,-4.5) {$\ddots$};
  \node at (4,-8) {$\sqn{A\vlte B,\Gamma}$};
  \node at (6,-4.5) {$\iddots$};
\end{tikzpicture}
$$
A possible interpretation of the splitting lemma is that it finds the information about how to split the context, i.e., it finds $\Gamma_1$ and $\Gamma_2$, without actually constructing the sequent proof.

On the other hand, Maehara's method~\cite{maehara:interpolation} extracts the information about the interpolant from a (cut-free) sequent proof. A natural question to ask is then, whether a variant of the splitting lemma can also extract that information without actually constructing the sequent proof.\looseness=-1

This is the purpose of the \emph{generalized splitting lemma}. Unlike the standard splitting lemma, it does not start from a proof of a formula $(A\vlte B)\vlpa K$ but from a derivation with a non-empty premise. The lemma is stated below and proved in \Cref{app:splitting}.

\begin{restatable}[Generalized Splitting]{lemma}{lemSplit}\label{lem:splitting}
  Let $J,K,A,B$ be formulas and $a$ be an atom or negated atom.
  \begin{enumerate}
  \item\label{s:tensor} If $\nnodv{J}{\deri}{(A\vlte B)\vlpa K}{\LScp}$ then there is an m-context $\Chole$ and formulas $J_A,J_B,K_A,K_B$ such that
    $$
    \odv{J}{\deri_J}{\Ccons{J_A\vlte J_B}}{\SLScpU}
    \quand
    \odv{\Ccons{K_A\vlpa K_B}}{\deri_K}{K}{\LScp}
    \quand
    \odv{J_A}{\deri_A}{A\vlpa K_A}{\LScp}
    \quand
    \odv{J_B}{\deri_B}{B\vlpa K_B}{\LScp}
    $$
  \item\label{s:with} If $\nnodv{J}{\deri}{(A\vlwi B)\vlpa K}{\LScp}$ then there is either
    (i) an m-context $\Chole$ and formulas $J_A,J_B$ such that
    $$
    \odv{J}{\deri_J}{\Ccons{J_A\vlwi J_B}}{\SLScpU}
    \quand
    \odv{\Ccons{\vlbot}}{\deri_K}{K}{\LScp}
    \quand
    \odv{J_A}{\deri_A}{A}{\LScp}
    \quand
    \odv{J_B}{\deri_B}{B}{\LScp}
    $$
    or (ii) an m-context $\Chole$ and formulas $J_A,J_B,K_A,K_B$ such that
    $$
    \odv{J}{\deri_J}{\Ccons{J_A\vlwi J_B}}{\SLScpU}
    \quand
    \odv{\Ccons{K_A\vlpl K_B}}{\deri_K}{K}{\LScp}
    \quand
    \odv{J_A}{\deri_A}{A\vlpa K_A}{\LScp}
    \quand
    \odv{J_B}{\deri_B}{B\vlpa K_B}{\LScp}
    $$
  \item\label{s:plus} If $\nnodv{J}{\deri}{(A\vlpl B)\vlpa K}{\LScp}$ then there is either
    (i) an m-context $\Chole$ and formulas $J_A,J_B$ such that
    $$
    \odv{J}{\deri_J}{\Ccons{J_A\vlpl J_B}}{\SLScpU}
    \quand
    \odv{\Ccons{\vlbot}}{\deri_K}{K}{\LScp}
    \quand
    \odv{J_A}{\deri_A}{A}{\LScp}
    \quand
    \odv{J_B}{\deri_B}{B}{\LScp}
    $$
    or (ii) a m-context $\Chole$ and formulas $J_A,J_B,K_A,K_B$ such that
    $$
    \odv{J}{\deri_J}{\Ccons{J_A\vlva J_B}}{\SLScpU}
    \quand
    \odv{\Ccons{K_A\vlva K_B}}{\deri_K}{K}{\LScp}
    \quand
    \odv{J_A}{\deri_A}{A\vlpa K_A}{\LScp}
    \quand
    \odv{J_B}{\deri_B}{B\vlpa K_B}{\LScp}
    $$
    where $\vlva$ is either $\vlpl$ or $\vlwi$.
  \item\label{s:oc} If $\nnodv{J}{\deri}{\vloc A\vlpa K}{\LScp}$ then there is either (i) an m-context $\Chole$ and a formula $J_A$ such that
    $$
    \odv{J}{\deri_J}{\Ccons{\vloc J_A}}{\SLScpU}
    \quand
    \odv{\Ccons{\vlbot}}{\deri_K}{K}{\LScp}
    \quand
    \odv{J_A}{\deri_A}{A}{\LScp}
    $$
    or (ii) an m-context $\Chole$ and formulas $J_A$ and $K_A$ such that
    $$
    \odv{J}{\deri_J}{\Ccons{\vloc J_A}}{\SLScpU}
    \quand
    \odv{\Ccons{\vlwn K_A}}{\deri_K}{K}{\LScp}
    \quand
    \odv{J_A}{\deri_A}{A\vlpa K_A}{\LScp}
    $$
  \item\label{s:wn} If $\nnodv{J}{\deri}{\vlwn A\vlpa K}{\LScp}$ then there is either (i) an m-context $\Chole$ and a formula $J_A$ such that
    $$
    \odv{J}{\deri_J}{\Ccons{\vlwn J_A}}{\SLScpU}
    \quand
    \odv{\Ccons{\vlbot}}{\deri_K}{K}{\LScp}
    \quand
    \odv{J_A}{\deri_A}{A}{\LScp}
    $$
    or (ii) an m-context $\Chole$ and formulas $J_A$ and $K_A$ such that
    $$
    \odv{J}{\deri_J}{\Ccons{\vlda J_A}}{\SLScpU}
    \quand
    \odv{\Ccons{\vlda K_A}}{\deri_K}{K}{\LScp}
    \quand
    \odv{J_A}{\deri_A}{A\vlpa K_A}{\LScp}
    $$
    where $\vlda$ is either $\vloc$ or $\vlwn$.
  \item\label{s:one} If $\nnodv{J}{\deri}{\vlone\vlpa K}{\LScp}$ then there is an m-context $\Chole$ such that
    $$
    \odv{J}{\deri_J}{\Ccons{\vlone}}{\SLScpU}
    \quand
    \odv{\Ccons{\vlbot}}{\deri_K}{K}{\LScp}
    $$
   \item\label{s:top} If $\nnodv{J}{\deri}{\vltop\vlpa K}{\LScp}$ then there is an m-context $\Chole$ and formulas $\ttt\in\set{\vlone,\vltop}$ and $\fff\in\set{\vlzer,\vlbot}$ such that
     $$
    \odv{J}{\deri_J}{\Ccons{\ttt}}{\SLScpU}
    \quand
    \odv{\Ccons{\fff}}{\deri_K}{K}{\LScp}
    $$
    (i.e., there is one of four possible cases). 
   \item\label{s:zero} If $\nnodv{J}{\deri}{\vlzer\vlpa K}{\LScp}$ then there is an m-context $\Chole$ such that
     $$
     \mbox{either (i)}\quad
    \odv{J}{\deri_J}{\Ccons{\vlzer}}{\SLScpU}
    \quand
    \odv{\Ccons{\fff}}{\deri_K}{K}{\LScp}
    \quad,\qquad\mbox{or (ii)}\quad
    \odv{J}{\deri_J}{\Ccons{\ttt}}{\SLScpU}
    \quand
    \odv{\Ccons{\vltop}}{\deri_K}{K}{\LScp}
    $$
    for some $\fff\in\set{\vlzer,\vlbot}$ or some $\ttt\in\set{\vlone,\vltop}$ (i.e., there is again one of four possible cases). 
   \item\label{s:atom} If $\nnodv{J}{\deri}{a\vlpa K}{\LScp}$ then there is an m-context $\Chole$ such that
     $$
     \mbox{either (i)}\quad
    \odv{J}{\deri_J}{\Ccons{a}}{\SLScpU}
    \quand
    \odv{\Ccons{\vlbot}}{\deri_K}{K}{\LScp}
     \quad,\qquad\mbox{or (ii)}\quad
    \odv{J}{\deri_J}{\Ccons{\vlone}}{\SLScpU}
    \quand
    \odv{\Ccons{\lneg a}}{\deri_K}{K}{\LScp}
    $$
 \end{enumerate}
\end{restatable}

\begin{remark}
  Observe that this lemma holds for system $\LScp$. First, it is crucial that the rules $\ppD$, $\dpD$ and $\dzpD$ are present. They are not needed for the standard splitting lemma (see \cite{dissvonlutz} or \cite{CGS:foccos}), but our generalized splitting cannot be proved without them. Second, it is crucial that we stay in the core fragment. Even though, it is possible to prove a splitting lemma for the whole system, including the non-core (see, e.g., \cite{dissvonlutz,CGS:foccos}), such a more general statement would make it much harder to derive the flipping lemma in the next section.  
\end{remark}

\section{Flipping plus Decomposition gives Interpolation}
\label{sec:flipping}

\begin{restatable}[Core Flipping]{lemma}{lemCoreFlip}\label{lem:coreflipping}
  Let $H,F,J,K$ be formulas.\\
  If there is a derivation
  $\od{\odd{\odd{\odh{H}}{\derib}{
        J}{\SLScpU}}{\deri}{
      F\vlpa K}{\LScp}}
  $
  then there is a formula $I$, such that we have
    $\od{\odd{\odd{\odh{\lneg F\vlte H}}{\deric}{
        I}{\SLScpU}}{\derid}{
      K}{\LScp}}
  $.
\end{restatable}
The proof is a straightforward induction on $F$, repeatedly applying the generalized splitting to the derivation $\deri$. The details are in \Cref{app:flipping}.

This lemma is enough to obtain interpolation for $\LScp$.
But in order to extend it to all of $\LSp$, we need to deal with the non-core rules (the last two lines in \Cref{fig:SLS}). This is done via a \emph{decomposition} result, that allows us to separate the core from the non-core in a derivation.

\begin{restatable}[Decomposition]{lemma}{lemDecomposition}\label{lem:decomposition}
  Let $n\ge0$ and $P,Q_1,\ldots, Q_n$ be formulas. Then any derivation
  $$
  \odV{P}{\deri}{Q_1\vlpa Q_2,\vlpa\cdots\vlpa Q_n}{\LSp}
  \mbox{\quad can be decomposed into \quad}
  \odv{P}{\deri'}{
    \odv{Q'_1}{\deri''_1}{Q_1}{\LSnc}
    \vlpa
    \cdots
    \vlpa
    \odv{Q'_n}{\deri''_n}{Q_n}{\LSnc}
  }{\LScp}
  $$
  for some formulas $Q'_1,\ldots, Q'_n$.
\end{restatable}
This says that any derivation in $\LSp$ can be decomposed into a core-part and a non-core part, such that the $\vlpa$-structure of the conclusion is preserved. Observe that in the case of $n=0$, this says that a derivation in $\LSp$ with conclusion $\vlbot$ can be transformed into a derivation in $\LScp$, with same premise and conclusion.
\Cref{lem:decomposition} can be proved via straightforward rule permutations, where the two rules $\dpD$ and~$\dzpD$ play a crucial role. More details can be found in \Cref{app:flipping}.

\begin{remark}
  The decomposition in \Cref{lem:decomposition} is much stronger than existing decompositions in the literature. For example, the work in~\cite{str:02} or~\cite{brunnler:tiu:01} can only permute atomic versions of contraction and weakening, and the work in~\cite{str:MELL,SIS-IV} cannot deal with the additives.
\end{remark}

\begin{restatable}[Flipping]{theorem}{thmFlipping}\label{thm:flipping}
  Let $H,F,J,K$ be formulas.\\
  If there is a derivation
  $\od{\odd{\odd{\odh{H}}{\derib}{
        J}{\SLSpU}}{\deri}{
      F\vlpa K}{\SLSpD}}
  $
  then there is a formula $I$ and a derivation
    $\od{\odd{\odd{\odh{\lneg F\vlte H}}{\deric}{
        I}{\SLSpU}}{\derid}{
      K}{\SLSpD}}
  $.
\end{restatable}
\begin{proof}
  We can perform the following transformation:
  $$
  \od{\odd{\odd{\odh{H}}{\derib}{
        J}{\SLSpU}}{\deri}{
      F\vlpa K}{\SLSpD}}
  \qsto{1}
  \od{\odd{\odd{\odh{\hskip1.3em\odv{H}{\derib''}{H'}{\SLSncU}}}{\derib'}{
        J}{\SLScpU}}{\deri'}{
      \odv{F'}{\deri_F''}{F}{\SLSncD}
      \vlpa
      \odv{K'}{\deri_K''}{K}{\SLSncD}}{\SLScpD}}
  \qsto{2}
  \od{\odd{\odd{\odh{
          \odv{\lneg F}{\lneg{\deri_F''}}{\lneg {F'}}{\SLSncU}
          \vlte
          \odv{H}{\derib''}{H'}{\SLSncU}}}{\deric'}{
        I}{\SLScpU}}{\derid'}{\hskip1.2em
      \odv{K'}{\deri_K''}{K}{\SLSncD}}{\SLScpD}}
  $$
  where Step~1 is the application of \Cref{lem:decomposition} to $\deri$ and its dual to $\derib$. Then, Step~2 applies \Cref{lem:coreflipping} to the derivation from $H'$ to $F'\vlpa K'$ and \Cref{prop:DI}.\ref{prop:flip2} to $\deri_F''$.
\end{proof}

\begin{remark}
  Maybe at this point we should recall \Cref{rem:rulesSKS}, because all our design choices for the inference rules are motivated by our proof of \Cref{thm:flipping}. For example, the curious reader might have wondered why there are the three rules $\wpD$, $\wrD$ and $\wlD$. It is perfectly sensible to expect that either only $\wpD$, or only $\wrD$ and $\wlD$ should have been enough. But the need for all three rules comes from the fact that we want a separation of the proof system into core and non-core such that both, (i) the decomposition of \Cref{lem:decomposition} preserving the $\vlpa$-structure of the conclusion, and (ii) the generalized splitting in~\Cref{lem:splitting} allowing for the proof of the core flipping in \Cref{lem:coreflipping}. For the same reason, we have the two rules $\topD$ and $\dzD$ instead of just a rule $\vlinf{\topD}{}{\vltop\vlpa\vlzer}{\vlone}$.
%
\end{remark}


\begin{corollary}[Interpolation]\label{cor:interpolSLS}
  Every derivation $\upsmash{\odV{A}{\deri}{B}{\SLSp}}$ can be decomposed into
  $\upsmash{
    \od{\odd{\odd{\odh{A}}{\deric}{
        I}{\SLSpU}}{\derid}{
      B}{\SLSpD}}
  }$
  for some formula~$I$. Furthermore, this formula $I$ is a Lyndon interpolant for $A\vlli B$, i.e., every atom occuring positively (resp.~negatively) in $I$ also occurs   positively (resp.~negatively) in $A$ and in~$B$.
\end{corollary}

\begin{proof}
  We can perform the following transformations:
  \begin{equation}
    \label{eq:int-proof}
    \hfill
  \odv{A}{\deri}{B}{\SLSp}
  \qsto{1}
  \odv{\vlone}{\deri'}{\lneg A\vlpa B}{\SLSp}
  \qsto{2}
  \odv{\vlone}{\deri''}{\lneg A\vlpa B}{\LSp}
  \qsto{3}
  \upsmash{
    \od{\odd{\odd{\odh{A}}{\deric}{
        I}{\SLSpU}}{\derid}{
        B}{\SLSpD}}
    }
  \hfill
  \end{equation}
  where Step~1 is \Cref{prop:DI}.\ref{prop:flip1},
  Step~2 is cut elimination (\Cref{thm:cutelim} and \Cref{prop:DI}.\ref{prop:cutup}), and
  Step~3 is \Cref{thm:flipping} (with $H=J=\vlone$ and $F=\lneg A$ and $K=B$).
  It follows immediately that $I$ is a Lyndon interpolant because no $\uparrow$-rule can introduce new atoms when going down in the derivation, and no $\downarrow$-rule can delete atoms when going down in the derivation.
\end{proof}

\begin{remark}
  Even though we used cut elimination to prove interpolation, it should be observed that the decomposition in up- and down-fragment in \Cref{cor:interpolSLS} entails cut-elimination in two different ways. The first is obvious: if $A=\vlone$, then necessarily we also have $I=\vlone$ because no rule in the up-fragment has $\lone$ as premise.\footnote{Except for $\wpU$, which can be replaced by an instance of $\topD$ if the premise is $\vlone$.} Hence, we also have a cut-free derivation of $B$. The second way is more interesting: we can flip $\deric$ via \Cref{prop:DI}.\ref{prop:flip2} and construct the following cut-free derivation of $A\vlli B$:
  \begin{equation}
    \label{eq:interflip}
    \hfil
    \odn{\vlone}{\iD}{
      \odv{\lneg I}{\lneg{\deric}}{\lneg A}{\LSp}
      \vlpa
      \odv{I}{\derid}{B}{\LSp}}{}
  \end{equation}
  This derivation is closely related to the third derivation in~\eqref{eq:int-proof}. More precisely, by inspecting the proof of the generalized splitting lemma in the previous section, one can observe that Step~3 in the proof of \Cref{cor:interpolSLS} above, and the ``down-flipping'' in \eqref{eq:interflip} above are inverse operations. However, to make this formal, we would need to find a way to deal with certain rule permutations. For this, we would need to develop a theory of atomic flows for $\SLSp$, similar to the recent development in~\cite{ASZ:BVcat} for the logic~$\BV$.
  In any case, our observation here is in accordance with the observation of~\cite{saurin:25} that in the sequent calculus, interpolation is a particular case of an inverse cut elimination.
\end{remark}
\section{Classical Logic}
\label{sec:classical}
The \defn{formulas} of classical logic are generated by the grammar
$$
\Fms\grammareq \Atms\mid\cneg\Atms\mid\vltt\mid\vlff\mid\Fms\vlan\Fms\mid\Fms\vlor\Fms
$$
where, as before, $\Atms=\set{a,b,c,\ldots}$ is a countable set of \defn{atoms}. The \defn{units} $\vltt,\vlff$ are called \defn{true} and \defn{false}, respectively, and the binary connectives $\vlan,\vlor$ are called \defn{and} and \defn{or}, respectively. \defn{Classical negation} $\cneg\cdot$ is extended to  all formulas via De~Morgan duality:
$$
\ccneg a=a
\qquad
\wcneg{A\vlan B}=\cneg A\vlor \cneg B
\qquad
\wcneg{A\vlor B}=\cneg A\vlan \cneg B
\qquad
\cneg\vltt=\vlff
\qquad
\cneg\vlff=\vltt
$$
It follows that $\ccneg A=A$ for all $A$, and we can write $A\vlim B$ for $\cneg A\vlor B$. We define the two equivalences $\eqor$ and $\eqan$ as follows:
$$
\begin{array}{r@{\;}c@{\;}l@{\hskip3em}r@{\;}c@{\;}l@{\hskip3em}r@{\;}c@{\;}l}
  A\vlor(B\vlor C) & \eqor &(A\vlor B)\vlor C
  &
  A\vlor  B&\eqor&B\vlor A
  &
  A\vlor\vlff &\eqor&A
  \\
  A\vlan(B\vlan C) & \eqan&(A\vlan B)\vlan C
  &
  A\vlan  B&\eqan&B\vlan A
  &
  A\vlan\vltt &\eqan&A
\end{array}
$$
The inference rules are shown in \Cref{fig:SKS}, and we define the following proof systems:
\begin{equation}
  \label{eq:systemsKS}
  \hfil
  \begin{array}{@{\hskip-1em}l@{\;=\;}l@{\qquad\qquad}l@{\;=\;}l}
    \KSc&\set{\aiD,\sD,\andD,\eqD}
    &
    \SKSc&\KSc\cup\set{\aiU,\sU,\withU,\eqU}
    \\
    \KSnc&\set{\wD,\cD}
    &
    \SKSnc&\KSnc\cup\set{\cU,\wU}
    \\
    \KS&\KSc\cup\KSnc 
    &
    \SKS&\SKSc\cup\SKSnc 
  \end{array}
\end{equation}

\begin{figure}[!t]
  \def\mskipa{-1.5ex}
  \newcommand{\mcol}[2][c@{\quad}]{\multicolumn{2}{#1}{#2}}
  \vskip-\baselineskip
  $$
  \begin{array}{c@{\qquad}c@{\quad}|@{\quad}c@{\qquad}c} 
    \vlinf{\andD}{}{\vltt\vlan\vltt}{\vltt}
    &
    \vlinf{\eqD}{\proviso{if $A\eqor B$}}{\;B\;}{A}
    &
    \vlinf{\eqU}{\proviso{if $A\eqan B$}}{\;B\;}{A}
    &
    \vlinf{\andU}{}{\vlff}{\vlff\vlor\vlff}
    \\&&\\[\mskipa]
    \vlinf{\aiD}{}{\cneg a\vlor a}{\vltt}
    &
    \vlinf{\sD}{}{(A\vlan B)\vlor(C\vlor D)}{(A\vlor C)\vlan(B\vlor D)}       
    &
    \vlinf{\sU}{}{(A\vlan B)\vlor(C\vlan D)}{(A\vlor C)\vlan(B\vlan D)}
    &
    \vlinf{\aiU}{}{\vlff}{\cneg a\vlan a}
    \\&&\\[\mskipa]
    \vlinf{\wD}{}{A}{\vlff}
    &
    \vlinf{\cD}{}{A}{A\vlor A}
    &
    \vlinf{\cU}{}{A\vlan A}{A}
    &
    \vlinf{\wU}{}{\vltt}{A}
  \end{array}
  $$
  \caption{System $\SKS$, a deep inference system for classical logic. Left: Down-fragment, Right: Up-fragment. First two lines: core fragment $\SKSc$. Last line: non-core fragment $\SKSnc$.  Then, systems $\KS,\KSc,\KSnc$ are the corresponding systems without the up-fragment.}
  \label{fig:SKS}
\end{figure}
\begin{remark}
  As in the case of $\LS$ and $\LS'$ in \Cref{sec:prelim}, we have made some design choices with respect to the rules that make our variant of $\KS$ slightly differ from the ones in the literature \cite{brunnler:tiu:01,AGR:cls2017}. The main difference being the different treatment of the equational theory, using $\sD$ and $\sU$ instead of the self-dual \emph{switch} rule, and not having atomic weakening and contraction together with the so-called \emph{medial} rule. The motivation for this is two-fold. First, as in the case of linear logic we want to prove the Flipping~\Cref{thm:flipping} also for $\SKS$, and second, we do not want to do all the work again, i.e., the system is designed such that the result for $\SKS$ follows immediately from the result for $\SLSp$.
\end{remark}

In fact now all results that we stated for $\SLSp$ and $\LSp$ on the previous three sections also hold for $\SKS$ and $\KS$.\footnote{For the corresponding results for \Cref{sec:prelim} see~\cite{brunnler:tiu:01,brunnler:06:locality,brunnler:phd}.} In particular, we have the following. 

\begin{restatable}{theorem}{thmSoundKS}\label{thm:soundnessKS}
  Every proof system $\sysX$ with $\KS\subseteq\sysX\subseteq\SKS$ is sound and complete for classical propositional logic.
\end{restatable}

\begin{proof}
  As for linear logic, the simplest way of proving this is to show that (i) every rule in $\SKS$ constitutes a correct implication, and (ii) every rule in any sound and complete cut-free sequent calculus for classical logic is derivable in $\KS$. For details, see~\cite{brunnler:06:locality,brunnler:phd} or \Cref{app:classical}.\footnote{Alternatively, one can also show equivalence to the variants of $\KS$ and $\SKS$ in~\cite{brunnler:tiu:01}.}
\end{proof}

\begin{corollary}[Cut Elimination]\label{thm:cutelimSKS}
  If\/ $\proves[\SKS] A$ then also\/ $\proves[\KS] A$.
\end{corollary}

This follows immediately from \Cref{thm:soundnessKS}, as every $\SKS$-derivation can be translated into a sequent proof with cut. We can eliminate the cut, and then translate the cut-free sequent calculus into $\KS$. Alternatively, it is also possible to prove cut elimination directly in deep inference. For this, we refer the reader to~\cite{brunnler:06:locality,brunnler:phd}, as it would go beyond the scope of this paper.\looseness=-1


The notion of \defn{context} is defined analogously to the case of linear logic. But the situation here is simpler, because the notions of context and \defn{m-context} coincide for classical logic. 

\begin{restatable}[Generalized Splitting for $\KSc$]{lemma}{lemSplitKSc}\label{lem:splittingKSc}
  Let $J,K,A,B$ be formulas and $a$ be an atom or negated atom.
  \begin{enumerate}
  \item\label{s:and} If $\nnodv{J}{\deri}{(A\vlan B)\vlor K}{\KSc}$ then there is an m-context $\Chole$ and formulas $J_A,J_B,K_A,K_B$ such that
    $$
    \odv{J}{\deri_J}{\Ccons{J_A\vlan J_B}}{\SKScU}
    \quand
    \odv{\Ccons{K_A\vlor K_B}}{\deri_K}{K}{\KSc}
    \quand
    \odv{J_A}{\deri_A}{A\vlor K_A}{\KSc}
    \quand
    \odv{J_B}{\deri_B}{B\vlor K_B}{\KSc}
    $$
  \item\label{s:tt} If $\nnodv{J}{\deri}{\vltt\vlan K}{\KSc}$ then there is an m-context $\Chole$ such that
    $$
    \odv{J}{\deri_J}{\Ccons{\vltt}}{\SKScU}
    \quand
    \odv{\Ccons{\vlff}}{\deri_K}{K}{\KSc}
    $$
   \item\label{s:catom} If $\nnodv{J}{\deri}{a\vlor K}{\KSc}$ then there is an m-context $\Chole$ such that
     $$
     \mbox{either (i)}\quad
    \odv{J}{\deri_J}{\Ccons{a}}{\SKScU}
    \quand
    \odv{\Ccons{\vlff}}{\deri_K}{K}{\KSc}
     \quad,\qquad\mbox{or (ii)}\quad
    \odv{J}{\deri_J}{\Ccons{\vltt}}{\SKScU}
    \quand
    \odv{\Ccons{\cneg a}}{\deri_K}{K}{\KSc}
    $$
  \end{enumerate}
\end{restatable}
\begin{proof}
  We can either redo the proof that we did for $\LScp$, or we can observe that $\KSc$ is the same system as $\MLS$ (the multiplicative fragment of $\LScp$) under the translation
  \begin{equation}
    \label{eq:transl}
    \hfil
    \vlte\transl\vlan
    \qquad\vlpa\transl\vlor
    \qquad
    \vlone\transl\vltt
    \qquad
    \vlbot\transl\vlff
  \end{equation}
  Then the three statements above are \Cref{lem:splitting}.\ref{s:tensor}, \ref{lem:splitting}.\ref{s:one}, and \ref{lem:splitting}.\ref{s:atom}.
\end{proof}

\begin{restatable}[Decomposition for $\KS$]{lemma}{lemDecompositionKS}\label{lem:decompositionKS}
  Let $n\ge0$ and $P,Q_1,\ldots, Q_n$ be formulas. Then
  $$
  \odV{P}{\deri}{Q_1\vlor Q_2,\vlor\cdots\vlor Q_n}{\KS}
  \mbox{\quad can be decomposed into \quad}
  \downsmash{\odv{P}{\deri'}{
    \odv{Q'_1}{\deri''_1}{Q_1}{\KSnc}
    \vlor
    \cdots
    \vlor
    \odv{Q'_n}{\deri''_n}{Q_n}{\KSnc}
    }{\KSc}}
  $$
  for some formulas $Q'_1,\ldots, Q'_n$.
\end{restatable}
\begin{proof}
  This is a simple rule permutation that is very similar to the proof of \Cref{lem:decomposition} (but simpler, as there are less cases).
\end{proof}

\begin{remark}
  It might come as a surprise that \Cref{lem:decompositionKS} has not been stated for $\KS$ in the literature so far. But in fact, as it stands, \Cref{lem:decompositionKS} does not hold for the variant of $\KS$ used in \cite{brunnler:tiu:01,bru:gug:PC-DI,AGR:cls2017} because our $\eqU$-rule is part of it. And indeed, \Cref{lem:decomposition} and \Cref{lem:decompositionKS} are the reason for separating $\eqD$ and $\eqU$ in this paper. However, a weaker variant of \Cref{lem:decompositionKS} is implicitly present in the work on combinatorial proofs~\cite{hughes:pws,hughes:invar,ralph:str:tableaux19}, which make use of the fact that whenever $A$ is a classical tautology, then there is a formula $A'$ such that $\vltt\proves[\KSc+\eqU]A'$ and $A'\proves[\KSnc]A$ (i.e., if $P=\vltt$ then we can have $\eqU$ is present in \Cref{lem:decompositionKS}).
  %
\end{remark}

\begin{restatable}[Flipping for $\SKS$]{theorem}{thmFlippingSKS}\label{thm:flippingSKS}
  Let $H,F,J,K$ be formulas.\\
  If there is a derivation
  $\od{\odd{\odd{\odh{H}}{\derib}{
        J}{\SKSU}}{\deri}{
      F\vlor K}{\SKSD}}
  $
  then there is a formula $I$ and a derivation
    $\od{\odd{\odd{\odh{\cneg F\vlan H}}{\deric}{
        I}{\SKSU}}{\derid}{
      K}{\SKSD}}
  $.
\end{restatable}

\begin{proof}
  With \Cref{lem:splittingKSc} we can prove flipping for $\SKSc$, like we did in \Cref{lem:coreflipping} for $\SLScp$. Together with \Cref{lem:decompositionKS}, the proof of flipping for $\SKS$ is then literally the same as for $\SLS'$ in \Cref{thm:flipping}.
\end{proof}

\begin{corollary}[Interpolation]\label{cor:interpolSKS}
  Every derivation \nosmash{$\odV{A}{\deri}{B}{\SKS}$} can be decomposed into
  $\nosmash{
    \od{\odd{\odd{\odh{A}}{\deric}{
        I}{\SKSU}}{\derid}{
      B}{\SKSD}}
  }$
  for some formula~$I$. Furthermore, this formula $I$ is a Lyndon interpolant for $A\vlim B$, i.e., every atom occuring positively (resp.~negatively) in $I$ also occurs   positively (resp.~negatively) in $A$ and in~$B$.
\end{corollary}

\begin{proof}
  This follows in exactly the same way as in \Cref{cor:interpolSLS}.
\end{proof}

\begin{remark}
  It is quite astonishing that something simple and elegant as \Cref{cor:interpolSKS} and \Cref{thm:flippingSKS} has not been observed for classical logic before. More so, as deep inference proof systems like $\SKS$ have been around for more than 25 years~\cite{brunnler:tiu:01}. Now that we have these results, we can conjecture that (i) simpler proofs will be found eventually, and that (ii) also for classical logic, Step~3 in the proof of \Cref{cor:interpolSLS} and the ``down-flipping'' in \eqref{eq:interflip} are inverse operations. 
\end{remark}

\section{Modal Logic}
\label{sec:modal}

In this section we show how the work on linear logic allows us to extend the results of the previous section from classical logic to the modal logics $\K$, $\KT$, $\Kfour$, and $\Sfour$.

The formulas of modal logic are generated from the grammar:
$$
\Fms\grammareq \Atms\mid\cneg\Atms\mid\vltt\mid\vlff\mid
\Fms\vlan\Fms\mid\Fms\vlor\Fms\mid\vlbox\Fms\mid\vldia\Fms
$$
which is classical logic extended by the two modalities $\vlbox$ (called \defn{box}) and $\vldia$ (called \defn{diamond}), that are De~Morgan dual to each other:
$
\wcneg{\vlbox A}=\vldia\cneg A
$ and $
\wcneg{\vldia A}=\vlbox\cneg A
$.
We need the additional inference rules shown in~\Cref{fig:modalrules}, and we define the following proof systems:
\begin{equation}
  \label{eq:modalsys}
  \hfil
  \begin{array}{@{\hskip-1em}l@{\;=\;}l@{\qquad\qquad}l@{\;=\;}l}
    \KSkc&\KSc\cup\set{\eD,\pD,\ppD}
    &
    \SKSkc&\KSkc\cup\set{\eU,\pU,\ppU}
    \\
    \KSknc&\KSnc
    &
    \SKSknc&\SKSnc
    \\
    \KSktnc&\KSnc\cup\set{\tD}
    &
    \SKSktnc&\SKSnc\cup\set{\tU}
    \\
    \KSkfnc&\KSnc\cup\set{\gD}
    &
    \SKSkfnc&\SKSnc\cup\set{\gU}
    \\
    \KSsfnc&\KSnc\cup\set{\tD,\gD}
    &
    \SKSsfnc&\SKSnc\cup\set{\tU,\gU}
    \\
    \KSy&\KSkc\cup\KSync 
    &
    \SKSy&\SKSkc\cup\SKSync 
  \end{array}
\end{equation}
where the last line is for every $\sysY\in\set{\K,\KT,\Kfour,\Sfour}$.\footnote{To be precise, we should remove the rules $\ppD$ and $\ppU$ from $\KSkc$ and $\SKSkc$ (because they are not needed for completeness) and introduce $\KSkc'$ and $\SKSkc'$ as we did for linear logic. However, there are already quite a lot of different systems, and there is no ``standard'' for modal logic in deep inference in the literature, so we decided against this additional confusion.}

\begin{figure}[!t]
  \def\mskipa{-1.5ex}
  \newcommand{\mcol}[2][c@{\quad}]{\multicolumn{2}{#1}{#2}}
  \vskip-\baselineskip
  $$
  \begin{array}{c@{\qquad}c@{\qquad}c@{\quad}|@{\quad}c@{\qquad}c@{\qquad}c} 
    \vlinf{\eD}{}{\vlbox\vltt}{\vltt}
    &
    \vlinf{\pD}{}{\vlbox A\vlor \vldia B}{\vlbox(A\vlor B)}
    &
    \vlinf{\ppD}{}{\vldia A\vlor \vldia B}{\vldia(A\vlor B)}
    &
    \vlinf{\ppU}{}{\vlbox(A\vlan B)}{\vlbox A\vlan \vlbox B}
    &
    \vlinf{\pU}{}{\vldia(A\vlan B)}{\vldia A\vlan \vlbox B}
    &
    \vlinf{\eU}{}{\vlff}{\vldia\vlff}
    \\&&\\[\mskipa]
    &
    \vlinf{\tD}{}{\vldia A}{A}
    &
    \vlinf{\gD}{}{\vldia A}{\vldia\vldia A}
    &
    \vlinf{\gU}{}{\vlbox\vlbox A}{\vlbox A}
    &
    \vlinf{\tU}{}{A}{\vlbox A}
    \\[-2ex]
  \end{array}
  $$
  \caption{Additional rules for systems $\SKSk$, $\SKSkt$, $\SKSkf$, $\SKSsf$, deep inference proof systems for modal logics $\K$, $\KT$, $\Kfour$, $\Sfour$, respectively. Left: Down-fragment, Right: Up-fragment. First line: core-fragment. Second line: non-core fragment.}
  \label{fig:modalrules}
\end{figure}

Again, all results that we presented in Sections~\ref{sec:prelim}--\ref{sec:classical} do also hold for all $\KSy$ and $\SKSy$. But surprisingly, there are no standard deep inference systems for modal logic in the literature.\footnote{The only work we could find are~\cite{hein:stewart:05,stewart:stouppa:05}, where rules $\tD$ and $\tU$ are inversed.} This means that even for soundness and completeness we cannot refer to published work, although the proof is literally the same as for linear logic. 

\begin{restatable}{theorem}{thmSoundKSk}\label{thm:soundnessKSk}
  For every $\sysY\in\set{\K,\KT,\Kfour,\Sfour}$ we have that every proof system $\sysX$ with $\KSy\setminus\set{\ppD}\subseteq\sysX\subseteq\SKSy$ is sound and complete for the modal logic~$\sysY$.
\end{restatable}

\begin{proof}
  First, for every inference rule $\vlinf{\rr}{}GF$ in $\SKSy$, we have that the implication $F\vlim G$ is a theorem of the modal logic $\sysY$. This entails soundness. For completeness, we can pick a cut-free sequent calculus for $\sysY$ (e.g.~\cite{troelstra:schwichtenberg:00,Wansing:02,poggiolesi:11}), 
  and show that each rule is derivable in $\KSy\setminus\set{\ppD}$.
  More details are in \Cref{app:modal}.
\end{proof}

All other Propositions and Lemmas of \Cref{sec:prelim} also follow trivially for the proof systems defined in~\eqref{eq:modalsys}. In particular, we also have \Cref{prop:DI} and cut elimination for all four modal logics discussed in this paper, and the proofs are the same as for linear logic.

\begin{restatable}{proposition}{propAIKSk}\label{prop:DI-KSk}
  For every $\sysY\in\set{\K,\KT,\Kfour,\Sfour}$ we have the following:
  \begin{enumerate}
  \item\label{prk:ai}
    The rule $\iD$ is derivable in $\KSkc$, and dually, the rule $\iU$ is derivable in $\SKSkcU$.
  \item \label{prk:updown}
    Every rule $\rU\in\SKSyU$ is derivable in $\set{\rD,\iD,\sD,\iU,\eqD,\eqU}$.
  \item\label{prk:flip1}
    For every $\sys\in\set{\SKSkc,\SKSy}$, we have $A\proves[\sys]B$ iff\/ $\proves[\sys]\cneg A\vlor B$. 
  \item\label{prk:flip2}
    For
    every $\sys\in\set{\SKSkc,\SKSy,\SKSync}$, we have $A\proves[\sysD]B$ iff\/ $\cneg B\proves[\sysU]\cneg A$.
  \item\label{prk:cutup}
  For every $\sys\in\set{\KSkc,\KSy}$, we have $\proves[\sysSX]A$ iff\/ $\proves[\sys\cup\set{\iU}]A$.\\[-2ex]
  \end{enumerate}
\end{restatable}

\begin{corollary}[Cut Elimination]\label{thm:cutelimSKSk}
  For all $\sysY\in\set{\K,\KT,\Kfour,\Sfour}$, if\/ $\proves[\SKSy] A$ then also\/ $\proves[\KSy] A$.
\end{corollary}

And at this point it should come at no surprise that the generalized splitting lemma and the decomposition lemma also hold for the four modal logics mentioned above. The notion of m-context is the same as for classical logic, i.e., in the case of modal logic, an \defn{m-context} is a context where the hole $\conhole$ does not occur inside the scope of a $\vlbox$ or $\vldia$.

\begin{lemma}[Generalized Splitting for $\KSkc$]\label{lem:splittingKSkc}
  Let $J,K,A,B$ be formulas and $a$ be an atom or negated atom. The three statements of \Cref{lem:splittingKSc} do also hold for $\KSkc$. Furthermore: 
  \begin{enumerate}\setcounter{enumi}{3}
  \item\label{s:box} If $\nnodv{J}{\deri}{\vlbox A\vlor K}{\KSkc}$ then there is either (i) an m-context $\Chole$ and a formula $J_A$ such that
    $$
    \odv{J}{\deri_J}{\Ccons{\vlbox J_A}}{\SKSkcU}
    \quand
    \odv{\Ccons{\vlff}}{\deri_K}{K}{\KSkc}
    \quand
    \odv{J_A}{\deri_A}{A}{\KSkc}
    $$
    or (ii) an m-context $\Chole$ and formulas $J_A$ and $K_A$ such that
    $$
    \odv{J}{\deri_J}{\Ccons{\vlbox J_A}}{\SKScU}
    \quand
    \odv{\Ccons{\vldia K_A}}{\deri_K}{K}{\LScp}
    \quand
    \odv{J_A}{\deri_A}{A\vlor K_A}{\LScp}
    $$
  \item\label{s:dia} If $\nnodv{J}{\deri}{\vldia A\vlor K}{\KSkc}$ then there is either (i) an m-context $\Chole$ and a formula $J_A$ such that
    $$
    \odv{J}{\deri_J}{\Ccons{\vldia J_A}}{\SKSkcU}
    \quand
    \odv{\Ccons{\vlff}}{\deri_K}{K}{\KSkc}
    \quand
    \odv{J_A}{\deri_A}{A}{\KSkc}
    $$
    or (ii) an m-context $\Chole$ and formulas $J_A$ and $K_A$ such that
    $$
    \odv{J}{\deri_J}{\Ccons{\vlca J_A}}{\SKSkcU}
    \quand
    \odv{\Ccons{\vlca K_A}}{\deri_K}{K}{\KSkc}
    \quand
    \odv{J_A}{\deri_A}{A\vlor K_A}{\KSkc}
    $$
    where $\vlca$ is either $\vlbox$ or $\vldia$.
  \end{enumerate}
\end{lemma}

\begin{proof}
  The proof is the same as in the previous sections. We can extend the translation in~\eqref{eq:transl} by $\vloc\transl\vlbox$ and $\vlwn\transl\vldia$. Then the two statements above are exactly the same as \Cref{lem:splitting}.\ref{s:oc} and \ref{lem:splitting}.\ref{s:wn}.
\end{proof}

\begin{lemma}[Decomposition for Modal Logics]\label{lem:decompositionModal}
  Let $\sysY\in\set{\K,\KT,\Kfour,\Sfour}$, let $n\ge0$, and let $P,Q_1,\ldots, Q_n$ be formulas. Then\vadjust{\vskip-2ex}
  $$
  \odV{P}{\deri}{Q_1\vlor Q_2,\vlor\cdots\vlor Q_n}{\KSy}
  \mbox{\quad can be decomposed into \quad}
  \downsmash{\odv{P}{\deri'}{
    \odv{Q'_1}{\deri''_1}{Q_1}{\KSync}
    \vlor
    \cdots
    \vlor
    \odv{Q'_n}{\deri''_n}{Q_n}{\KSync}
    }{\KSkc}}
  $$
  for some formulas $Q'_1,\ldots, Q'_n$.
\end{lemma}

\begin{proof}
  The proof is the same as for \Cref{lem:decomposition}.
\end{proof}
\begin{restatable}[Flipping for Modal Logic]{theorem}{thmFlippingModal}\label{thm:flippingModal}
  Let $\sysY\in\set{\K,\KT,\Kfour,\Sfour}$ and let $H,F,J,K$ be formulas.\\[-1.5ex]
  If there is a derivation
  \nosmash{\scalebox{1}{$\od{\odd{\odd{\odh{H}}{\derib}{
        J}{\SKSyU}}{\deri}{
      F\vlor K}{\SKSyD}}
  $}}
  then there is a formula $I$ and a derivation
  \nosmash{\scalebox{1}{$\od{\odd{\odd{\odh{\lneg F\vlan H}}{\deric}{
        I}{\SKSyU}}{\derid}{
      K}{\SKSyD}}
  $}}.
\end{restatable}

\begin{proof}
  First, with the help of the Splitting \Cref{lem:splittingKSkc}, we can prove flipping for $\SKSkc$. Then, together with \Cref{lem:decompositionModal}, the proof is exactly the same as for \Cref{thm:flipping}. 
\end{proof}

\begin{corollary}
  Let $\sysY\in\set{\K,\KT,\Kfour,\Sfour}$, and let $A$ and $B$ be formulas. Then every derivation \nosmash{$\odV{A}{\deri}{B}{\SKSy}$} can be decomposed into
  $\nosmash{
    \od{\odd{\odd{\odh{A}}{\deric}{
        I}{\SKSyU}}{\derid}{
      B}{\SKSyD}}
  }$
  for some formula~$I$. Furthermore, this formula $I$ is a Lyndon interpolant for $A\vlim B$, i.e., every atom occuring positively (resp.~negatively) in $I$ also occurs   positively (resp.~negatively) in $A$ and in~$B$.
\end{corollary}

\begin{proof}
  The proof is the same as for \Cref{cor:interpolSLS}.
\end{proof}

\section{Conclusion and Future Work}
\label{sec:conclusion}


In \cite{girard:87:a}, Girard argues that the lack of modularity is one of the main technical limitations in current proof theory. Even though his statement concerns mainly cut elimination, where a minor change in the proof system means that one has to redo the whole cut elimination proof from scratch, this is also true for all traditional proof theoretical methods of proving interpolation. And even after four decades, not much progress has been made. 

However, with the use of deep inference in this paper, the proof of interpolation becomes \emph{modular}. The key ingredient is the separation of the proof system into  \emph{core} and \emph{non-core}, such that (i) every derivation is decomposed into a core-part and a non-core part, and (ii) the core-fragment of the proof system satisfies a generalized splitting lemma.

Traditional sequent systems do not allow this separation, as a typical sequent calculus rule is a mixture of core and non-core parts, as for example the additive conjunction rule:
$$
\vliinf{\vlwi}{}{\sqn{A\vlwi B,\Gamma}}{\sqn{A,\Gamma}}{\sqn{B,\Gamma}}
\qquato
\vlinf{\dD}{}{(A\vlwi B)\vlpa(C\vlpl D)}{(A\vlpa C)\vlwi(B\vlpa D)}
\quad
+
\quad
\vlinf{\cD}{}{A}{A\vlpl A}
$$
This raises the question whether this also possible for other logics, with or without an existing sequent calculus representation.

\subparagraph*{Other Modal Logics}
Deep inference is not yet very well developed for modal logics. For this reason we considered in this paper only the modal logics $\K,\KT,\Kfour,\Sfour$. But we conjecture that our method can also be applied to other modal logics, for which there is no sequent calculus. However, for investigating this, we first need to be able to prove cut elimination inside the deep inference system for these modal logics. 

\subparagraph*{First-Order Logic}
Craig's original proof~\cite{craig:57:interpol} was mainly concerned with first-order logic whereas we are here only concerned with propositional logics. So, in some sense our work here is orthogonal to Craig's proof. But there are some striking similarities. For example, Craig uses so-called \emph{linear inferences}~\cite{craig:57:linear}, that can be seen as a precursor to deep inference. So, it is only natural to ask if we can extend our approach to first order logic. There are already several different deep inference systems for first-order logic~\cite{brunnler:phd,brunnler:06:herbrand,ralph:phd,HSW:focp}, but it is not clear which of them (if any) allows for a separation between core and non-core.

\subparagraph*{Relation to Combinatorial Proofs}
In the case of classical logic (see~\Cref{sec:classical}), the separation into \emph{core} and \emph{non-core} coincides with the separation into \emph{linear} and \emph{resource management} part in combinatorial proofs~\cite{hughes:pws,str:FSCD17,omi:str:wollic22,ralph:str:tableaux19}. Something similar can be observed for combinatorial proofs for modal logics~\cite{acc:str:modal}. This raises several questions. For example, could combinatorial proofs for first-order logic~\cite{HSW:focp} help to find the right core/non-core systems to apply our method to first-order logic? Or could our core/non-core separation of system $\SLSp$ help to design a new form of combinatorial proofs for linear logic? Or is there even a closer relationship  between the two properties of having combinatorial proofs and having interpolation for a logic? Or could there be a transfer in the other direction, i.e., could our separation into \emph{core} and \emph{non-core} lead to a new kind of combinatorial proofs or proof nets for linear logic? 

\subparagraph*{General Criteria for Interpolation}
There are many logics that do not enjoy interpolation, for example the modal logic $\Kfourthree$~\cite{maksimova:79}. This opens two possible research directions: (i) can we use deep inference to find general criteria, when a logic has interpolation? And (ii) can we use our method to investigate the \emph{interpolant existence problem}~\cite{KWZ:InterpolSeparate}.

\subparagraph*{Intuitionistic Logic}
The proof of the inperpolation theorem that we presented in this paper heavily relies on the \emph{up-down-duality} of the proof systems. But it does not necessarily rely on the \emph{De~Morgan duality} of the connectives. This raises the question whether the same machinery can also be used for intuitionistic logic, which does not have the De~Morgan duality, but still exhibits an up-down duality in some proof systems~\cite{gue:str:14:CSL-LICS}.

\subparagraph{} \emph{Written by humans for humans.}



\bibliography{lutzrefs.bib}

\begin{thebibliography}{10}

\bibitem{acc:str:modal}
Matteo Acclavio and Lutz Stra{\ss}burger.
\newblock On combinatorial proofs for modal logic.
\newblock In Serenella Cerrito and Andrei Popescu, editors, {\em Automated
  Reasoning with Analytic Tableaux and Related Methods - 28th International
  Conference, {TABLEAUX} 2019, London, UK, September 3-5, 2019, Proceedings},
  volume 11714 of {\em Lecture Notes in Computer Science}, pages 223--240.
  Springer, 2019.
\newblock \href {https://doi.org/10.1007/978-3-030-29026-9\_13}
  {\path{doi:10.1007/978-3-030-29026-9\_13}}.

\bibitem{ASZ:BVcat}
Matteo Acclavio, Lutz Stra{\ss}burger, and Vladimir Zamdzhiev.
\newblock Proof identity and categorical models of {BV}.
\newblock {\em CoRR}, abs/2604.25501, 2026.
\newblock to appear in FSCD 2026.
\newblock URL: \url{https://doi.org/10.48550/arXiv.2604.25501}, \href
  {https://arxiv.org/abs/2604.25501} {\path{arXiv:2604.25501}}, \href
  {https://doi.org/10.48550/ARXIV.2604.25501}
  {\path{doi:10.48550/ARXIV.2604.25501}}.

\bibitem{AGR:cls2017}
Andrea Aler~Tubella, Alessio Guglielmi, and Benjamin Ralph.
\newblock {Removing Cycles from Proofs}.
\newblock In Valentin Goranko and Mads Dam, editors, {\em 26th EACSL Annual
  Conference on Computer Science Logic (CSL 2017)}, volume~82 of {\em Leibniz
  International Proceedings in Informatics (LIPIcs)}, pages 9:1--9:17,
  Dagstuhl, Germany, 2017. Schloss Dagstuhl -- Leibniz-Zentrum f{\"u}r
  Informatik.
\newblock URL:
  \url{https://drops.dagstuhl.de/entities/document/10.4230/LIPIcs.CSL.2017.9},
  \href {https://doi.org/10.4230/LIPIcs.CSL.2017.9}
  {\path{doi:10.4230/LIPIcs.CSL.2017.9}}.

\bibitem{benedikt:26}
Michael Benedikt.
\newblock Interpolation and query rewriting, 2026.
\newblock URL: \url{https://arxiv.org/abs/2606.15737}, \href
  {https://arxiv.org/abs/2606.15737} {\path{arXiv:2606.15737}}.

\bibitem{BCI:interpolation}
Nick Bezhanishvili, Balder ten Cate, and Rosalie Iemhoff.
\newblock Six proofs of interpolation for the modal logic {K}, 2025.
\newblock URL: \url{https://arxiv.org/abs/2510.16398}, \href
  {https://arxiv.org/abs/2510.16398} {\path{arXiv:2510.16398}}.

\bibitem{brunnler:phd}
Kai Br{\"u}nnler.
\newblock {\em Deep Inference and Symmetry for Classical Proofs}.
\newblock PhD thesis, Tech\-ni\-sche Uni\-ver\-si\-t{\"a}t Dres\-den, 2003.

\bibitem{brunnler:06:herbrand}
Kai Br{\"u}nnler.
\newblock Cut elimination inside a deep inference system for classical
  predicate logic.
\newblock {\em Studia Logica}, 82(1):51--71, 2006.

\bibitem{brunnler:06:locality}
Kai Br{\"u}nnler.
\newblock Locality for classical logic.
\newblock {\em Notre Dame Journal of Formal Logic}, 47(4):557--580, 2006.
\newblock URL: \url{http://www.iam.unibe.ch/~kai/Papers/LocalityClassical.pdf}.

\bibitem{brunnler:tiu:01}
Kai Br{\"u}nnler and Alwen~Fernanto Tiu.
\newblock A local system for classical logic.
\newblock In R.~Nieuwenhuis and A.~Voronkov, editors, {\em {LPAR} 2001}, volume
  2250 of {\em LNAI}, pages 347--361. Springer, 2001.

\bibitem{bru:gug:PC-DI}
Paola Bruscoli and Alessio Guglielmi.
\newblock On the proof complexity of deep inference.
\newblock {\em ACM Transactions on Computational Logic}, 10(2):1--34, 2009.
\newblock Article 14.

\bibitem{CGRS:foccos-long}
Kaustuv Chaudhuri, Nicolas Guenot, Mikheil Rukhaia, and Lutz Strassburger.
\newblock {Splitting and Focusing for Full Propositional Linear Logic in the
  Calculus of Structures}.
\newblock preprint, 2017.
\newblock URL: \url{https://inria.hal.science/hal-05493721}.

\bibitem{CGS:foccos}
Kaustuv Chaudhuri, Nicolas Guenot, and Lutz Stra{\ss}burger.
\newblock The focused calculus of structures.
\newblock In Marc Bezem, editor, {\em CSL'11}, volume~12 of {\em LIPIcs}, pages
  159--173. Schloss Dagstuhl -- Leibniz-Zentrum fuer Informatik, 2011.

\bibitem{craig:57:linear}
William Craig.
\newblock Linear reasoning: A new form of the {H}erbrand-{G}entzen theorem.
\newblock {\em The Journal of Symbolic Logic}, 22:250--268, 1957.

\bibitem{craig:57:interpol}
William Craig.
\newblock Three uses of the {H}erbrand-{G}entzen theorem in relating model
  theory and proof theory.
\newblock {\em The Journal of Symbolic Logic}, 22:269--285, 1957.

\bibitem{fitting:kuznets:15}
Melvin Fitting and Roman Kuznets.
\newblock Modal interpolation via nested sequents.
\newblock {\em Ann. Pure Appl. Logic}, 166(3):274--305, 2015.

\bibitem{girard:87}
Jean-Yves Girard.
\newblock Linear logic.
\newblock {\em Theoretical Computer Science}, 50:1--102, 1987.

\bibitem{girard:87:a}
Jean-Yves Girard.
\newblock {\em Proof Theory and Logical Complexity, Volume I}, volume~1 of {\em
  Studies in Proof Theory}.
\newblock Bibliopolis, edizioni di filosofia e scienze, 1987.

\bibitem{gue:str:14:CSL-LICS}
Nicolas Guenot and Lutz Stra{\ss}burger.
\newblock Symmetric normalisation for intuitionistic logic.
\newblock In Tom Henzinger and Dale Miller, editors, {\em Proceedings of the
  Twenty-Third EACSL Annual Conference on Computer Science Logic ({CSL}) and
  the Twenty-Ninth Annual {ACM/IEEE} Symposium on Logic in Computer Science
  ({LICS})}, 2014.
\newblock \href {https://doi.org/10.1145/2603088.2603160}
  {\path{doi:10.1145/2603088.2603160}}.

\bibitem{gug:SIS}
Alessio Guglielmi.
\newblock A system of interaction and structure.
\newblock {\em ACM Transactions on Computational Logic}, 8(1):1--64, 2007.

\bibitem{GGP:open}
Alessio Guglielmi, Tom Gundersen, and Michel Parigot.
\newblock {A Proof Calculus Which Reduces Syntactic Bureaucracy}.
\newblock In Christopher Lynch, editor, {\em Proceedings of the 21st
  International Conference on Rewriting Techniques and Applications}, volume~6
  of {\em Leibniz International Proceedings in Informatics (LIPIcs)}, pages
  135--150, Dagstuhl, Germany, 2010. Schloss Dagstuhl -- Leibniz-Zentrum
  f{\"u}r Informatik.
\newblock URL:
  \url{https://drops.dagstuhl.de/entities/document/10.4230/LIPIcs.RTA.2010.135},
  \href {https://doi.org/10.4230/LIPIcs.RTA.2010.135}
  {\path{doi:10.4230/LIPIcs.RTA.2010.135}}.

\bibitem{gug:str:01}
Alessio Guglielmi and Lutz Stra{\ss}burger.
\newblock Non-commutativity and {MELL} in the calculus of structures.
\newblock In Laurent Fribourg, editor, {\em Computer Science Logic, {CSL}
  2001}, volume 2142 of {\em LNCS}, pages 54--68. Springer-Verlag, 2001.

\bibitem{hein:stewart:05}
Robert Hein and Charles Stewart.
\newblock Purity through unravelling.
\newblock In Paola Bruscoli, Fran{\c c}ois Lamarche, and Charles Stewart,
  editors, {\em Structures and Deduction}, pages 126--143. Technische
  Universit{\"a}t Dresden, 2005.
\newblock ICALP Workshop. ISSN 1430-211X.
\newblock URL: \url{http://bitschnitzer.de/rh+cas--sd05.pdf}.

\bibitem{hughes:pws}
Dominic Hughes.
\newblock Proofs {W}ithout {S}yntax.
\newblock {\em Annals of Mathematics}, 164(3):1065--1076, 2006.

\bibitem{hughes:invar}
Dominic Hughes.
\newblock Towards {H}ilbert's 24\({}^{\mbox{th}}\) problem: Combinatorial proof
  invariants: (preliminary version).
\newblock {\em Electr. Notes Theor. Comput. Sci.}, 165:37--63, 2006.

\bibitem{HSW:focp}
Dominic J.~D. Hughes, Lutz Stra{\ss}burger, and Jui{-}Hsuan Wu.
\newblock Combinatorial proofs and decomposition theorems for first-order
  logic.
\newblock In {\em 36th Annual {ACM/IEEE} Symposium on Logic in Computer
  Science, {LICS} 2021, Rome, Italy, June 29 - July 2, 2021}, pages 1--13.
  {IEEE}, 2021.
\newblock \href {https://doi.org/10.1109/LICS52264.2021.9470579}
  {\path{doi:10.1109/LICS52264.2021.9470579}}.

\bibitem{JKK:25:interpol}
Jean~Christoph Jung, Patrick Koopmann, and Matthias Knorr.
\newblock Interpolation in knowledge representation, 2025.
\newblock URL: \url{https://arxiv.org/abs/2512.08833}, \href
  {https://arxiv.org/abs/2512.08833} {\path{arXiv:2512.08833}}.

\bibitem{KWZ:InterpolSeparate}
Agi Kurucz, Frank Wolter, and Michael Zakharyaschev.
\newblock From interpolating formulas to separating languages and back again,
  2025.
\newblock URL: \url{https://arxiv.org/abs/2508.12805}, \href
  {https://arxiv.org/abs/2508.12805} {\path{arXiv:2508.12805}}.

\bibitem{kuznets:lellmann:21}
Roman Kuznets and Björn Lellmann.
\newblock Interpolation for intermediate logics via injective nested sequents.
\newblock {\em Journal of Logic and Computation}, 31(3):797--831, 04 2021.
\newblock \href {https://doi.org/10.1093/logcom/exab015}
  {\path{doi:10.1093/logcom/exab015}}.

\bibitem{lyndon:59}
Roger~C. Lyndon.
\newblock An interpolation theorem in the predicate calculus.
\newblock {\em Pacific Journal of Mathematics}, 9:129--142, 1959.
\newblock URL: \url{https://doi.org/10.2140/pjm.1959.9.129}.

\bibitem{maehara:interpolation}
Shoji Maehara.
\newblock On the interpolation theorem of {C}raig.
\newblock {\em S{\^u}gaku}, 12(4), 1960.

\bibitem{maksimova:79}
L.L. Maksimova.
\newblock Interpolation theorems in modal logics and amalgamable varieties of
  topological boolean algebras.
\newblock {\em Algebra and Logic}, 18, 1979.
\newblock \href {https://doi.org/doi.org/10.1007/BF01673502}
  {\path{doi:doi.org/10.1007/BF01673502}}.

\bibitem{omi:str:wollic22}
Giti Omidvar and Lutz Stra{\ss}burger.
\newblock Combinatorial flows as bicolored atomic flows.
\newblock In Agata Ciabattoni, Elaine Pimentel, and Ruy J. G.~B. de~Queiroz,
  editors, {\em Logic, Language, Information, and Computation - 28th
  International Workshop, WoLLIC 2022, Ia{\c{s}}i, Romania, September 20-23,
  2022, Proceedings}, volume 13468 of {\em Lecture Notes in Computer Science},
  pages 141--157. Springer, 2022.
\newblock \href {https://doi.org/10.1007/978-3-031-15298-6\_9}
  {\path{doi:10.1007/978-3-031-15298-6\_9}}.

\bibitem{poggiolesi:11}
Francesca Poggiolesi.
\newblock {\em Gentzen Calculi for Modal Propositional Logic}.
\newblock Springer, 2011.
\newblock \href {https://doi.org/doi.org/10.1007/978-90-481-9670-8}
  {\path{doi:doi.org/10.1007/978-90-481-9670-8}}.

\bibitem{ralph:phd}
Benjamin Ralph.
\newblock {\em Modular Normalisation of Classical Proofs}.
\newblock PhD thesis, University of Bath, 2019.

\bibitem{ralph:str:tableaux19}
Benjamin Ralph and Lutz Stra{\ss}burger.
\newblock Towards a combinatorial proof theory.
\newblock In Serenella Cerrito and Andrei Popescu, editors, {\em Automated
  Reasoning with Analytic Tableaux and Related Methods - 28th International
  Conference, {TABLEAUX} 2019, London, UK, September 3-5, 2019, Proceedings},
  volume 11714 of {\em Lecture Notes in Computer Science}, pages 259--276.
  Springer, 2019.
\newblock \href {https://doi.org/10.1007/978-3-030-29026-9\_15}
  {\path{doi:10.1007/978-3-030-29026-9\_15}}.

\bibitem{roorda:94}
Dirk Roorda.
\newblock Interpolation in fragments of classical linear logic.
\newblock {\em Journal of Symbolic Logic}, 59(2):419--444, 1994.
\newblock \href {https://doi.org/10.2307/2275398} {\path{doi:10.2307/2275398}}.

\bibitem{rummer:26}
Philipp Rümmer.
\newblock Craig interpolation in program verification, 2026.
\newblock URL: \url{https://arxiv.org/abs/2602.08532}, \href
  {https://arxiv.org/abs/2602.08532} {\path{arXiv:2602.08532}}.

\bibitem{saurin:25}
Alexis Saurin.
\newblock {Interpolation as Cut-Introduction: On the Computational Content of
  Craig-Lyndon Interpolation}.
\newblock In Maribel Fern\'{a}ndez, editor, {\em 10th International Conference
  on Formal Structures for Computation and Deduction ({FSCD} 2025)}, volume 337
  of {\em Leibniz International Proceedings in Informatics (LIPIcs)}, pages
  32:1--32:21, Dagstuhl, Germany, 2025. Schloss Dagstuhl -- Leibniz-Zentrum
  f{\"u}r Informatik.
\newblock URL:
  \url{https://drops.dagstuhl.de/entities/document/10.4230/LIPIcs.FSCD.2025.32},
  \href {https://doi.org/10.4230/LIPIcs.FSCD.2025.32}
  {\path{doi:10.4230/LIPIcs.FSCD.2025.32}}.

\bibitem{stewart:stouppa:05}
Charles Stewart and Phiniki Stouppa.
\newblock A systematic proof theory for several modal logics.
\newblock In R.~A. Schmidt, I.~Pratt-Hartmann, M.~Reynolds, and H.~Wansing,
  editors, {\em Advances in Modal Logic, Volume 5}, pages 309--333. King's
  College Publications, 2005.

\bibitem{str:02}
Lutz {Stra\ss burger}.
\newblock A local system for linear logic.
\newblock In Matthias Baaz and Andrei Voronkov, editors, {\em Logic for
  Programming, Artificial Intelligence, and Reasoning, LPAR 2002}, volume 2514
  of {\em LNAI}, pages 388--402. Springer-Verlag, 2002.
\newblock URL:
  \url{http://www.lix.polytechnique.fr/~lutz/papers/LocSysLinLog.pdf}.

\bibitem{dissvonlutz}
Lutz Stra{\ss}burger.
\newblock {\em Linear Logic and Noncommutativity in the Calculus of
  Structures}.
\newblock PhD thesis, Tech\-ni\-sche Uni\-ver\-si\-t{\"a}t Dres\-den, 2003.

\bibitem{str:MELL}
Lutz Stra{\ss}burger.
\newblock {MELL} in the {C}alculus of {S}tructures.
\newblock {\em Theoretical Computer Science}, 309(1--3):213--285, 2003.

\bibitem{str:FSCD17}
Lutz Stra{\ss}burger.
\newblock Combinatorial flows and their normalisation.
\newblock In Dale Miller, editor, {\em 2nd International Conference on Formal
  Structures for Computation and Deduction, {FSCD} 2017, September 3-9, 2017,
  Oxford, {UK}}, volume~84 of {\em LIPIcs}, pages 31:1--31:17. Schloss Dagstuhl
  - Leibniz-Zentrum fuer Informatik, 2017.

\bibitem{SIS-IV}
Lutz Stra{\ss}burger and Alessio Guglielmi.
\newblock A system of interaction and structure {IV}: The exponentials and
  decomposition.
\newblock {\em ACM Trans. Comput. Log.}, 12(4):23, 2011.

\bibitem{troelstra:schwichtenberg:00}
Anne~Sjerp Troelstra and Helmut Schwichtenberg.
\newblock {\em Basic Proof Theory}.
\newblock Cambridge University Press, second edition, 2000.

\bibitem{ATS:esslli2019}
Andrea~Aler Tubella and Lutz Stra{\ss}burger.
\newblock Introduction to deep inference.
\newblock Lecture notes for ESSLLI'19, 2019.
\newblock URL: \url{https://hal.inria.fr/hal-02390267}.

\bibitem{GJK:interpolation}
Iris van~der Giessen, Raheleh Jalali, and Roman Kuznets.
\newblock Interpolation in proof theory, 2026.
\newblock URL: \url{https://arxiv.org/abs/2602.16318}, \href
  {https://arxiv.org/abs/2602.16318} {\path{arXiv:2602.16318}}.

\bibitem{Wansing:02}
Heinrich Wansing.
\newblock Sequent systems for modal logics.
\newblock In D.~M. Gabbay and F.~Guenthner, editors, {\em Handbook of
  Philosophical Logic: Volume 8}, pages 61--145. Springer Netherlands,
  Dordrecht, 2002.
\newblock \href {https://doi.org/10.1007/978-94-010-0387-2_2}
  {\path{doi:10.1007/978-94-010-0387-2_2}}.

\end{thebibliography}

\clearpage

\appendix

\section{Additional Material for Section~\ref{sec:prelim}}\label{app:prelim}

This appendix is mainly for the convenience of the reader, as the material is already available in the literature (e.g., \cite{dissvonlutz,str:02,str:MELL,ATS:esslli2019,CGS:foccos}). The same is true for \Cref{app:classical}. Only \Cref{app:splitting}, \Cref{app:flipping}, and part of \Cref{app:modal} contain novel material.

\propAI*

\begin{proof}
  \begin{enumerate}
  \item 
  For $\iD$, We proceed by induction on $A$, using the following derivation:
  $$
  \od{\odi{\odi{\odh{\vlone}}{
        \tensD}{
        \odv{\vlone}{\deri_A}{\lneg A\vlpa A}{\LSc}
        \vlte
        \odv{\vlone}{\deri_B}{\lneg B\vlpa B}{\LSc}}{}}{
      \sD}{(\lneg A\vlte \lneg B)\vlpa (A\vlpa B)}{}}
  \qquad
  \od{\odi{\odi{\odh{\vlone}}{
        \withD}{
        \odv{\vlone}{\deri_A}{\lneg A\vlpa A}{\LSc}
        \vlwi
        \odv{\vlone}{\deri_B}{\lneg B\vlpa B}{\LSc}}{}}{
      \sD}{(\lneg A\vlwi \lneg B)\vlpa (A\vlpl B)}{}}
  \qquad
  \od{\odi{\odi{\odh{\vlone}}{
        \eD}{
        \vloc
        \odv{\vlone}{\deri_A}{\lneg A\vlpa A}{\LSc}}{}}{
      \sD}{(\loc\lneg A\vlpa \vlwn A)}{}}
  $$
  The case for $\iU$ is dual.
\item 
  An instance of $\vlinf{\rU}{}BA$ can be derived with the following derivation:
  \begin{equation}
    \label{eq:rU}
    \hfill
    \od{\odi{\odi{\odi{\odh{A}}{
            \eqU}{
            \odn{A}{\eqD}{A\vlpa \vlbot}{}
            \vlte
            \odn{\vlone}{\iD}{\lneg B\vlpa B}{}}{}}{
          \sD}{
          \odn{A\vlte\odn{\lneg B}{\rD}{\lneg A}{}}{\iU}{\vlbot}{}
          \vlpa(\vlbot\vlpa B)}{}}{
        \eqD}{B}{}}
    \hfill
  \end{equation}
\item
  This follows immediately from \Cref{prop:ai} and the two derivations:
  $$
  \odn{\vlone}{\iD}{\lneg A\vlpa\odv{A}{}{B}{\sysX}}{}
  \qquad
  \qquad
  \od{\odi{\odi{\odi{\odh{\vlone}}{
          \eqU}{
          \odn{A}{\eqD}{A\vlpa\vlbot}{}
          \vlte
          \odv{\vlone}{}{\lneg A\vlpa B}{\sysX}}{}}{
        \sD}{
        \odn{A\vlte\lneg A}{\iU}{\vlbot}{}
        \vlpa (\vlbot\vlpa B)}{}}{
      \eqD}{B}{}}
  $$
\item
  This follows trivially by dualizing every inference rule instance in the derivation.
\item 
  This follows immediately from \Cref{prop:ai} and \Cref{prop:updown}:\\
  Given a derivation of $A$ in $\sysX\cup\set{\iU}$, we can with \Cref{prop:ai} obtain a derivation of $A$ in $\sysSX$ because $\SLScU$ is contained in all of them. \\
  Conversely, Given a derivation of $A$ in $\sysSX$, we can by \Cref{prop:updown} obtain a derivation in $\sysX\cup\set{\iU,\eqU}$. By inspecting the derivation in~\eqref{eq:rU}, we see that the only use of $\eqU$ is in the form
  $$
  \vlinf{\eqU}{}{A\vlte\vlone}{A}
  $$
  This can be straightforwardly permuted up, until it is of the form
  $$
  \vlinf{\eqU}{}{\vlone\vlte\vlone}{\vlone}
  $$
  when it becomes an instance of $\tensD$.
  \qedhere
  \end{enumerate}
\end{proof}

\lemSwitch*

\begin{proof}
  We proceed by induction on $\Chole$. The case when $\Chole=\conhole$
  is trivial. The other cases for $\deri$ are:
  $$
  \od{\odi{\odi{\odi{\odh{
            D\vlte\odv{\Cpcons{A\vlpa B}}{\deri'}{A\vlpa\Cpcons{B}}{\set{\eqD,\sD}}}}{
          \eqD}{
          (D\vlpa\vlbot)\vlte(\Cpcons{B}\vlpa A')}{}}{
        \sD}{
        (D\vlte\Cpcons{B})\vlpa (\vlbot\vlpa A')}{}}{
      \eqD}{A\vlpa(D\vlte\Cpcons{B})}{}}
  \hskip4em
  \od{\odi{\odi{\odi{\odh{
            \odv{\Cpcons{A\vlpa B}}{\deri'}{A\vlpa\Cpcons{B}}{\set{\eqD,\sD}}\vlte D}}{
          \eqD}{
          (\Cpcons{B}\vlpa A')\vlte(D\vlpa\vlbot)}{}}{
        \sD}{
        (\Cpcons{B}\vlte D)\vlpa (A'\vlpa\vlbot)}{}}{
      \eqD}{A\vlpa(\Cpcons{B}\vlte D)}{}}
  $$
  $$
  \odn{D\vlpa\odv{\Cpcons{A\vlpa B}}{\deri'}{A\vlpa\Cpcons{B}}{\set{\eqD,\sD}}}{
    \eqD}{A\vlpa(D\vlpa\Cpcons{B})}{}
  \hskip6em
  \odn{\odv{\Cpcons{A\vlpa B}}{\deri'}{A\vlpa\Cpcons{B}}{\set{\eqD,\sD}}\vlpa D}{
    \eqD}{A\vlpa(\Cpcons{B}\vlpa D)}{}
  $$
  The cases for $\derib$ are dual.
\end{proof}

\section{Proofs for Section~\ref{sec:splitting}}\label{app:splitting}

\begin{definition}
  We define the size $\sizeof F$ of a formula $F$ inductively as follows:
  $$
  \sizeof\vlbot=\sizeof\vlone=0
  \qquad\qquad
  \sizeof{a}=\sizeof{\lneg a}=\sizeof\vltop=\sizeof\vlzer=1
  \qquad\qquad
  \sizeof{\vloc A}=\sizeof{\vlwn A}=1+\sizeof{ A}
  $$
  $$
  \sizeof{A\vlte B}=\sizeof{A\vlpa B}=
  \sizeof{A\vlwi B}=\sizeof{A\vlpl B}=\sizeof{A}+\sizeof B
  $$
\end{definition}

\lemSplit*
\begin{proof}
  Let $\deri$ be the given derivation for $\nnodv{J}{\deri}{X\vlpa K}{\LScp}$.
  We prove all statements simultaneously by induction on the lexicographic pair $\tuple{\sizeof\deri,\sizeof{X\vlpa K}}$ where $\sizeof\deri$ is the number of inference rule instances in $\deri$ that are not $\eqD$ or $\eqU$, and where $\sizeof {X\vlpa K}$ is the size of the conclusion of the derivation $\deri$. We now proceed by case analysis on $X$, always picking some bottom-most rule instance in $\deri$, and making another case analysis on that rule instance.
  \begin{enumerate}
  \item Let $\deri$ be the given derivation for $\nnodv{J}{\deri}{(A\vlte B)\vlpa K}{\LScp}$.
    \begin{enumerate}
    \item if $J=(A\vlte B)\vlpa K$ (i.e., $\sizeof\deri=0$) then we let $\Chole=\conhole\vlpa K$ and $J_A=J_B=\vlone$ and $K_A=K_B=\vlbot$.
    \item\label{s:trivial} If $\deri$ is of shape
      $$
      \odv{J}{\deri'}{(\odn{A'}{\rD}{A}{}\vlte B)\vlpa K}{\LScp}
      \quor
      \odv{J}{\deri'}{(A\vlte\odn{B'}{\rD}{B}{})\vlpa K}{\LScp}
      \quor
      \odv{J}{\deri'}{(A\vlte B)\vlpa \odn{K'}\rD K{}}{\LScp}
      $$
      then we can immediately proceed by induction hypothesis.
    \item\label{s:commsplit}
      If $\deri$ is of shape\footnote{For better readability, we omit some instances of $\eqD$.}
      $$\odv{J}{\deri'}{\odn{
        ((A\vlte B)\vlpa K_1\vlpa K_2)\vlte(K_3\vlpa K_4)}\sD{
        (A\vlte B)\vlpa K_1 \vlpa (K_2\vlte K_3)\vlpa K_4}{}\vlpa K_5}{\LScp}
      $$
      we apply the induction hypothesis to $\deri'$ and obtain
      $$
      \odv{J}{\deri'_J}{\Ccons[1]{J_L\vlte J_R}}{\SLScpU}
      \quand
      \odv{\Ccons[1]{K_L\vlpa K_R}}{\deri'_K}{K_5}{\LScp}
      \quand
      \odv{J_L}{\deri_L}{(A\vlte B)\vlpa K_1\vlpa K_2\vlpa K_L}{\LScp}
      \quand
      \odv{J_R}{\deri_R}{K_3\vlpa K_4\vlpa K_R}{\LScp}
      $$
      Then we apply the induction hypothesis again to $\deri_L$ and obtain
      $$
      \odv{J_L}{\deri''_J}{\Ccons[2]{J_A\vlte J_B}}{\SLScpU}
      \quand
      \odv{\Ccons[2]{K_A\vlpa K_B}}{\deri''_K}{K_1\vlpa K_2\vlpa K_L}{\LScp}
      \quand
      \odv{J_A}{\deri_A}{A\vlpa K_A}{\LScp}
      \quand
      \odv{J_B}{\deri_B}{B\vlpa K_B}{\LScp}
      $$
      We let $\Chole=\Ccons[1]{\Ccons[2]{\cdot}\vlte J_R}$ and construct $\deri_J$ and $\deri_K$ as follows:
      $$
      \odv{J}{\deri'_J}{
        \CCons[1]{
          \odv{J_L}{\deri''_J}{\Ccons[2]{J_A\vlte J_B}}{\SLScpU}
          \vlte J_R}}{\SLScpU}
      \qquand
      \odv{\CCons[1]{\odn{
            \odv{\Ccons[2]{K_A\vlpa K_B}}{\deri''_K}{K_1\vlpa K_2\vlpa K_L}{\LScp}
            \vlte
            \odv{J_R}{\deri_R}{K_3\vlpa K_4\vlpa K_R}{\LScp}
          }\sD{
        K_1 \vlpa (K_2\vlte K_3)\vlpa K_4\vlpa K_L\vlpa K_R}{}
      }}{\derisw}{
        K_1 \vlpa (K_2\vlte K_3)\vlpa K_4\vlpa
        \odv{\Ccons[1]{K_L\vlpa K_R}}{\deri'_K}{K_5}{\LScp}}{\set{\eqD,\sD}}
      $$
      where $\derisw$ exist by \Cref{lem:switch}. 

    \item
      If $\deri$ is of shape
      $$\odv{J}{\deri'}{\odn{
        (A\vlpa K_1)\vlte(B\vlpa K_2)}\sD{
        (A\vlte B)\vlpa (K_1\vlpa K_2)}{}\vlpa K_3}{\LScp}
      $$
      Then we obtain by applying the induction hypothesis to $\deri'$ the following:
      $$
      \odv{J}{\deri_J}{\Ccons{J_A\vlte J_B}}{\SLScpU}
      \quand
      \odv{\Ccons{K_L\vlpa K_R}}{\deri'_K}{K_3}{\LScp}
      \quand
      \odv{J_A}{\deri_A}{A\vlpa K_1\vlpa K_L}{\LScp}
      \quand
      \odv{J_B}{\deri_B}{B\vlpa K_2\vlpa K_R}{\LScp}
      $$
      and we can let $K_A=K_1\vlpa K_L$ and $K_B=K_2\vlpa K_R$ and construct the derivation $\deri_K$ as follows:
      $$
      \odv{
        \Ccons{K_1\vlpa K_L\vlpa K_2\vlpa K_R}}{\derisw}{
        K_1\vlpa K_2\vlpa\odv{\Ccons{K_L\vlpa K_R}}{\deri'_K}{K_3}{\LScp}}{\set{\eqD,\sD}}
      $$
      where $\derisw$ exist by \Cref{lem:switch}.
    \item 
      If $\deri$ is of shape
      $$\odv{J}{\deri'}{\odn{
        \vlone}\tensD{
        \vlone\vlte\vlone}{}\vlpa K}{\LScp}
      $$
      then we can apply Case~\ref{s:one} to $\deri'$ and conclude by letting $J_A=J_B=\vlone$ and $K_A=K_B=\vlbot$.
    \end{enumerate}
  \item Let $\deri$ be the given derivation for $\nnodv{J}{\deri}{(A\vlwi B)\vlpa K}{\LScp}$.
    \begin{enumerate}
    \item if $J=(A\vlwi B)\vlpa K$ (i.e., $\sizeof\deri=0$) then we are in case (i) and let $\Chole=\conhole\vlpa K$ and $J_A=A$ and  $J_B=B$.
    \item If $\deri$ is of shape
      $$
      \odv{J}{\deri'}{(\odn{A'}{\rD}{A}{}\vlwi B)\vlpa K}{\LScp}
      \quor
      \odv{J}{\deri'}{(A\vlwi\odn{B'}{\rD}{B}{})\vlpa K}{\LScp}
      \quor
      \odv{J}{\deri'}{(A\vlwi B)\vlpa \odn{K'}\rD K{}}{\LScp}
      $$
      then, as in Case~\ref{s:trivial}, we can immediately proceed by induction hypothesis.
    \item
      If $\deri$ is of shape
      $$\odv{J}{\deri'}{\odn{
        ((A\vlwi B)\vlpa K_1\vlpa K_2)\vlte(K_3\vlpa K_4)}\sD{
        (A\vlwi B)\vlpa K_1 \vlpa (K_2\vlte K_3)\vlpa K_4}{}\vlpa K_5}{\LScp}
      $$
      then we proceed analogously to Case~\ref{s:commsplit}.
    \item
      If $\deri$ is of shape
      $$\odv{J}{\deri'}{\odn{
        (A\vlpa K_1)\vlwi(B\vlpa K_2)}\dD{
        (A\vlwi B)\vlpa (K_1\vlpl K_2)}{}\vlpa K_3}{\LScp}
      $$
      Then we obtain by applying the induction hypothesis to $\deri'$ either 
      $$
      \odv{J}{\deri_J}{\Ccons{J_A\vlwi J_B}}{\SLScpU}
      \quand
      \odv{\Ccons{\vlbot}}{\deri'_K}{K_3}{\LScp}
      \quand
      \odv{J_A}{\deri_A}{A\vlpa K_1}{\LScp}
      \quand
      \odv{J_B}{\deri_B}{B\vlpa K_2}{\LScp}
      $$
      or
      $$
      \odv{J}{\deri_J}{\Ccons{J_A\vlwi J_B}}{\SLScpU}
      \quand
      \odv{\Ccons{K_L\vlpl K_R}}{\deri'_K}{K_3}{\LScp}
      \quand
      \odv{J_A}{\deri_A}{A\vlpa K_1\vlpa K_L}{\LScp}
      \quand
      \odv{J_B}{\deri_B}{B\vlpa K_2\vlpa K_R}{\LScp}
      $$
      In both cases we arrive at situation (ii), by letting in the first case $K_A=K_1$ and $K_B=K_2$, and in the second case $K_A=K_1\vlpa K_L$ and $K_B=K_2\vlpa K_R$, where $\deri_k$ is in the first case the derivation on the left, and in the second case the derivation on the right below:
      $$
      \odv{
        \Ccons{K_1\vlpl K_2}}{\derisw}{
        (K_1\vlpl K_2)\vlpa\odv{\Ccons{\vlbot}}{\deri'_K}{K_3}{\LScp}}{\set{\eqD,\sD}}
      \qquor
      \odv{
        \CCons{\odn{(K_1\vlpa K_L)\vlpl (K_2\vlpa K_R)}{\dpD}{
         (K_1\vlpl K_2)\vlpa (K_L\vlpl K_R)}{}}}{\derisw}{
        (K_1\vlpl K_2)\vlpa\odv{\Ccons{K_L\vlpl K_R}}{\deri'_K}{K_3}{\LScp}}{\set{\eqD,\sD}}
      $$
      where in both cases $\derisw$ exist by \Cref{lem:switch}.
    \item 
      If $\deri$ is of shape
      $$\odv{J}{\deri'}{\odn{
        \vlone}\withD{
        \vlone\vlwi\vlone}{}\vlpa K}{\LScp}
      $$
      then we can apply Case~\ref{s:one} to $\deri'$ and arrive at case (i) by letting $\Chole=\conhole\vlpa K$ and $J_A=J_B=\vlone$.
    \end{enumerate}
  \item Let $\deri$ be the given derivation for $\nnodv{J}{\deri}{(A\vlpl B)\vlpa K}{\LScp}$.
    \begin{enumerate}
    \item if $J=(A\vlpl B)\vlpa K$ (i.e., $\sizeof\deri=0$) then we are in case (i) and let $\Chole=\conhole\vlpa K$ and $J_A=A$ and  $J_B=B$.
    \item If $\deri$ is of shape
      $$
      \odv{J}{\deri'}{(\odn{A'}{\rD}{A}{}\vlpl B)\vlpa K}{\LScp}
      \quor
      \odv{J}{\deri'}{(A\vlpl\odn{B'}{\rD}{B}{})\vlpa K}{\LScp}
      \quor
      \odv{J}{\deri'}{(A\vlpl B)\vlpa \odn{K'}\rD K{}}{\LScp}
      $$
      then, as in Case~\ref{s:trivial}, we can immediately proceed by induction hypothesis.
    \item
      If $\deri$ is of shape
      $$\odv{J}{\deri'}{\odn{
        ((A\vlpl B)\vlpa K_1\vlpa K_2)\vlte(K_3\vlpa K_4)}\sD{
        (A\vlpl B)\vlpa K_1 \vlpa (K_2\vlte K_3)\vlpa K_4}{}\vlpa K_5}{\LScp}
      $$
      then we proceed analogously to Case~\ref{s:commsplit}.
    \item
      If $\deri$ is of shape
      $$\odv{J}{\deri'}{\odn{
        (A\vlpa K_1)\vlva(B\vlpa K_2)}\dhD{
        (A\vlpl B)\vlpa (K_1\vlva K_2)}{}\vlpa K_3}{\LScp}
      $$
      for some $\vlva\in\set{\vlpl,\vlwi}$ and $\dhD\in\set{\dD,\dpD}$.
      Then we obtain by applying the induction hypothesis to $\deri'$ either 
      $$
      \odv{J}{\deri_J}{\Ccons{J_A\vlva J_B}}{\SLScpU}
      \quand
      \odv{\Ccons{\vlbot}}{\deri'_K}{K_3}{\LScp}
      \quand
      \odv{J_A}{\deri_A}{A\vlpa K_1}{\LScp}
      \quand
      \odv{J_B}{\deri_B}{B\vlpa K_2}{\LScp}
      $$
      or 
      $$
      \odv{J}{\deri_J}{\Ccons{J_A\vlwa J_B}}{\SLScpU}
      \quand
      \odv{\Ccons{K_L\vlwa K_R}}{\deri'_K}{K_3}{\LScp}
      \quand
      \odv{J_A}{\deri_A}{A\vlpa K_1\vlpa K_L}{\LScp}
      \quand
      \odv{J_B}{\deri_B}{B\vlpa K_2\vlpa K_R}{\LScp}
      $$
      for some $\vlwa\in\set{\vlpl,\vlwi}$ if  $\vlva=\vlpl$, or 
      $$
      \odv{J}{\deri_J}{\Ccons{J_A\vlwi J_B}}{\SLScpU}
      \quand
      \odv{\Ccons{K_L\vlpl K_R}}{\deri'_K}{K_3}{\LScp}
      \quand
      \odv{J_A}{\deri_A}{A\vlpa K_1\vlpa K_L}{\LScp}
      \quand
      \odv{J_B}{\deri_B}{B\vlpa K_2\vlpa K_R}{\LScp}
      $$
      if $\vlva=\vlwi$.
      In all three cases we arrive at situation (ii), by letting in the first case $K_A=K_1$ and $K_B=K_2$, and in the second and third case $K_A=K_1\vlpa K_L$ and $K_B=K_2\vlpa K_R$,  where $\deri_k$ is one of the following four derivations (1: in the first case above; 2: in the second case with $\vlva=\vlpl$ and $\vlwa=\vlpl$; 3: in the second case with $\vlva=\vlpl$ and $\vlwa=\vlwi$; and 4: in the third case above with $\vlva=\vlwi$):
      $$
      1:\; 
      \odv{
        \Ccons{K_1\vlva K_2}}{\derisw}{
        (K_1\vlva K_2)\vlpa\odv{\Ccons{\vlbot}}{\deri'_K}{K_3}{\LScp}}{\set{\eqD,\sD}}
      \hskip4em
      2:\;
      \odv{
        \CCons{\odn{(K_1\vlpa K_L)\vlpl (K_2\vlpa K_R)}{\dpD}{
         (K_1\vlpl K_2)\vlpa (K_L\vlpl K_R)}{}}}{\derisw}{
        (K_1\vlpl K_2)\vlpa\odv{\Ccons{K_L\vlpl K_R}}{\deri'_K}{K_3}{\LScp}}{\set{\eqD,\sD}}
      $$
      $$
      3:\;\odv{
        \CCons{\odn{(K_1\vlpa K_L)\vlwi (K_2\vlpa K_R)}{\dpD}{
         (K_1\vlpl K_2)\vlpa (K_L\vlwi K_R)}{}}}{\derisw}{
        (K_1\vlpl K_2)\vlpa\odv{\Ccons{K_L\vlwi K_R}}{\deri'_K}{K_3}{\LScp}}{\set{\eqD,\sD}}
      \qquad
      4:\;
      \odv{
        \CCons{\odn{(K_1\vlpa K_L)\vlwi (K_2\vlpa K_R)}{\dpD}{
         (K_1\vlwi K_2)\vlpa (K_L\vlpl K_R)}{}}}{\derisw}{
        (K_1\vlwi K_2)\vlpa\odv{\Ccons{K_L\vlpl K_R}}{\deri'_K}{K_3}{\LScp}}{\set{\eqD,\sD}}
      $$
      where in all cases $\derisw$ exist by \Cref{lem:switch}.
    \end{enumerate}
  \item Let $\deri$ be the given derivation for $\nnodv{J}{\deri}{\vloc A\vlpa K}{\LScp}$.
    \begin{enumerate}
    \item if $J=\vloc A\vlpa K$ (i.e., $\sizeof\deri=0$) then we are in case (i) and let $\Chole=\conhole\vlpa K$ and $J_A=A$.
    \item If $\deri$ is of shape
      $$
      \odv{J}{\deri'}{\vloc \odn{A'}{\rD}{A}{}\vlpa K}{\LScp}
      \quor
      \odv{J}{\deri'}{\vloc A\vlpa \odn{K'}\rD K{}}{\LScp}
      $$
      then, as in Case~\ref{s:trivial}, we can immediately proceed by induction hypothesis.
    \item
      If $\deri$ is of shape
      $$\odv{J}{\deri'}{\odn{
        (\vloc A\vlpa K_1\vlpa K_2)\vlte(K_3\vlpa K_4)}\sD{
        \vloc A \vlpa K_1 \vlpa (K_2\vlte K_3)\vlpa K_4}{}\vlpa K_5}{\LScp}
      $$
      then we proceed analogously to Case~\ref{s:commsplit}.
    \item
      If $\deri$ is of shape
      $$\odv{J}{\deri'}{\odn{
        \vloc (A\vlpa K_1)}\pD{
        \vloc A \vlpa \vlwn K_1}{}\vlpa K_2}{\LScp}
      $$
      Then we obtain by applying the induction hypothesis to $\deri'$ either 
      $$
      \odv{J}{\deri_J}{\Ccons{\vloc J_A}}{\SLScpU}
      \quand
      \odv{\Ccons{\vlbot}}{\deri'_K}{K_2}{\LScp}
      \quand
      \odv{J_A}{\deri_A}{A\vlpa K_1}{\LScp}
      $$
      or
      $$
      \odv{J}{\deri_J}{\Ccons{\vloc J_A}}{\SLScpU}
      \quand
      \odv{\Ccons{\vlwn K'}}{\deri'_K}{K_2}{\LScp}
      \quand
      \odv{J_A}{\deri_A}{A\vlpa K_1\vlpa K'}{\LScp}
      $$
      In both cases we arrive at situation (ii), by letting in the first case $K_A=K_1$ and $\deri_k$ is in the derivation on the left below, and in the second case $K_A=K_1\vlpa K'$ and $\deri_k$ is the derivation on the right below:
      $$
      \odv{
        \Ccons{\vlwn K_1}}{\derisw}{
        \vlwn K_1\vlpa\odv{\Ccons{\vlbot}}{\deri'_K}{K_2}{\LScp}}{\set{\eqD,\sD}}
      \qquad\qquad
      \odv{
        \CCons{\odn{\vlwn(K_1\vlpa K')}{\ppD}{
         \vlwn K_1\vlpa \vlwn K'}{}}}{\derisw}{
        \vlwn K_1\vlpa\odv{\Ccons{\vlwn K'}}{\deri'_K}{K_2}{\LScp}}{\set{\eqD,\sD}}
      $$
      where in both cases $\derisw$ exist by \Cref{lem:switch}.
    \item 
      If $\deri$ is of shape
      $$\odv{J}{\deri'}{\odn{
        \vlone}\eD{
        \vloc\vlone}{}\vlpa K}{\LScp}
      $$
      then we can apply Case~\ref{s:one} to $\deri'$ and have case (i) by letting $\Chole=\conhole\vlpa K$ and $J_A=\vlone$.
    \end{enumerate}
  \item Let $\deri$ be the given derivation for $\nnodv{J}{\deri}{\vlwn A\vlpa K}{\LScp}$.
    \begin{enumerate}
    \item if $J=\vlwn A\vlpa K$ (i.e., $\sizeof\deri=0$) then we are in case (i) and let $\Chole=\conhole\vlpa K$ and $J_A=A$.
    \item If $\deri$ is of shape
      $$
      \odv{J}{\deri'}{\vlwn \odn{A'}{\rD}{A}{}\vlpa K}{\LScp}
      \quor
      \odv{J}{\deri'}{\vlwn A\vlpa \odn{K'}\rD K{}}{\LScp}
      $$
      then, as in Case~\ref{s:trivial}, we can immediately proceed by induction hypothesis.
    \item
      If $\deri$ is of shape
      $$\odv{J}{\deri'}{\odn{
        (\vlwn A\vlpa K_1\vlpa K_2)\vlte(K_3\vlpa K_4)}\sD{
        \vlwn A \vlpa K_1 \vlpa (K_2\vlte K_3)\vlpa K_4}{}\vlpa K_5}{\LScp}
      $$
      then we proceed analogously to Case~\ref{s:commsplit}.
    \item
      If $\deri$ is of shape
      $$
      \odv{J}{\deri'}{\odn{
        \vlwn (A\vlpa K_1)}\ppD{
          \vlwn A \vlpa \vlwn K_1}{}\vlpa K_2}{\LScp}
      \qquor
      \odv{J}{\deri'}{\odn{
        \vloc (A\vlpa K_1)}\pD{
        \vlwn A \vlpa \vloc K_1}{}\vlpa K_2}{\LScp}
      $$
      Then we obtain by applying the induction hypothesis to $\deri'$ either 
      $$
      \odv{J}{\deri_J}{\Ccons{\vlda J_A}}{\SLScpU}
      \quand
      \odv{\Ccons{\vlbot}}{\deri'_K}{K_2}{\LScp}
      \quand
      \odv{J_A}{\deri_A}{A\vlpa K_1}{\LScp}
      $$
      for $\vlda\in\set{\vloc,\vlwn}$, or
      $$
      \odv{J}{\deri_J}{\Ccons{\vlda J_A}}{\SLScpU}
      \quand
      \odv{\Ccons{\vlda K'}}{\deri'_K}{K_2}{\LScp}
      \quand
      \odv{J_A}{\deri_A}{A\vlpa K_1\vlpa K'}{\LScp}
      $$
      for $\vlda\in\set{\vloc,\vlwn}$, or
      $$
      \odv{J}{\deri_J}{\Ccons{\vloc J_A}}{\SLScpU}
      \quand
      \odv{\Ccons{\vlwn K'}}{\deri'_K}{K_2}{\LScp}
      \quand
      \odv{J_A}{\deri_A}{A\vlpa K_1\vlpa K'}{\LScp}
      $$
      In all three cases we arrive at situation (ii), by letting in the first case $K_A=K_1$ and $\deri_k$, and in the second and third case $K_A=K_1\vlpa K'$. Then, for each case, the corresponding derivation $\deri_k$ is shown below:
      $$
      \odv{
        \Ccons{\vlda K_1}}{\derisw}{
        \vlda K_1\vlpa\odv{\Ccons{\vlbot}}{\deri'_K}{K_2}{\LScp}}{\set{\eqD,\sD}}
      \qquad\qquad
      \odv{
        \CCons{\odn{\vlda(K_1\vlpa K')}{\ppD}{
         \vlwn K_1\vlpa \vlda K'}{}}}{\derisw}{
        \vlwn K_1\vlpa\odv{\Ccons{\vlda K'}}{\deri'_K}{K_2}{\LScp}}{\set{\eqD,\sD}}
      \qquad\qquad
      \odv{
        \CCons{\odn{\vloc(K_1\vlpa K')}{\ppD}{
         \vloc K_1\vlpa \vlwn K'}{}}}{\derisw}{
        \vloc K_1\vlpa\odv{\Ccons{\vlwn K'}}{\deri'_K}{K_2}{\LScp}}{\set{\eqD,\sD}}
      $$
      where in all three cases $\derisw$ exist by \Cref{lem:switch}.
    \item 
      If $\deri$ is of shape
      $$\odv{J}{\deri'}{\odn{
        \vlone}\eD{
        \vloc\vlone}{}\vlpa K}{\LScp}
      $$
      then we can apply Case~\ref{s:one} to $\deri'$ and have case (i) by letting $\Chole=\conhole\vlpa K$ and $J_A=\vlone$.
    \end{enumerate}
  \item Let $\deri$ be a derivation for $\nnodv{J}{\deri}{\vlone\vlpa K}{\LScp}$.
    \begin{enumerate}
    \item if $J=\vlone\vlpa K$ (i.e., $\sizeof\deri=0$) then we let $\Chole=\conhole\vlpa K$ and we are done because we can let $\deri_K$ be $\vlinf{\eqD}{}{K}{\vlbot\vlpa K}$.
    \item If $\deri$ is of shape
      $$
      \odv{J}{\deri'}{\vlone\vlpa \odn{K'}\rD K{}}{\LScp}
      $$
      then we can immediately proceed by induction hypothesis, as in Case~\ref{s:trivial}.
    \item
      If $\deri$ is of shape
      $$\odv{J}{\deri'}{\odn{
        (\vlone\vlpa K_1\vlpa K_2)\vlte(K_3\vlpa K_4)}\sD{
        \vlone\vlpa K_1 \vlpa (K_2\vlte K_3)\vlpa K_4}{}\vlpa K_5}{\LScp}
      $$
      we proceed analogously to Case~\ref{s:commsplit}, using the induction hypothesis and Case~\ref{s:tensor}.
    \end{enumerate}
  \item Let $\deri$ be a derivation for $\nnodv{J}{\deri}{\vltop\vlpa K}{\LScp}$.
    \begin{enumerate}
    \item\label{s:toptrivial} if $J=\vltop\vlpa K$ (i.e., $\sizeof\deri=0$) then we let $\Chole=\conhole\vlpa K$ and $\ttt=\vltop$ and $\fff=\vlbot$ and we can let $\deri_K$ be $\vlinf{\eqD}{}{K}{\vlbot\vlpa K}$.
    \item If $\deri$ is of shape
      $$
      \odv{J}{\deri'}{\vltop\vlpa \odn{K'}\rD K{}}{\LScp}
      $$
      then we can immediately proceed by induction hypothesis, as in Case~\ref{s:trivial}.
    \item
      If $\deri$ is of shape
      $$\odv{J}{\deri'}{\odn{
        (\vltop\vlpa K_1\vlpa K_2)\vlte(K_3\vlpa K_4)}\sD{
        \vltop\vlpa K_1 \vlpa (K_2\vlte K_3)\vlpa K_4}{}\vlpa K_5}{\LScp}
      $$
      we proceed analogously to Case~\ref{s:commsplit}, using the induction hypothesis and Case~\ref{s:tensor}.
    \item If $\deri$ is of shape
      $$\odv{J}{\deri'}{\odn{
        \vltop}\dzD{
        \vltop\vlpa\vlzer}{}\vlpa K}{\LScp}
      $$
      then we apply the induction hypothesis to $\delta'$ and get 
      $$
      \odv{J}{\deri_J}{\Ccons{\ttt}}{\SLScpU}
      \quand
      \odv{\Ccons{\fff}}{\deri'_K}{K}{\LScp}
      $$
      for some $\ttt\in\set{\vlone,\vltop}$ and $\fff\in\set{\vlzer,\vlbot}$,
      Depending on $\fff$, we can construct $\deri_K$ as follows:
      $$
      \odv{
        \Ccons{\vlzer}}{\derisw}{
        \vlzer\vlpa\odv{\Ccons{\vlbot}}{\deri'_K}{K}{\LScp}}{\set{\eqD,\sD}}
      \qquor
      \odv{
        \CCons{\odn{\vlzer}\dzpD{
            \vlzer\vlpa\vlzer}{}
        }
      }{\derisw}{
        \vlzer\vlpa\odv{\Ccons{\vlzer}}{\deri'_K}{K}{\LScp}}{\set{\eqD,\sD}}
      $$
      where in both cases $\derisw$ exist by \Cref{lem:switch}.
    \item If $\deri$ is of shape
      $$\odv{J}{\deri'}{\odn{
        \vlone}\topD{
        \vltop}{}\vlpa K}{\LScp}
      $$
      then we apply the induction hypothesis  (Case~\ref{s:one}) to $\delta'$ and get 
      $$
      \odv{J}{\deri_J}{\Ccons{\vlone}}{\SLScpU}
      \quand
      \odv{\Ccons{\vlbot}}{\deri_K}{K}{\LScp}
      \quadfs
      $$
    \end{enumerate}
  \item Let $\deri$ be  a derivation for $\nnodv{J}{\deri}{\vlzer\vlpa K}{\LScp}$.
    \begin{enumerate}
    \item if $J=\vlzer\vlpa K$ (i.e., $\sizeof\deri=0$) then we let $\Chole=\conhole\vlpa K$ and we are in Case (i) with $\fff=\vlbot$ and we can let $\deri_K$ be $\vlinf{\eqD}{}{K}{\vlbot\vlpa K}$.
    \item If $\deri$ is of shape
      $$
      \odv{J}{\deri'}{\vltop\vlpa \odn{K'}\rD K{}}{\LScp}
      $$
      then we can immediately proceed by induction hypothesis, as in Case~\ref{s:trivial}.
    \item
      If $\deri$ is of shape
      $$\odv{J}{\deri'}{\odn{
        (\vltop\vlpa K_1\vlpa K_2)\vlte(K_3\vlpa K_4)}\sD{
        \vltop\vlpa K_1 \vlpa (K_2\vlte K_3)\vlpa K_4}{}\vlpa K_5}{\LScp}
      $$
      we proceed analogously to Case~\ref{s:commsplit}, using the induction hypothesis and Case~\ref{s:tensor}.
    \item If $\deri$ is of shape
      $$\odv{J}{\deri'}{\odn{
        \vltop}\dzD{
        \vlzer\vlpa\vltop}{}\vlpa K}{\LScp}
      \qquor      
      $$
      then we apply the induction hypothesis (Case~\ref{s:top}) to $\delta'$ and get 
      $$
      \odv{J}{\deri_J}{\Ccons{\ttt}}{\SLScpU}
      \quand
      \odv{\Ccons{\fff}}{\deri'_K}{K}{\LScp}
      $$
      for some $\ttt\in\set{\vlone,\vltop}$ and $\fff\in\set{\vlzer,\vlbot}$,
      and depending on $\fff$, we can construct $\deri_K$ as follows:
      $$
      \odv{
        \Ccons{\vltop}}{\derisw}{
        \vltop\vlpa\odv{\Ccons{\vlbot}}{\deri'_K}{K}{\LScp}}{\set{\eqD,\sD}}
      \qquor
      \odv{
        \CCons{\odn{\vltop}\dzD{
            \vltop\vlpa\vlzer}{}
        }
      }{\derisw}{
        \vltop\vlpa\odv{\Ccons{\vlzer}}{\deri'_K}{K}{\LScp}}{\set{\eqD,\sD}}
      $$
      where in both cases $\derisw$ exist by \Cref{lem:switch}, and we have Case~(ii).
    \item If $\deri$ is of shape
      $$\odv{J}{\deri'}{\odn{
        \vlzer}\dzpD{
        \vlzer\vlpa\vlzer}{}\vlpa K}{\LScp}
      \qquor      
      $$
      then we apply the induction hypothesis to $\delta'$ and get
      either
      $$
      \mbox{(i)}\quad
      \odv{J}{\deri_J}{\Ccons{\vlzer}}{\SLScpU}
      \quand
      \odv{\Ccons{\fff}}{\deri'_K}{K}{\LScp}
      \qquor
      \mbox{(ii)}\quad
      \odv{J}{\deri_J}{\Ccons{\ttt}}{\SLScpU}
      \quand
      \odv{\Ccons{\vltop}}{\deri'_K}{K}{\LScp}
      $$
      for some $\fff\in\set{\vlzer,\vlbot}$ and $\ttt\in\set{\vlone,\vltop}$.
      In the first case, we construct $\deri_K$, depending on $\fff$, as shown on the left or the middle below, and in the second case, we can construct $\deri_K$ as shown on the right below:
      $$
      \odv{
        \Ccons{\vlzer}}{\derisw}{
        \vlzer\vlpa\odv{\Ccons{\vlbot}}{\deri'_K}{K}{\LScp}}{\set{\eqD,\sD}}
      \qquor
      \odv{
        \CCons{\odn{\vlzer}\dzpD{
            \vlzer\vlpa\vlzer}{}
        }
      }{\derisw}{
        \vlzer\vlpa\odv{\Ccons{\vlzer}}{\deri'_K}{K}{\LScp}}{\set{\eqD,\sD}}
       \qquor
      \odv{
        \CCons{\odn{\vltop}\dzD{
            \vltop\vlpa\vlzer}{}
        }
      }{\derisw}{
        \vlzer\vlpa\odv{\Ccons{\vltop}}{\deri'_K}{K}{\LScp}}{\set{\eqD,\sD}}
     $$
      where in all three cases $\derisw$ exist by \Cref{lem:switch}.
    \end{enumerate}
  \item Let a derivation $\nodv{J}{\deri}{a\vlpa K}{\LScp}$ be given.
    \begin{enumerate}
    \item if $J=a\vlpa K$ (i.e., $\sizeof\deri=0$) then we let $\Chole=\conhole\vlpa K$ and we are in Case (i), similarly to Case~\ref{s:toptrivial}.
    \item If $\deri$ is of shape
      $$
      \odv{J}{\deri'}{a\vlpa \odn{K'}\rD K{}}{\LScp}
      $$
      then we can immediately proceed by induction hypothesis, as in Case~\ref{s:trivial}.
    \item
      If $\deri$ is of shape
      $$\odv{J}{\deri'}{\odn{
        (a\vlpa K_1\vlpa K_2)\vlte(K_3\vlpa K_4)}\sD{
        a\vlpa K_1 \vlpa (K_2\vlte K_3)\vlpa K_4}{}\vlpa K_5}{\LScp}
      $$
      we proceed analogously to Case~\ref{s:commsplit}, using the induction hypothesis and Case~\ref{s:tensor}.
    \item If $\deri$ is of shape
      $$\odv{J}{\deri'}{\odn{
        \vlone}\topD{
        a\vlpa\lneg a}{}\vlpa K}{\LScp}
      $$
      then we apply the induction hypothesis to $\delta'$ and get 
      $$
      \odv{J}{\deri_J}{\Ccons{\vlone}}{\SLScpU}
      \quand
      \odv{\Ccons{\vlbot}}{\deri'_K}{K}{\LScp}
      $$
      Then we have Case (ii) and construct $\deri_K$ as follows:
      $$
      \odv{
        \Ccons{\lneg a}}{\derisw}{
        \lneg a\vlpa\odv{\Ccons{\vlbot}}{\deri'_K}{K}{\LScp}}{\set{\eqD,\sD}}
      $$
      where $\derisw$ exist by \Cref{lem:switch}.
    \end{enumerate}
  \end{enumerate}
\end{proof}

\section{Proofs for \Cref{sec:flipping}}\label{app:flipping}

\lemCoreFlip*

\begin{proof}
  We proceed by structural induction on $F$, applying \Cref{lem:splitting} to $\deri$. Before we begin the case analysis, observe that in all cases in \Cref{lem:splitting} we obtain a context $\Chole$ and a derivation $\downsmash{\odv{J}{\deri_J}{\Ccons{\Jast}}{\SLScpU}}$. With this we can construct the following derivation
  \begin{equation}
    \label{eq:deriH}
    \hfill
    \odv{
      \lneg F\vlte
      \od{\odd{\odd{\odh{H}}{\derib}{J}{\SLScpU}}{
          \deri_J}{\Ccons{\Jast}}{\SLScpU}}}{
      \deri_s}{\CCons{\lneg F\vlte \Jast}
    }{\SLScpU}
    \hfill
  \end{equation}
  where $\deri_s$ exits by \Cref{lem:splitting}. We use $\deri_H$ to denote the derivation in~\eqref{eq:deriH} above. What follows is a case analysis on $F$.
  \begin{enumerate}
  \item\label{f:atom} $F=a$. We apply \Cref{lem:splitting}.\ref{s:atom} to $\deri$.
    In Case~(i) we let $I=\Ccons{\vlbot}$, and in Case~(ii) we let $I=\Ccons{\lneg a}$.
    In both cases we let $\derid=\deri_K$, and $\deric$ is constructed as follows:
    $$
    \mbox{(i)}\quad
    \odV{\lneg a\vlte H}{
      \deri_H}{\CCons{\odn{\lneg a\vlte a}{\aiU}{\vlbot}{}}
    }{\SLScpU}
    \qquor
    \mbox{(ii)}\quad
    \odV{\lneg a\vlte H}{
      \deri_H}{\CCons{\odt{\lneg a\vlte\lone}{\eqU}{\lneg a}{}}
    }{\SLScpU}
    $$
    where in both cases, $\deri_J$ comes from \Cref{lem:splitting}.\ref{s:atom}, and $\deri_s$ exits by \Cref{lem:splitting}.
  \item\label{f:bot} $F=\vlbot$. This case is trivial with $I=J$.
  \item\label{f:one} $F=\vlone$. We apply \Cref{lem:splitting}.\ref{s:one} to $\deri$, let $I=\Ccons{\vlbot}$ and $\derid=\deri_K$ and $\deric$ is constructed analogously to Case~\ref{f:atom}.(ii) above.
  \item \label{f:top} $F=\vltop$. We apply \Cref{lem:splitting}.\ref{s:top} to $\deri$, giving us four possibilities. In all we have $\derid=\deri_K$ and $\deric$ is one of the following:
    $$
    \odV{\vlzer\vlte H}{
      \deri_H}{\CCons{\odn{\vlzer\vlte\vltop}{\dzU}{\vlzer}{}}
    }{\SLScpU}
    \quor
    \odV{\vlzer\vlte H}{
      \deri_H}{\CCons{\odt{\vlzer\vlte\vlone}{\eqU}{\vlzer}{}}
    }{\SLScpU}
    \quor
    \odV{\vlzer\vlte H}{
      \deri_H}{\CCons{
        \od{\odi{\odi{\odh{\vlzer\vlte\vltop}}{
              \dzU}{\vlzer}{}}{
            \topU}{\vlbot}{}}}
    }{\SLScpU}
    \quor
    \odV{\vlzer\vlte H}{
      \deri_H}{\CCons{
        \od{\odi{\odo{\odh{\vlzer\vlte\vlone}}{
              \eqU}{\vlzer}{}}{
            \topU}{\vlbot}{}}}
    }{\SLScpU}
    $$
  \item \label{f:zero} $F=\vlzer$. We apply \Cref{lem:splitting}.\ref{s:zero} to $\deri$, giving us four possibilities. In all we have $\derid=\deri_K$ and $\deric$ is one of the following:
    $$
    \odV{\vltop\vlte H}{
      \deri_H}{\CCons{
        \od{\odi{\odi{\odh{\vlzer\vlte\vltop}}{
              \dzU}{\vlzer}{}}{
            \topU}{\vlbot}{}}}
    }{\SLScpU}
    \quor
    \odV{\vltop\vlte H}{
      \deri_H}{\CCons{\odn{\vlzer\vlte\vltop}{\dzU}{\vlzer}{}}
    }{\SLScpU}
    \quor
    \odV{\vltop\vlte H}{
      \deri_H}{\CCons{\odt{\vltop\vlte\vlone}{\eqU}{\vltop}{}}
    }{\SLScpU}
    \quor
    \odV{\vltop\vlte H}{
      \deri_H}{\CCons{\odn{\vltop\vlte\vltop}{\dzpU}{\vltop}{}}
    }{\SLScpU}
    $$
  \item \label{f:tensor} $F=A\vlte B$. We apply \Cref{lem:splitting}.\ref{s:tensor} to $\deri$, giving us derivations $\deri_J,\deri_K,\deri_A,\deri_B$. We apply the induction hypothesis to the derivations $\deri_A$ and $\deri_B$, giving us formulas $I_A$ and $I_B$, and the following derivations:
    $$
    \od{\odd{\odd{\odh{\lneg A \vlte J_A}}{\deric_A}{I_A}{\SLScpU}}{
        \derid_A}{K_A}{\LScp}}
    \qquand
    \od{\odd{\odd{\odh{\lneg B \vlte J_B}}{\deric_B}{I_B}{\SLScpU}}{
        \derid_B}{K_B}{\LScp}}
    $$
    Now we can let $I=\Ccons{I_A\vlpa I_B}$, and $\deric$ and $\derid$ are the two derivations below:
    $$
    \deric =
    \odV{(\lneg A\vlpa\lneg B)\vlte H}{
      \deri_H}{\CCons{\odn{
          (\lneg A\vlpa\lneg B)\vlte (J_A\vlte J_B)}{
          \sU}{\odv{\lneg A \vlte J_A}{\deric_A}{I_A}{\SLScpU}
          \vlpa
          \odv{\lneg B \vlte J_B}{\deric_B}{I_B}{\SLScpU}
        }{}}
    }{\SLScpU}
    \qquand
    \derid=
    \odV{\CCons{
        \odv{I_A}{\derid_A}{K_A}{\LScp}
        \vlpa
        \odv{I_B}{\derid_B}{K_B}{\LScp}
    }}{\deri_K}{K}{\LScp}
    $$
    where $\deri_K$ comes from \Cref{lem:splitting}.\ref{s:tensor}.
  \item \label{f:with} $F=A\vlwi B$. We apply \Cref{lem:splitting}.\ref{s:with} to $\deri$, giving us derivations $\deri_J,\deri_K,\deri_A,\deri_B$ (there are two cases). We apply the induction hypothesis to the derivations $\deri_A$ and $\deri_B$, giving us formulas $I_A$ and $I_B$, such that we have the following derivations:
    \begin{equation}
      \label{eq:AB}
    \mbox{(i)}\quad
    \od{\odd{\odd{\odh{\lneg A \vlte J_A}}{\deric_A}{I_A}{\SLScpU}}{
        \derid_A}{\vlbot}{\LScp}}
    \quand
    \od{\odd{\odd{\odh{\lneg B \vlte J_B}}{\deric_B}{I_B}{\SLScpU}}{
        \derid_B}{\vlbot}{\LScp}}
    \quad,\quad
    \mbox{or (ii)}\quad
    \od{\odd{\odd{\odh{\lneg A \vlte J_A}}{\deric_A}{I_A}{\SLScpU}}{
        \derid_A}{K_A}{\LScp}}
    \quand
    \od{\odd{\odd{\odh{\lneg B \vlte J_B}}{\deric_B}{I_B}{\SLScpU}}{
        \derid_B}{K_B}{\LScp}}
    \end{equation}
    By inspecting the rules for $\LScp$, it follows that in case (i), we have $I_A=I_B=\vlbot$.
    Then, in case (i), we let $I=\Ccons{\vlbot}$, and in case (ii) we let  $I=\Ccons{I_A\vlpl I_B}$. The derivation $\deric$ is constructed as follows:
    $$
    \mbox{(i) }
    \odV{(\lneg A\vlpl\lneg B)\vlte H}{
      \deri_H}{\CCons{\od{\odi{\odi{\odh{
          (\lneg A\vlpl\lneg B)\vlte (J_A\vlwi J_B)}}{
          \dU}{\odv{\lneg A \vlte J_A}{\deric_A}{\vlbot}{\SLScpU}
          \vlpl
          \odv{\lneg B \vlte J_B}{\deric_B}{\vlbot}{\SLScpU}
            }{}}{
            \withU}{\vlbot}{}}}
    }{\SLScpU}
    \quor
    \mbox{(ii) }
    \odV{(\lneg A\vlpl\lneg B)\vlte H}{
      \deri_H}{\CCons{\odn{
          (\lneg A\vlpl\lneg B)\vlte (J_A\vlwi J_B)}{
          \dU}{\odv{\lneg A \vlte J_A}{\deric_A}{I_A}{\SLScpU}
          \vlpl
          \odv{\lneg B \vlte J_B}{\deric_B}{I_B}{\SLScpU}
        }{}}
    }{\SLScpU}
    $$
    and $\derid$ is in case (i) just $\deri_K$, and in case (ii) analogously to Case~\ref{f:tensor} above.
  \item \label{f:plus} $F=A\vlpl B$. We apply \Cref{lem:splitting}.\ref{s:plus} to $\deri$, giving us derivations $\deri_J,\deri_K,\deri_A,\deri_B$ (there are again two cases). We apply the induction hypothesis to the derivations $\deri_A$ and $\deri_B$, giving us formulas $I_A$ and $I_B$, such that we have the same derivations as in~\eqref{eq:AB} above.
    As before, by inspecting the rules for $\LScp$, it follows that in case (i), we have $I_A=I_B=\vlbot$.
    In case~(i), we let $I=\Ccons{\vlbot}$, and in case~(ii) we let  $I=\Ccons{I_A\vlva I_B}$, where $\vlva\in\set{\vlwi,\vlpl}$ as given by \Cref{lem:splitting}.\ref{s:plus}. The derivation $\deric$ is constructed as follows:
    $$
    \mbox{(i) }
    \odV{(\lneg A\vlwi\lneg B)\vlte H}{
      \deri_H}{\CCons{\od{\odi{\odi{\odh{
          (\lneg A\vlwi\lneg B)\vlte (J_A\vlpl J_B)}}{
          \dU}{\odv{\lneg A \vlte J_A}{\deric_A}{\vlbot}{\SLScpU}
          \vlpl
          \odv{\lneg B \vlte J_B}{\deric_B}{\vlbot}{\SLScpU}
            }{}}{
            \withU}{\vlbot}{}}}
    }{\SLScpU}
    \quor
    \mbox{(ii) }
    \odV{(\lneg A\vlwi\lneg B)\vlte H}{
      \deri_H}{\CCons{\odn{
          (\lneg A\vlwi\lneg B)\vlte (J_A\vlva J_B)}{
          \dU}{\odv{\lneg A \vlte J_A}{\deric_A}{I_A}{\SLScpU}
          \vlva
          \odv{\lneg B \vlte J_B}{\deric_B}{I_B}{\SLScpU}
        }{}}
    }{\SLScpU}
    $$
    and $\derid$ is in case (i) just $\deri_K$, and in case (ii) analogously to Case~\ref{f:tensor} above.
  \item \label{f:oc} $F=\vloc A$. We apply \Cref{lem:splitting}.\ref{s:oc} to $\deri$, giving us derivations $\deri_J,\deri_K,\deri_A$ (there are two cases). We apply the induction hypothesis to $\deri_A$, giving us a formula $I_A$ and one of the following derivations:
    \begin{equation}
      \label{eq:bangA}
      \hfil
      \mbox{either (i)}\quad
      \od{\odd{\odd{\odh{\lneg A \vlte J_A}}{\deric_A}{I_A}{\SLScpU}}{
          \derid_A}{\vlbot}{\LScp}}
    \qquad,\qquad
    \mbox{or (ii)}\quad
    \od{\odd{\odd{\odh{\lneg A \vlte J_A}}{\deric_A}{I_A}{\SLScpU}}{
        \derid_A}{K_A}{\LScp}}
    \end{equation}
    where, as before, it follows that in the first case we have $I_A=\vlbot$.
    Then, in case (i), we let $I=\Ccons{\vlbot}$, and in case (ii) we let  $I=\Ccons{\vlwn I_A}$. The derivations $\deric$ and $\derid$ are shown below:
    $$
    \begin{array}{c@{\;}c@{\;}c@{\hskip6em}c@{\;}c@{\;}c}
      &&\mbox{(i)}&&&  \mbox{(ii)}\\
      \deric&=&
      \odV{\vlwn\lneg A\vlte H}{
        \deri_H}{\CCons{\od{\odi{\odi{\odh{
                  \vlwn \lneg A\vlte \vloc J_A}}{
                \pU}{\vlwn\odv{\lneg A \vlte J_A}{\deric_A}{\vlbot}{\SLScpU}
              }{}}{
              \eU}{\vlbot}{}}}
      }{\SLScpU}
      &
      \deric&=&
      \odV{\vlwn\lneg A\vlte H}{
        \deri_H}{\CCons{\odn{
            \vlwn \lneg A\vlte \vloc J_A}{
            \pU}{\vlwn\odv{\lneg A \vlte J_A}{\deric_A}{I_A}{\SLScpU}
          }{}}
      }{\SLScpU}
      \\
      \derid&=&
      \odV{\Ccons{\vlbot}}{\deri_K}{K}{\LScp}
      &
      \derid&=&
      \odV{\CCons{
          \vlwn
          \odv{I_A}{\derid_A}{K_A}{\LScp}
      }}{\deri_K}{K}{\LScp}  
    \end{array}
    $$
  \item \label{f:wn} $F=\vlwn A$. We apply \Cref{lem:splitting}.\ref{s:wn} to $\deri$, giving us derivations $\deri_J,\deri_K,\deri_A$ (there are two cases). We apply the induction hypothesis to $\deri_A$, giving us $I_A$ and a derivation as in~\eqref{eq:bangA} above,
    where, as before, it follows that in the first case we have $I_A=\vlbot$.
    Then, in case (i), we let $I=\Ccons{\vlbot}$, and in case (ii) we let  $I=\Ccons{\vlda I_A}$, where $\vlda\in\set{\vlwn,\vloc}$ is given by  \Cref{lem:splitting}.\ref{s:wn}. The derivations $\deric$ and $\derid$ are shown below:
    $$
    \begin{array}{c@{\;}c@{\;}c@{\hskip6em}c@{\;}c@{\;}c}
      &&\mbox{(i)}&&&  \mbox{(ii)}\\
      \deric&=&
      \odV{\vloc\lneg A\vlte H}{
        \deri_H}{\CCons{\od{\odi{\odi{\odh{
                  \vloc \lneg A\vlte \vlwn J_A}}{
                \pU}{\vlwn\odv{\lneg A \vlte J_A}{\deric_A}{\vlbot}{\SLScpU}
              }{}}{
              \eU}{\vlbot}{}}}
      }{\SLScpU}
      &
      \deric&=&
      \odV{\vloc\lneg A\vlte H}{
        \deri_H}{\CCons{\odn{
            \vloc \lneg A\vlte \vlda J_A}{
            \phU}{\vlda\odv{\lneg A \vlte J_A}{\deric_A}{I_A}{\SLScpU}
          }{}}
      }{\SLScpU}
      \\
      \derid&=&
      \odV{\Ccons{\vlbot}}{\deri_K}{K}{\LScp}
      &
      \derid&=&
      \odV{\CCons{
          \vlda
          \odv{I_A}{\derid_A}{K_A}{\LScp}
      }}{\deri_K}{K}{\LScp}  
    \end{array}
    $$
    where $\phU$ is either $\pU$ or $\ppU$, depending on $\vlda$.
  \item  \label{f:par} $F=A\vlpa B$. In this case we cannot rely on  \Cref{lem:splitting} as there is no case for $\vlpa$. Instead, we will apply the induction hypothesis twice as follows:
    $$
    \od{\odd{\odd{\odh{H}}{\derib}{
          J}{\SLScpU}}{\deri}{
        A\vlpa B\vlpa K}{\LScp}}
    \qualto
    \od{\odd{\odd{\odh{\lneg A\vlte H}}{\derib'}{
          J'}{\SLScpU}}{\deri'}{
        B\vlpa K}{\LScp}}
    \qualto    
    \od{\odd{\odd{\odh{\lneg A\vlte\lneg B\vlte H}}{\deric}{
          I}{\SLScpU}}{\derid}{
        K}{\LScp}}    
    $$
    where we left instances of $\eqD$ and $\eqU$ implicit.
    \qedhere
  \end{enumerate}
\end{proof}


\lemDecomposition*

\begin{proof}
  This is proved by rule permutations. Assume we have an instance $\smash{\vlinf{\rcD}{}GF}$ of a core rule and and instance $\vlinf{\rncD}{}QP$ of a non-core rule. We show that the two can be permuted by a case analysis of the possible interactions. We begin with the trivial cases where there is no interaction and the two rule instances can ``pass through'' each other:
  \begin{itemize}
  \item If $\rcD$ and $\rncD$ are in independent parts of the derivation, then we have
    $$
    \odn{\CCons[1]{F}\vlva\CCons[2]{\odn P{\rncD}Q{}}}{
      \eqD}{
      \CCons[1]{\odn{F}{\rcD}{G}{}}\vlva\CCons[2]{Q}
    }{}
    \qquor
    \od{\odh{\CCons[1]{\odn{F}{\rcD}{G}{}}\vlva\CCons[2]{\odn P{\rncD}Q{}}}}
    $$
    where $\vlva\in\set{\vlte,\vlpa,\vlwi,\vlpl}$. Both derivations above can be replaced by
    $$
    \odn{\CCons[1]{\odn{F}{\rcD}{G}{}}\vlva\CCons[2]{P}}{
      \eqD}{
      \CCons[1]{G}\vlva\CCons[2]{\odn P{\rncD}Q{}}
    }{}
    $$
  \item If $\rncD$ is inside a passive formula in the premise of $\rcD$, then we are in a situation as shown on the left below, where $\rcD$ is $\sD$, and where $A\conhole$ is an arbitrary context. This can be replaced by the derivation shown on the right below.
    $$
    \odn{(A\Cons{\odn P{\rncD}Q{}}\vlpa C)\vlte(B\vlpa D)}{
      \sD}{
      (A\cons{Q}\vlte B)\vlpa(C\vlpa D)
    }{}
    \qqualto
    \odn{(A\Cons{P}\vlpa C)\vlte(B\vlpa D)}{
      \sD}{
      (A\Cons{\odn P{\rncD}Q{}}\vlte B)\vlpa(C\vlpa D)
    }{}  
    $$
    The situation is similar when $\rncD$ occurs inside $B$, $C$, or $D$. And for $\rcD$ being one of the rules $\pD$, $\ppD$, $\dD$, $\dpD$, and $\eqD$, the cases are similar to $\sD$ above.
  \item If $\rcD$ is inside a passive formula in the conclusion of $\rncD$, the we have three subcases:
    \begin{itemize}
    \item $\rcD$ just passes through, similarly to what happens in the previous case above:
      $$
      \odn{A\cons{F}}{\wlD}{A\Cons{\odn{F}{\rcD}{G}{}}\vlpl B}{}
      \qqualto
      \odn{A\Cons{\odn{F}{\rcD}{G}{}}}{\wlD}{A\Cons{G}\vlpl B}{}
      $$
      The situation is similar for $\rncD$ being $\wrD$ and $\rcD$ inside $B$, or $\rncD$ being $\tD$ or $\gD$.
    \item $\rcD$ is deleted because of weakening:
      $$
      \odn{A}{\wlD}{A\vlpl B\Cons{\odn{F}{\rcD}{G}{}}}{}
      \qqualto
      \odn{A}{\wlD}{A\vlpl B\Cons{G}}{}
      $$
      and similarly for $\rncD$ being $\wrD$ and $\rcD$ inside $A$, or $\rncD$ being $\wD$ or $\wpD$.
    \item $\rcD$ is duplicated because of contraction:
      $$
      \odn{A\cons{F}\vlpl A\cons{F}}{\cpD}{A\Cons{\odn{F}{\rcD}{G}{}}}{}
      \qqualto
      \odn{A\Cons{\odn{F}{\rcD}{G}{}}\vlpl A\Cons{\odn{F}{\rcD}{G}{}}}{
        \cpD}{A\Cons{G}}{}
      $$
      and similarly for $\rncD$ being $\cD$.
    \end{itemize}
  \end{itemize}
  Next, we consider the cases where the conclusion of $\rncD$ and the premise of $\rcD$ interact in a non-trivial way:
  \begin{itemize}
  \item If $\rncD$ is $\wpD$, then we have a situation like the following:
    $$
    \odn{\odn{\vlzer}{\wpD}{A\vlpa C}{}\vlte(B\vlpa D)}{
      \sD}{(A\vlte B)\vlpa(C\vlpa D)}{}
    \qqualto
    \odn{\odn{\vlzer}{\dzpD}{\vlzer\vlpa\vlzer}{}\vlte(B\vlpa D)}{
      \sD}{(\odn{\vlzer}{\wpD}{A}{}\vlte B)\vlpa(\odn{\vlzer}{\wpD}{C}{}\vlpa D)}{}
    $$
    The situation is similar for $\wpD$ acting on the subformula $B\vlpa D$. And similarly for $\rcD\in\set{\pD,\ppD,\dD,\dpD,\eqD}$. The important observation is that such a situation can only occur with a $\vlpa$-subformula, and we can resolve it by using $\dzpD$, which is part of the core-fragment.
  \item If $\rncD$ is $\cpD$, then we have a situation like the following:
    $$
    \odn{\odn{(A\vlpa C)\vlpl(A\vlpa C)}{\cpD}{A\vlpa C}{}\vlte(B\vlpa D)}{
      \sD}{(A\vlte B)\vlpa(C\vlpa D)}{}
    \qqualto
    \odn{\odn{(A\vlpa C)\vlpl(A\vlpa C)}{\dpD}{(A\vlpl A)\vlpa(C\vlpl C)}{}
      \vlte(B\vlpa D)}{
      \sD}{(\odn{A\vlpl A}{\cpD}{A}{}\vlte B)\vlpa(\odn{C\vlpl C}{\cpD}{C}{}\vlpa D)}{}
    $$
    The situation is similar for $\cpD$ acting on the subformula $B\vlpa D$. And similarly for $\rcD\in\set{\pD,\ppD,\dD,\dpD,\eqD}$. Again, the important observation is that such a situation can only occur with a $\vlpa$-subformula, that we can resolve by using $\dpD$, which is part of the core-fragment.
  \item The critical pair can also come from a $\vlpl$ formula: Then $\rcD$ is $\dpD$ and we have a situation as follows:
    $$
    \od{\odi{\odi{\odh{A\vlpa C}}{
          \wlD}{(A\vlpa C)\vlpl(B\vlpa D)}{}}{
        \dpD}{(A\vlpl B)\vlpa(C\vlpl D)}{}}
    \qqualto
    \od{\odh{\odn{A}{\wlD}{A\vlpl B}{}
        \vlpa
        \odn{C}{\wlD}{C\vlpl D}{}}}
    $$
    and similarly for $\rncD$ being $\wrD$.
  \item Finally, the critical pair can come from a $\vlwn$-formula. Then $\rcD$ is $\ppD$ and there are four subcases, one for each $\rncD\in\set{\wD,\rD,\gD,\cD}$:
    $$
    \od{\odi{\odi{\odh{\vlbot}}{
          \wD}{\vlwn(A\vlpa B)}{}}{
        \ppD}{\vlwn A\vlpa \vlwn B}{}}
    \qqualto
    \odn{\vlbot}{
      \eqD}{
      \odn{\vlbot}{\wD}{\vlwn A}{}
      \vlpa
      \odn{\vlbot}{\wD}{\vlwn B}{}
    }{}
    $$
    $$
    \od{\odi{\odi{\odh{A\vlpa B}}{
          \tD}{\vlwn(A\vlpa B)}{}}{
        \ppD}{\vlwn A\vlpa \vlwn B}{}}
    \qqualto
    \od{\odh{\odn{A}{\tD}{\vlwn A}{}
        \vlpa
        \odn{B}{\tD}{\vlwn B}{}}}
    $$
    $$
    \od{\odi{\odi{\odh{\vlwn\vlwn(A\vlpa B)}}{
          \gD}{\vlwn(A\vlpa B)}{}}{
        \ppD}{\vlwn A\vlpa \vlwn B}{}}
    \qqualto
    \odn{\vlwn\left( \odn{\vlwn(A\vlpa B)}{\ppD}{\vlwn A\vlpa \vlwn B}{}\right)}{
      \ppD}{
      \odn{\vlwn\vlwn A}{\gD}{\vlwn A}{}
      \vlpa
      \odn{\vlwn\vlwn B}{\gD}{\vlwn B}{}
    }{}
    $$
    $$
    \od{\odi{\odi{\odh{\vlwn(A\vlpa B)\vlpa\vlwn(A\vlpa B)}}{
          \cD}{\vlwn(A\vlpa B)}{}}{
        \ppD}{\vlwn A\vlpa \vlwn B}{}}
    \qqualto
    \odn{\left( \odn{\vlwn(A\vlpa B)}{\ppD}{\vlwn A\vlpa \vlwn B}{}\right)
      \vlpa
\left( \odn{\vlwn(A\vlpa B)}{\ppD}{\vlwn A\vlpa \vlwn B}{}\right)
    }{
      \eqD}{
      \left( \odn{\vlwn A\vlpa\vlwn A}{\cD}{\vlwn A}{}\right)
      \vlpa
      \left( \odn{\vlwn B\vlpa\vlwn B}{\cD}{\vlwn B}{}\right)
    }{}
    $$
  \end{itemize}
  This completes the proof of the decomposition lemma.
\end{proof}

\section{Additional Material for \Cref{sec:classical}}\label{app:classical}

For the convenience of the reader, we collect here some material that is also available elsewhere (e.g.~\cite{brunnler:phd,brunnler:tiu:01,brunnler:06:locality,ATS:esslli2019}) but might help understanding the paper.

Below is the one-sided sequent calculus $\GSone$
from~\cite{troelstra:schwichtenberg:00} for classical logic.\footnote{We added the rule for  $\vltt$ to include the units. It is easy to see that cut elimination is preserved.}
\begin{equation}
  \label{eq:GS1}
  \hfill
  \begin{array}{c}
    \vlinf{\vltt}{}{\sqn{\vltt}}{}
    \qquad
    \vlinf{\axr}{}{\sqn{a,\cneg a}}{}
    \qquad
    \vlinf{\weakr}{}{\sqn{\Gamma,A}}{\sqn{\Gamma}}
    \qquad
    \vlinf{\conr}{}{\sqn{\Gamma,A}}{\sqn{\Gamma,A,A}}
    \\\\[-1ex]
    \vlinf{\vlor_L}{}{\sqn{\Gamma,A\vlor B}}{\sqn{\Gamma,A}}
    \qquad
    \vlinf{\vlor_R}{}{\sqn{\Gamma,A\vlor B}}{\sqn{\Gamma,B}}
    \qquad
    \vliinf{\vlan}{}{\sqn{\Gamma,A\vlan B}}{\sqn{\Gamma,A}}{\sqn{\Gamma,B}}
  \end{array}
  \hfill
\end{equation}

\begin{theorem}[\cite{troelstra:schwichtenberg:00}]
  The cut rule
  \begin{equation}
    \label{eq:cutSC}
    \hfill
    \vliinf{\cutr}{}{\sqn{\Gamma,\Delta}}{\sqn{\Gamma,A}}{\sqn{\Delta,\cneg A}}
    \hfill
  \end{equation}
  is admissible for $\GSone$.
\end{theorem}

\begin{lemma}\label{lem:SKStoGS1}
  If $\odv{A}{\deri}{B}{\SKS}$ then the sequent $\sqn{\cneg A,B}$ is provable in $\GSone+\cutr$.
\end{lemma}
\begin{proof}
  By proceed by structural induction on $\deri$:
  \begin{itemize}
  \item If $\deri$ is just a formula we have a $\GSone$-proof of $\sqn{\cneg A,A}$ by induction on $A$.
  \item If $\deri$ is $\downsmash{\odv{A_1}{\deri_1}{B_1}{\SKS}\vlor\odv{A_2}{\deri_2}{B_2}{\SKS}}$
    then we construct
    $$
    \vlderivation{
      \vliin{\vlan}{}{\sqn{\cneg A_1\vlan\cneg A_2,B_1\vlor B_2}}{
        \vlin{\vlor_L}{}{\sqn{\cneg A_1,B_1\vlor B_2}}{
          \vlhtr{\pi_1}{\sqn{\cneg A_1,B_1}}}}{
        \vlin{\vlor_R}{}{\sqn{\cneg A_2,B_1\vlor B_2}}{
          \vlhtr{\pi_2}{\sqn{\cneg A_2,B_2}}}}}
    $$
    where $\pi_1$ and $\pi_2$ exist by induction hypothesis.
  \item The case for $\odv{A_1}{\deri_1}{B_1}{\SKS}\vlan\odv{A_2}{\deri_2}{B_2}{\SKS}$ is similar.
  \item if $\deri$ is $\odn{\deri_1}{\rr}{\deri_2}{}$ with $\odv{A}{\deri_1}{B}{\SKS}$ and
    $\odv{C}{\deri_2}{D}{\SKS}$ then we construct
    $$
    \vlderivation{
      \vliin{\cutr}{}{\sqn{\cneg A,D}}{
        \vliin{\cutr}{}{\sqn{\cneg A,C}}{
          \vlhtr{\pi_1}{\sqn{\cneg A,B}}}{
          \vlhtr{\pi_2}{\sqn{\cneg B,C}}}}{
        \vlhtr{\pi_3}{\sqn{\cneg C,D}}}}
    $$
    where $\pi_1$ and $\pi_3$ exist by induction hypothesis,
    and $\pi_2$ exists because for every inference rule $\vlinf{\rr}{}{C}{B}$ in $\SKS$, the implication $B\vlim C$ is a tautology and therefore we have a proof of $\sqn{\cneg B,C}$ in $\GSone$.
    \qedhere
  \end{itemize}
\end{proof}

\begin{lemma}\label{lem:GS1toKS}
  If a sequent $\sqn{C_1,C_2,\ldots,C_n}$ is provable in $\GSone$ then we have a derivation
  $$
  \odv{\vltt}{\deri}{C_1\vlor C_2\vlor \cdots\vlor C_n}{\KS}
  $$
\end{lemma}
\begin{remark}
  We keep the bracketing implicit, as any choice of bracketing of the conclusion will be provable in $\KS$, due to the rule $\eqD$.
\end{remark}
\begin{proof}
  We proceed by induction on the sequent proof, considering the bottom-most rule instance. 
  The base cases $\vltt$ and $\axr$ are trivial, the other cases are as follows:
  $$
  \weakr:
  \od{
    \odi{\odd{\odh{\vltt}}{
        \deri'}{C_1\vlor \cdots\vlor C_{n-1}}{\KS}}{
      \eqD}{C_1\vlor \cdots\vlor C_{n-1}\vlor\odn{\vlff}{\wD}{A}{}}{}}
  \qquad
  \conr:
  \odv{\vltt}{\deri'}{C_1\vlor \cdots\vlor C_{n-1}\vlor\odn{A\vlor A}{\cD}{A}{}}{\KS}
  $$
  $$
  \vlor_L:
  \od{
     \odi{\odd{\odh{\vltt}}{
         \deri'}{C_1\vlor \cdots\vlor C_{n-1}\vlor A}{\KS}}{
       \eqD}{C_1\vlor \cdots\vlor C_{n-1}\vlor A\vlor\odn{\vlff}{\wD}{B}{}}{}}
  \qquad
  \vlor_R:
  \od{
     \odi{\odd{\odh{\vltt}}{
         \deri'}{C_1\vlor \cdots\vlor C_{n-1}\vlor B}{\KS}}{
       \eqD}{C_1\vlor \cdots\vlor C_{n-1}\vlor \odn{\vlff}{\wD}{A}{}\vlor B}{}}
  $$
  $$
  \vlan:
  \od{
    \odi{\odi{\odh{\vltt}}{
        \andD}{
        \odv{\vltt}{\deri'}{C_1\vlor \cdots\vlor C_{n-1}\vlor A}{\KS}
        \vlan
        \odv{\vltt}{\deri''}{C_1\vlor \cdots\vlor C_{n-1}\vlor B}{\KS}}{}}{
      \sD}{
      \odn{(C_1\vlor \cdots\vlor C_{n-1})\vlor(C_1\vlor \cdots\vlor C_{n-1})}{
        \cD}{C_1\vlor \cdots\vlor C_{n-1}}{}
      \vlor(A\vlan B)}{}}
  $$
  where $\deri'$ and $\deri''$ exist by induction hypothesis.
\end{proof}

These two lemmas are enough to prove \Cref{thm:soundnessKS} and \Cref{thm:cutelimSKS}.

\medskip

We also have the general \defn{axiom} and \defn{cut} rules, respectively:
\begin{equation}\label{eq:cutKS}
  \hfill
  \scalebox{.9}{$
    \vlinf{\iD}{}{\cneg A \vlor A}{\vltt}
    \qquand
    \vlinf{\iU}{}{\vlff}{A \vlan \cneg A}
    $}
  \hfill
\end{equation}

together with the classical versions of \Cref{prop:DI} and \Cref{lem:switch}:

\begin{restatable}{proposition}{propAIKS}\label{prop:DI-KS}
  We have the following:
  \begin{enumerate}
  \item\label{pr:ai}
    The rule $\iD$ is derivable in $\KSc$, and dually, the rule $\iU$ is derivable in $\SKScU$.
  \item \label{pr:updown}
    Every rule $\rU\in\SKSU$ is derivable in $\set{\rD,\iD,\sD,\iU,\eqD,\eqU}$.
  \item\label{pr:flip1}
    For every $\sys\in\set{\SKSc,\SKS}$, we have $A\proves[\sys]B$ iff\/ $\proves[\sys]\cneg A\vlor B$. 
  \item\label{pr:flip2}
    For
    every $\sys\in\set{\SKSc,\SKS,\SKSnc}$, we have $A\proves[\sysD]B$ iff\/ $\cneg B\proves[\sysU]\cneg A$.
  \item\label{pr:cutup}
  For every $\sys\in\set{\KSc,\KS}$, we have $\proves[\sysSX]A$ iff\/ $\proves[\sys\cup\set{\iU}]A$.\\[-2ex]
  \end{enumerate}
\end{restatable}

\begin{restatable}{lemma}{lemSwitchKS}\label{lem:switchKS}
  For all formulas $A,B$ and m-contexts $\Chole$, there are derivations
  $$
  \odv{\Chole[A\vlor B]}{\deri}{A\vlor\Chole[B]}{\set{\eqD,\sD}}
  \qquand
  \odv{A\vlan\Chole[B]}{\derib}{\Chole[A\vlan B]}{\set{\eqU,\sU}}
  $$
\end{restatable}

Their proofs are exactly the same as in \Cref{app:prelim}.

\section{Additional Material for \Cref{sec:modal}}\label{app:modal}

In this appendix we extend the material from the previous appendix to
the modal logics $\K,\KT,\Kfour,\Sfour$. We first need the following sequent rules:
\begin{equation}
  \label{eq:modalSC}
  \hfill
  \vlinf{\kr}{}{\sqn{\vldia\Gamma,\vlbox A}}{\sqn{\Gamma,A}}
  \qquad\qquad
  \vlinf{\ktr}{}{\sqn{\Gamma,\vldia A}}{\sqn{\Gamma,A}}
  \qquad\qquad
  \vlinf{\kfr}{}{\sqn{\vldia\Gamma,\vlbox A}}{\sqn{\vldia\Gamma,A}}
  \hfill
\end{equation}
where $\vldia\Gamma$ abbreviates $\vldia C_1,\ldots,\vldia C_n$ if $\Gamma=C_1,\ldots,C_n$.
With this we define the following proof systems:
$$
\GSk=\GSone\cup\set\kr
\qquad
\GSkt=\GSone\cup\set{\kr,\ktr}
\qquad
\GSkf=\GSone\cup\set{\kr,\kfr}
\qquad
\GSsf=\GSone\cup\set{\ktr,\kfr}
$$
Note that $\kr$ is derivable in $\set{\ktr,\kfr}$, and therefore is not needed in $\GSsf$.

\begin{theorem}
  For every $\sysY\in\set{\K,\KT,\Kfour,\Sfour}$, the sequent system $\GSy$ is sound and complete for the modal logic $\sysY$, and the cut rule~\eqref{eq:cutSC} is admissible for $\GSy$.
\end{theorem}

For a proof we refer the reader to the literature (e.g.~\cite{troelstra:schwichtenberg:00,Wansing:02,poggiolesi:11}).

\begin{lemma}\label{lem:SKSytoGS1y}
  For every $\sysY\in\set{\K,\KT,\Kfour,\Sfour}$, if $\odv{A}{\deri}{B}{\SKSy}$ then the sequent $\sqn{\cneg A,B}$ is provable in $\GSy+\cutr$.
\end{lemma}
\begin{proof}
  The proof is essentially the same as for \Cref{lem:SKStoGS1}, observing that indeed for every inference rule $\vlinf{\rr}{}{C}{B}$ in $\SKSy$ we have a proof of $\sqn{\cneg B,C}$ in $\GSy$. The missing cases for $\deri$ are:\vadjust{\vskip-2ex}
  \begin{itemize}
  \item If $\deri$ is 
    $\nosmash{\vlbox\left(\odv{A_1}{\deri_1}{B_1}{\SKS}\right)}$
    then we have
    \raise3ex\hbox{$
    \vlderivation{
      \vlin{\kr}{}{\sqn{\vldia\cneg A_1,\vlbox B_1}}{
        \vlhtr{\pi_1}{\sqn{\cneg A_1,B_1}}}}
    $}
    where $\pi_1$ exist by induction hypothesis.
  \item  If $\deri$ is 
    $\nosmash{\vldia\left(\odv{A_1}{\deri_1}{B_1}{\SKS}\right)}$ then we proceed similarly.
    \qedhere
  \end{itemize}
\end{proof}

\begin{lemma}\label{lem:GS1ytoKSy}
  For every $\sysY\in\set{\K,\KT,\Kfour,\Sfour}$, if a sequent $\sqn{C_1,C_2,\ldots,C_n}$ is provable in $\GSy$ then we have a derivation
  $$
  \odv{\vltt}{\deri}{C_1\vlor C_2\vlor \cdots\vlor C_n}{\KSy\setminus\set{\ppD}}
  $$
\end{lemma}
\begin{proof}
  We proceed as in the proof of \Cref{lem:GS1toKS}, listing here only the additional cases for the rules in~\eqref{eq:modalSC}:
  $$
  \kr:
  \od{
    \odi{\odi{\odh{\vltt}}{
        \eD}{\vlbox\odv{\vltt}{\deri'}{C_1\vlor C_2\vlor\cdots\vlor C_{n-1}\vlor C_n\vlor A}{\KSk\setminus\set{\ppD}}}{}}{
      \pD}{
      \vldia C_1\vlor
      \odv{\vlbox(C_2\vlor\cdots \vlor C_{n-1}\vlor C_n\vlor A)}{}{
        \vldia C_2\vlor\cdots\vlor\vldia C_{n-1}\vlor 
        \odn{\vlbox(C_n\vlor A)}{\pD}{\vldia C_n\vlor \vlbox A}{}
      }{\pD}}{}
    }
  $$
  $$
  \ktr:
  \odv{\vltt}{\deri'}{C_1\vlor \cdots\vlor C_{n-1}\vlor\odn{A}{\tD}{\vldia A}{}}{\KSk\setminus\set{\ppD}}
  $$
  $$
  \kfr:
  \od{
    \odi{\odi{\odh{\vltt}}{
        \eD}{\vlbox\odv{\vltt}{\deri'}{
          \vldia C_1\vlor\vldia C_2\vlor\cdots\vlor\vldia C_{n-1}\vlor\vldia C_n\vlor A}{\KSk\setminus\set{\ppD}}}{}}{
      \pD}{
      \odn{\vldia\vldia C_1}{\gD}{\vldia C_1}{}\vlor
      \odv{\vlbox(\vldia C_2\vlor\cdots \vlor \vldia C_{n-1}\vlor \vldia C_n\vlor A)}{}{
        \odn{\vldia\vldia C_2}{\gD}{\vldia C_2}{}
        \vlor\cdots\vlor
        \odn{\vldia\vldia C_{n-1}}{\gD}{\vldia C_{n-1}}{}
        \vlor 
        \odn{\vlbox(\vldia C_n\vlor A)}{\pD}{
          \odn{\vldia\vldia C_{n}}{\gD}{\vldia C_{n}}{}
          \vlor \vlbox A}{}
      }{\pD}}{}
    }
  $$
  where in each case, the derivation $\deri'$ exists by induction hypothesis.
\end{proof}

\thmSoundKSk*

\begin{proof}
  This now follows immediately from \Cref{lem:SKSytoGS1y} and \Cref{lem:GS1ytoKSy}.
\end{proof}




\end{document}